\newif\ifpreprint
\preprinttrue

\ifpreprint
  \documentclass[sigconf]{acmart}
\else
  \documentclass[sigconf, anonymous]{acmart}
\fi

\AtBeginDocument{%
  }

\usepackage{float}    
\usepackage{placeins} 

\usepackage{algorithm}
\usepackage{algpseudocode}
\usepackage{cleveref}
\crefname{appendix}{Appendix}{Appendices}
\Crefname{appendix}{Appendix}{Appendices}
\usepackage[most]{tcolorbox}
\usepackage{enumitem}

\usepackage{standalone}
\usepackage{pgfplots}
\pgfplotsset{compat=1.18}
\usepgfplotslibrary{fillbetween}
\usepgfplotslibrary{groupplots}
\usetikzlibrary{patterns,fit,backgrounds,arrows.meta}
\usepackage{booktabs}
\usepackage{wasysym}  
\usepackage{multirow}

\tcbset{
  protocolbox/.style={
    colframe=black,
    colback=white,
    fonttitle=\bfseries,
    title={#1},
    boxrule=1pt,
    arc=1mm,
    coltitle=black,
    colbacktitle=gray!30,    
    left=1mm,
    right=1mm,
    top=1mm,
    bottom=1mm,
  }
}

\usepackage{thmtools}
\usepackage{thm-restate}   
\declaretheorem[name=Lemma]{lemma}

\declaretheorem[name=Corollary]{corollary}

\definecolor{cTrusted}{HTML}{1F77B4}      
\definecolor{cUntrusted}{HTML}{FF7F0E}    
\definecolor{cComm}{HTML}{8E8E8E}         
\definecolor{cMatMul}{HTML}{2CA02C}       
\definecolor{cMaskWeight}{HTML}{6FA8DC}   
\definecolor{cMaskAct}{HTML}{C9DAF8}      
\definecolor{cNoise}{HTML}{D62728}        
\definecolor{cBaseline}{HTML}{6AA84F}     

\tikzset{
    trustedfill/.style   ={fill=cTrusted!70, draw=cTrusted!90!black},
    untrustedfill/.style ={fill=cUntrusted!70, draw=cUntrusted!90!black},
    commfill/.style      ={fill=cComm!40, draw=cComm!70!black,
                            postaction={pattern=north east lines,
                                        pattern color=cComm!70!black}},
}
\renewcommand{\vec}{}
\newcommand{\mat}{}
\newcommand{\reals}{\mathbb{R}}

\newcommand{\uniform}[1]{\mathcal{U}\left(#1\right)}
\newcommand{\sparse}[1]{\mathcal{S}\left(#1\right)}
\newcommand{\expect}[1]{\mathbb{E}\left[#1\right]}
\newcommand{\norm}[1]{\left\|#1\right\|}
\newcommand{\normal}[2]{\mathcal{N}\left(#1, #2\right)}
\renewcommand{\exp}[1]{\text{exp}\left(#1\right)}

\newtheorem{example}{Example}[section]

\newcommand{\cli}{\mathcal{C}}
\newcommand{\gpu}{\mathcal{G}}
\newcommand{\fpVec}[1]{\mathbb{F}_{p}^{#1}}

\newcommand{\rngVec}[1]{\mathbb{Z}_{2^{\kappa}}^{#1}}
\newcommand{\rngMat}[2]{\mathbb{Z}_{2^{\kappa}}^{#1 \times #2}}
\newcommand{\rng}{\mathbb{Z}_{2^{\kappa}}}

\newif\ifdraft
\drafttrue
\newcommand{\james}[1]{\ifdraft\textcolor{blue}{[\textbf{JC}: #1]}\fi}

\newif\ifrevisionmarks
\revisionmarkstrue
\definecolor{RevisionColor}{HTML}{0B5FA5}
\newcommand{\revised}[1]{\ifrevisionmarks\textcolor{RevisionColor}{#1}\else#1\fi}
\newenvironment{revisedblock}%
  {\par\ifrevisionmarks\color{RevisionColor}\fi\ignorespaces}%
  {\par}

\newcommand{\mosaicrepourl}{%
  \ifpreprint
    \url{https://github.com/jachiang/mosaic}%
  \else
    \url{https://anonymous.4open.science/r/mosaic-71FE}%
  \fi}

\draftfalse

\ifpreprint\revisionmarksfalse\else\revisionmarkstrue\fi

\setcopyright{acmlicensed} 
\copyrightyear{2018} 
\acmYear{2018} 
\acmDOI{XXXXXXX.XXXXXXX} 
\acmConference[Conference acronym 'XX]{Make sure to enter the correct
  conference title from your rights confirmation email}{June 03--05,
  2018}{Woodstock, NY}  
\acmISBN{978-1-4503-XXXX-X/2018/06}  

\newif\ifshowacmmeta
\ifpreprint\showacmmetafalse\else\showacmmetatrue\fi

\ifshowacmmeta\else
  \renewcommand{\footnotetextcopyrightpermission}[1]{}
\fi

\newif\ifshowconfheader
\ifpreprint\showconfheaderfalse\else\showconfheadertrue\fi

\begin{document}


\title{MOSAIC: Masked Outsourcing of Secure AI Computations}



\author{James Hsin-yu Chiang}
\affiliation{%
  \institution{ETH Zurich}
  \country{Switzerland}
}

\author{Sheila Zingg}
\affiliation{%
  \institution{ETH Zurich}
  \country{Switzerland}
}

\author{Kari Kostiainen}
\affiliation{%
  \institution{ETH Zurich}
  \country{Switzerland}
}

\author{Srdjan Capkun}
\affiliation{%
  \institution{ETH Zurich}
  \country{Switzerland}
}







\renewcommand{\shortauthors}{Chiang et al.}

\begin{abstract}
\begin{revisedblock}
We address the challenge of securely and efficiently outsourcing AI computations from a trusted but computationally weak client to an untrusted but powerful server, in the setting where the client holds both the input and the model, and the server must learn neither. We present MOSAIC, whose core is a novel matrix-multiplication masking protocol that scales to far larger matrices than prior work, enabling the safe outsourcing of modern workloads such as large transformer inference. By introducing small amounts of noise to the multiplication result and thereby relaxing correctness, MOSAIC achieves optimal asymptotic client overhead and concrete runtimes orders of magnitude faster than prior work. Its security reduces to the decisional LWE and LPN assumptions.

Because this noise accumulates across the many
layers of a transformer, a key technical challenge is bounding error growth; MOSAIC addresses this with an error-scaling mechanism based on random Hadamard rotations. On large 70B transformer models, MOSAIC's perplexity is comparable to popular quantization approaches and 
even matches full-precision BF16 inference on HumanEval.

Finally, we present an end-to-end implementation\footnote{Implementation available at
\mosaicrepourl}
showing how ideas like MOSAIC can promise a path towards large-scale confidential AI in modern data centers. Non-confidential inference is already distributed across phase (prefill/decode), layer, and time to maximize utilization of heterogeneous hardware, using RDMA-like networking to move activations, cached KV values, and weights across nodes. MOSAIC enables scaling of confidential compute by keeping 
the trusted computing base (TCB) small and outsourcing the bulk 
of the AI computation to untrusted accelerators.




\end{revisedblock}
\end{abstract}


\ifshowacmmeta
\fi


\maketitle


\ifshowconfheader\else
  \makeatletter
  \fancyhead[LE]{\ACM@linecountL}%
  \fancyhead[RO]{\ACM@linecountR}%
  \makeatother
\fi

\section{Introduction}
We investigate the problem of outsourcing AI computations,
which are dominated by linear matrix multiplications over billions
of model weights.
The large dimensions of model weight matrices have given rise
to highly parallelized GPU accelerators to speed up computations over billions of model weight values. 


Our contributions in this work enable \emph{practical} 
outsourcing of \emph{large} AI computations to \emph{untrusted} accelerators, which are not permitted to learn anything about the AI computation other than
approximate topology of the model architecture. 
Prior state-of-the-art solutions do not scale 
to larger model dimensions \cite{benhamouda2025encrypted,braverman2025practical}, or offer concretely impractical runtimes~\cite{chen2020maliciously,mann2023towards,li2024nimbus,gupta2023sigma,lu2023bumblebee,jiang2018secure,gao2024secure,zimerman2024power}
or no formal security~\cite{li2024translinkguard,xu2024tempo,shen2022soter,wang2025game}. 
\revised{MOSAIC promises a practical pathway towards confidential AI 
computations in modern AI data centers, where a 
standard (non-confidential) inference request cannot 
be practically served by
a single GPU node and is distributed across a heterogeneous
accelerator pool, with each node contributing to different 
model layers (attention vs. MLP) 
and computation phase (prefill vs. decode)~\cite{wu2026dualpath,spheron2026gpu,sun2026multi}.
Here, MOSAIC permits a small, user-dedicated TCB 
(e.g. confidential compute) 
to be scaled by outsourcing the bulk of the 
AI computation to untrusted accelerators.
Alternatively, extending the TCB would imply
secure attestation of the entire accelerator pool and
networking stack.} 

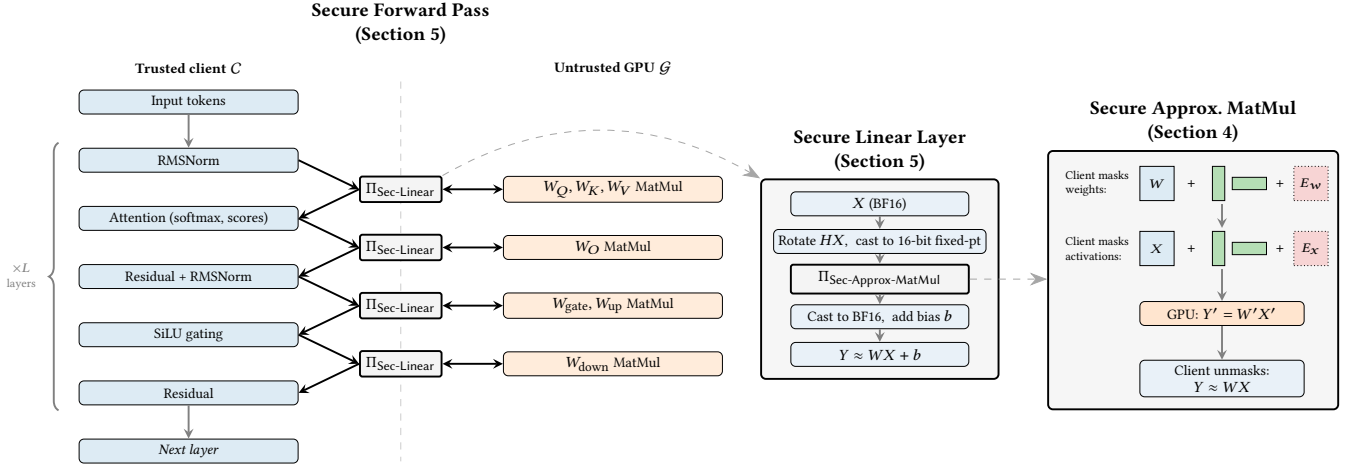
\begin{figure*}[t]
\centering
\resizebox{\textwidth}{!}{\begin{tikzpicture}[
    >=stealth,
    box/.style={draw, rounded corners=2pt, minimum width=3.2cm,
        minimum height=0.35cm, font=\scriptsize, align=center, inner sep=2pt},
    cbox/.style={box, fill=cTrusted!15},
    gbox/.style={box, fill=cUntrusted!15},
    proto/.style={draw, thick, rounded corners=1pt, fill=cComm!10,
        font=\scriptsize, align=center, inner sep=2.5pt},
    arrow/.style={->, thick},
    biarrow/.style={<->, thick},
]

\def\C{-3.1}
\def\G{3.1}
\def\S{0.85}

\node[font=\small\bfseries, align=center] at (0, 0.85)
    {Secure Forward Pass\\(\Cref{sec:secure-forward-pass})};

\node[font=\scriptsize\bfseries] at (\C, 0.2) {Trusted client $\cli$};
\node[font=\scriptsize\bfseries] at (\G, 0.2) {Untrusted GPU $\gpu$};

\draw[gray!40, dashed] (0, 0.0) -- (0, -0.3-6*\S-0.15);

\node[cbox] (input) at (\C, -0.3) {Input tokens};
\node[cbox] (norm1) at (\C, -0.3-\S) {RMSNorm};
\node[cbox] (attn)  at (\C, -0.3-2*\S) {Attention (softmax, scores)};
\node[cbox] (norm2) at (\C, -0.3-3*\S) {Residual + RMSNorm};
\node[cbox] (silu)  at (\C, -0.3-4*\S) {SiLU gating};
\node[cbox] (res)   at (\C, -0.3-5*\S) {Residual};
\node[cbox] (out)   at (\C, -0.3-6*\S) {\textit{Next layer}};

\draw[arrow, gray] (input) -- (norm1);
\draw[arrow, gray] (res) -- (out);

\node[gbox] (qkv)   at (\G, -0.3-1.5*\S) {$W_Q, W_K, W_V$ MatMul};
\node[gbox] (wo)     at (\G, -0.3-2.5*\S) {$W_O$ MatMul};
\node[gbox] (mlpup)  at (\G, -0.3-3.5*\S) {$W_\text{gate}, W_\text{up}$ MatMul};
\node[gbox] (mlpdn)  at (\G, -0.3-4.5*\S) {$W_\text{down}$ MatMul};

\node[proto] (p1) at (0, -0.3-1.5*\S) {$\Pi_{\textsf{Sec-Linear}}$};
\node[proto] (p2) at (0, -0.3-2.5*\S) {$\Pi_{\textsf{Sec-Linear}}$};
\node[proto] (p3) at (0, -0.3-3.5*\S) {$\Pi_{\textsf{Sec-Linear}}$};
\node[proto] (p4) at (0, -0.3-4.5*\S) {$\Pi_{\textsf{Sec-Linear}}$};

\draw[biarrow] (p1.east) -- (qkv.west);
\draw[biarrow] (p2.east) -- (wo.west);
\draw[biarrow] (p3.east) -- (mlpup.west);
\draw[biarrow] (p4.east) -- (mlpdn.west);

\draw[arrow] (norm1.east) -- (p1.west);
\draw[arrow] (p1.west) -- (attn.east);
\draw[arrow] (attn.east) -- (p2.west);
\draw[arrow] (p2.west) -- (norm2.east);
\draw[arrow] (norm2.east) -- (p3.west);
\draw[arrow] (p3.west) -- (silu.east);
\draw[arrow] (silu.east) -- (p4.west);
\draw[arrow] (p4.west) -- (res.east);

\draw[decorate, decoration={brace, amplitude=4pt, mirror},
    thick, gray] (-5.0, -0.3-0.7*\S) -- (-5.0, -0.3-5.3*\S)
    node[midway, left=6pt, font=\tiny, gray, align=center] {$\times L$\\layers};

\def\BX{7.0}  
\def\BY{-0.3-1.75*\S}  
\def\BS{0.55} 

\node[cbox, fill=cTrusted!10, minimum width=2.6cm] (b1) at (\BX, \BY)
    {$\mat{X}$ (BF16)};
\node[cbox, fill=cTrusted!10, minimum width=2.6cm] (b2) at (\BX, \BY-\BS)
    {Rotate $\mat{H}\mat{X}$,\; cast to 16-bit fixed-pt};
\node[proto, minimum width=2.6cm] (b3) at (\BX, \BY-2*\BS)
    {$\Pi_{\textsf{Sec-Approx-MatMul}}$};
\node[cbox, fill=cTrusted!10, minimum width=2.6cm] (b4) at (\BX, \BY-3*\BS)
    {Cast to BF16,\; add bias $\vec{b}$};
\node[cbox, fill=cTrusted!10, minimum width=2.6cm] (b5) at (\BX, \BY-4*\BS)
    {$\mat{Y} \approx \mat{W}\mat{X} + \vec{b}$};

\begin{scope}[on background layer]
\node[draw, thick, rounded corners=2pt, fill=cComm!8,
      fit=(b1)(b2)(b3)(b4)(b5), inner sep=5pt,
      label={[font=\small\bfseries, align=center]above:%
          Secure Linear Layer\\(\Cref{sec:secure-forward-pass})}]
      (blowup) {};
\end{scope}

\draw[arrow, gray] (b1) -- (b2);
\draw[arrow, gray] (b2) -- (b3);
\draw[arrow, gray] (b3) -- (b4);
\draw[arrow, gray] (b4) -- (b5);

\draw[dashed, gray!70, -{Stealth[scale=1.5]}] (p1.north east)
    .. controls +(30:1.8cm) and +(150:1.8cm) ..
    (blowup.north west);

\def\AX{12.0}
\def\AY{-0.3-1.4*\S}
\def\AS{0.95}

\tikzset{
  mglyph/.style={draw, font=\tiny, inner sep=0pt, align=center},
}

\begin{scope}[shift={(\AX, \AY)}]
  \node[font=\tiny, align=left, anchor=west] (rw0) at (-2.4, 0) {Client masks\\weights:};
  \node[mglyph, minimum width=0.50cm, minimum height=0.50cm,
        fill=cTrusted!15] (rW) at (-0.95, 0) {$W$};
  \node[font=\tiny] at (-0.45, 0) {$+$};
  \node[mglyph, minimum width=0.18cm, minimum height=0.50cm,
        fill=cMatMul!35] (rL) at (-0.05, 0) {};
  \node[mglyph, minimum width=0.50cm, minimum height=0.18cm,
        fill=cMatMul!35] (rM) at (0.40, 0) {};
  \node[font=\tiny] at (0.85, 0) {$+$};
  \node[mglyph, densely dotted, minimum width=0.50cm, minimum height=0.50cm,
        fill=cNoise!20] (rE) at (1.30, 0) {$E_w$};
\end{scope}

\begin{scope}[shift={(\AX, \AY-\AS)}]
  \node[font=\tiny, align=left, anchor=west] (rx0) at (-2.4, 0) {Client masks\\activations:};
  \node[mglyph, minimum width=0.50cm, minimum height=0.50cm,
        fill=cTrusted!15] (xX) at (-0.95, 0) {$X$};
  \node[font=\tiny] at (-0.45, 0) {$+$};
  \node[mglyph, minimum width=0.18cm, minimum height=0.50cm,
        fill=cMatMul!35] (xN) at (-0.05, 0) {};
  \node[mglyph, minimum width=0.50cm, minimum height=0.18cm,
        fill=cMatMul!35] (xR) at (0.40, 0) {};
  \node[font=\tiny] at (0.85, 0) {$+$};
  \node[mglyph, densely dotted, minimum width=0.50cm, minimum height=0.50cm,
        fill=cNoise!20] (xE) at (1.30, 0) {$E_x$};
\end{scope}

\node[gbox, minimum width=2.4cm, fill=cUntrusted!20] (cgpu)
    at (\AX, \AY-2*\AS) {GPU: $Y' = W'X'$};
\node[cbox, minimum width=2.4cm, fill=cTrusted!10] (cclient)
    at (\AX, \AY-3*\AS)
    {Client unmasks:\\$Y \approx WX$};

\draw[arrow, gray] (cgpu) -- (cclient);
\draw[arrow, gray] (\AX, \AY-0.30) -- (\AX, \AY-\AS+0.30);
\draw[arrow, gray] (\AX, \AY-\AS-0.30) -- (cgpu.north);

\begin{scope}[on background layer]
\node[draw, thick, rounded corners=2pt, fill=cComm!8,
      fit=(rw0)(rW)(rE)(rx0)(xX)(xE)(cgpu)(cclient),
      inner xsep=4pt, inner ysep=5pt,
      label={[font=\small\bfseries, align=center]above:%
          Secure Approx. MatMul\\(\Cref{sec:lwe-matmul})}]
      (matmulblowup) {};
\end{scope}

\draw[dashed, gray!70, -{Stealth[scale=1.5]}] (b3.east) -- (matmulblowup.west |- b3.east);

\end{tikzpicture}}
\caption{Secure, confidential forward pass for a single transformer layer.
         The trusted client~$\cli$ (left) performs all non-linear operations
         (attention scores, activations, normalization) locally.
         All weight matrix multiplications are outsourced to the
         untrusted GPU~$\gpu$ (right) via $\Pi_{\textsf{Sec-Linear}}$,
         which exchanges only masked activations and results at runtime.
         }
\label{fig:intro-outsourcing}
\end{figure*}

\begin{table*}
  \centering
  \resizebox{\textwidth}{!}{%
  \begin{tabular}{lcccccccc}
    \hline
    & Secure      & Private    & Client layer          & Practical    & \multicolumn{2}{c}{Security}             & Computational                               & \\
    \cline{6-7}
    & outsourcing & inference  & complexity            & runtimes     & Input      & Model                       & error                                       & Perplexity (70B) \\
    \hline
    FHE \cite{jiang2018secure,gao2024secure,zimerman2024power}            & \checkmark & \checkmark & $O((m+n)l)$           & \Circle      & \checkmark & \checkmark                  & fixed-point                                 & \multirow{5}{*}{1.827$^{4}$} \\
    MPC/FSS~\cite{chen2020maliciously,mann2023towards,li2024nimbus,gupta2023sigma,lu2023bumblebee}        & $\times$   & \checkmark & $O(mnl)$           & \Circle      & \checkmark & \checkmark                  & fixed-point                                 & \\
    Slalom \cite{tramer2018slalom} & \checkmark & $\times$   & $O((m+n)l)$           & \LEFTcircle  & \checkmark & $\times$                    & fixed-point                                 & \\
    Obfuscation \cite{li2024translinkguard,xu2024tempo,shen2022soter,wang2025game}    & \checkmark & $\times$   & $O((m+n)l)$           & \CIRCLE      & \multicolumn{2}{c}{heuristic / broken$^{5}$} & fixed-point                                 & \\
    Trapdoor Matrices \cite{braverman2025practical}      & \checkmark & $\times$   & $O(m n^{\epsilon} l)$ & \LEFTcircle  & \checkmark & \checkmark                  & fixed-point                                 & \\
    \textbf{MOSAIC (our work)} & \checkmark & $\times$ & $O(\revised{r}(m+n)l)^{1}$ & \CIRCLE & \checkmark & \checkmark & fixed-point $+\; e^{2}$ & \textbf{1.834$^{3}$} \\
    \hline
  \end{tabular}%
  }
  \caption{Comparison against related work on confidential outsourcing
  of linear computation. Client complexity is per multiplication
  of model weights $(m \times n)$ and activation $(n \times l)$.
  Perplexity is reported for LLaMA-3-70B 
  on WikiText-2.
  \textsuperscript{1}\revised{Parameter $r$ is fixed by the chosen security level and independent of model dimensions $m,n$ and batch size $l$.}
  \textsuperscript{2}$e$ denotes a low-norm Gaussian induced by our protocol.
  \textsuperscript{3}Our protocol perplexity is considerably lower (better) than standard NF4
  4-bit (2.247) quantization baselines and  very close to the BF16 reference (1.821).
  \textsuperscript{4}Idealized perplexity: prior-work perplexity entries assume our proposed random-rotation preconditioning
  that those works do not specify but which they could in principle adopt.
  \textsuperscript{5}Obfuscation-based security is heuristic; we break security of
  ArrowCloak~\cite{wang2025game} (USENIX'25) 
  in~\Cref{sec:obfuscation-security}. 
  }
  \label{tab:intro-comparison}
\end{table*}

\paragraph{Layer-by-layer outsourcing}
We consider a path towards practical 
confidential AI computation that considers the
total computational cost of a \emph{forward-pass} across a transformer-style
model, and focuses on outsourcing the dominating, 
\emph{linear} computations
at each model layer to an untrusted accelerator $\gpu$. 
%
In this approach, the 
\emph{non-linear} computation is performed by the trusted client $\cli$, inducing a slalom-like~\cite{tramer2018slalom} execution between a trusted client $\cli$ and untrusted server (accelerator) $\gpu$, 
as the forward-pass computation of a 
transformer-style model over hundreds of 
alternating \emph{linear} and \emph{non-linear} layers,
with expensive linear layers outsourced to computationally powerful $\gpu$ (as shown in \Cref{fig:intro-outsourcing}).

Linear computation consists of
matrix multiplications of
model weights $W$ ($m \times n$) and
activations $X$ ($n \times l$).
Given that repeated linear matrix arithmetic dominates 
computation during each forward-pass, 
there has been recent effort
to develop cryptographic protocols to 
securely outsource matrix-matrix multiplications~\cite{benhamouda2025encrypted,braverman2025practical}.
However, such solutions have significant limitations that prevent their deployment in practice.
\begin{enumerate}[leftmargin=*]
    \item State-of-the-art in secure outsourcing of
    matrix multiplication of $W$ ($m \times n$) and $X$ ($n \times l$) induces $O(mn^{\epsilon}l)$ overhead for the client~\cite{braverman2025practical}, where $\epsilon$
    must be parameterized for given $n$ to satisfy LPN security.
    This prevents its application to larger modern models,
    where inner dimensions are very large (e.g. an MLP projection of dim 29\,568 in Qwen2.5-72B).
    An optimal protocol achieves $O((m+n)l)$ trusted client complexity,
    and practical concrete runtimes.
    \item \revised{Cryptographic proposals generally operate in the discrete, integer ring domain. This requires fixed-point translation to emulate computation that is natively floating-point. The limited dynamic range of standard fixed-point arithmetic has a catastrophic effect on large, modern LLM architecture (\Cref{fig:layer-cosine-sim-qwen72b-nohad}).
    Because no prior work on cryptographic outsourcing of AI computation
    scales to modern (70B) models,
    this failure mode has consequently gone untreated.}
    \item Prior work does not investigate practical applications amenable to AI outsourcing that relies on layer-by-layer communication. 
\end{enumerate}
Our contributions address these 
open challenges as follows.


\paragraph{1) Inner protocol: Outsourcing of approximate MatMul} 
We propose a cryptographic protocol which permits the trusted client to efficiently outsource large matrix multiplications to an untrusted 
accelerator $\gpu$. Our simplified protocol (described in \Cref{sec:lwe-matmul-core} and illustrated in~\Cref{fig:intro-outsourcing}) masks both model weight and input matrices by padding with
\emph{low-rank} uniform matrices and a \emph{full-rank}, 
small-norm gaussian matrix. 
The security of our full protocol (\Cref{sec:lwe-lpn-matmul})
can be reduced to LWE
and LPN decisional hardness assumptions. 
Crucially, the matrix multiplication of the encrypted model weights and encrypted inputs can be recovered by the client with a small approximation error, which only induces a minor compromise in model accuracy in practice (\Cref{sec:model-accuracy}).

The underlying LWE security and noisy recovery
is critical in the scalability of this approach. The low-rank of the uniform padding matrices remains constant across all matrix dimensions
polynomial in the security parameter $\lambda$. 
For an inner rank $r$, fixed solely by the security parameter and gaussian noise norm, 
the optimal client overhead is thus $O(r(m+n)l) \approx O((m+n)l)$. \revised{We report runtimes showing 
orders of magnitude runtime improvement compared to the most efficient prior work~\cite{braverman2025practical},
and defer the discussion to the related work in \Cref{sec:related}.} 



\begin{figure}[H]
\centering
\resizebox{\columnwidth}{!}{%
\begin{tabular}{|l|l|l|l|l|l|}
\hline
\multicolumn{1}{|c|}{\textbf{Matrix (n × n)}} & \multicolumn{1}{c|}{\textbf{Protocol}} & \multicolumn{1}{c|}{\textbf{Server (ms)}} & \multicolumn{1}{c|}{\textbf{Client (ms)}} & \multicolumn{1}{c|}{\textbf{Total (ms)}} & \multicolumn{1}{c|}{\textbf{Client / Total}} \\ \hline
36864 × 36864                                 & MOSAIC                                 & 8.57                                      & 0.29                                      & 8.86                                     & 3.3\%                                        \\ \hline
16384 × 16384                                 & MOSAIC                                 & 2.04                                      & 0.16                                      & 2.20                                     & 7.3\%                                        \\
                                              & \cite{braverman2025practical}                                & 1270                                      & 47.5                                      & 1318                                   & 3.6\%                                        \\ \hline
8192 × 8192                                   & MOSAIC                                 & 0.90                                      & 0.12                                      & 1.02                                     & 11.8\%                                       \\
                                              & \cite{braverman2025practical}                                & 270.0                                     & 14.5                                      & 284.5                                    & 5.1\%                                        \\
                                              & \cite{benhamouda2025encrypted}$^{\dagger}$                    & –                                         & –                                         & $\sim$500                                & –                                          \\ \hline
4096 × 4096                                   & MOSAIC                                 & 0.48                                      & 0.12                                      & 0.60                                     & 20.0\%                                       \\
                                              & \cite{braverman2025practical}                                & 57.6                                      & 6.05                                      & 63.65                                    & 9.5\%                                        \\
                                              & \cite{benhamouda2025encrypted}$^{\ddagger}$                   & –                                         & –                                         & $\sim$130                                & –                                          \\ \hline
\end{tabular}%
}
\caption{\revised{Wall-clock runtime of a single secure matrix-vector multiplication
for MOSAIC versus prior state-of-the-art~\cite{braverman2025practical,benhamouda2025encrypted}.
Entries for~\cite{benhamouda2025encrypted} are 
placed in the nearest row by dimension:
$^{\dagger}$~$8192 \times 10000$ and $^{\ddagger}$~$2048 \times 10000$.}}
\label{fig:matmul-comparison}
\end{figure}

\paragraph{2) Secure linear layer and noise mitigation}
From our secure approximate MatMul protocol,
we construct a full linear layer protocol which 
is applied as a drop-in-replacement for the
General Matrix Multiply (GEMM) or $\alpha WX + \beta C$. 
We highlight two challenges.

Our secure linear layer ($\Pi_{\textsf{Sec-linear}}$),
shown in ~\Cref{fig:intro-outsourcing},
introduces both fixed-point quantization (16-bit)
and protocol noise ($WE_{x} + E_{w}X + E_{w}E_{x}$).
We note that the fixed-point quantization error is inherent
to all cryptographic protocols~\cite{benhamouda2025encrypted,braverman2025practical,wang2025game}
that securely outsource linear computation. 
However, fixed-point impact on model accuracy has not been evaluated in prior, related cryptographic techniques.
Let the private inputs $X$ and model weights ($W$)
be cast to 16-bit integers (with appropriate fixed-point scales) and $WX$ aggregated over a 32-bit integer ring.
To mitigate the quantization error, we propose the 
application of an efficient, random Hadamard rotation~\cite{ashkboos2024quarot} to activation and weights matrices.
With a series of error accumulation studies
on 70B LLM models, 
we show that even with our protocol noise, the accuracy 
of the internal activations across all internal layers is close to full precision (BF16) and can exceed common NF4 and INT8 quantization.
Further, we report 
accuracy studies (Perplexity, HumanEval) with our protocol
to demonstrate that our (noise) parameter range 
required for 140-bit security is practical
even for models with known quantization challenges (e.g., LLaMA3-70B~\cite{qin2024uniqueness}).

\paragraph{3) Case studies for layer-by-layer outsourcing.} 
We evaluate two application scenarios for our solution.
In our primary use-case in~\Cref{sec:conf-ai-casestudies}, 
we report on prefill and decoding runtimes 
in a \emph{data-center} setting, 
where a trusted computing base is extended with untrusted GPUs, connected with fast interconnects. 
\revised{
Here, we highlight two important ongoing trends in AI. 
Firstly, the leading open and closed sourced
frontier models now exceed 1 trillion parameters; 
larger model dimensions permit ever more efficient 
outsourcing of AI computations in MOSAIC (\Cref{fig:client-overhead,fig:theoretical-overhead}),
which induces a client overhead that scales \emph{optimally} with
larger matrices, unlike prior work. 
Secondly, whilst MOSAIC
requires layer-by-layer communication of masked activations,
a standard, non-confidential inference in practice
already induces layer-by-layer communication in today's AI data centers, as a single inference task for frontier models is no longer feasible on a single cluster,
necessitating fast remote direct memory access (RDMA)
networking across the data center~\cite{wu2026dualpath,spheron2026gpu}. 
Solutions such as MOSAIC can leverage existing ultra-low latency interconnect technology for scalable, confidential AI computation.
Indeed our experiments provide evidence that modern
GPU-to-GPU interconnects are not the bottle-neck for MOSAIC
and that outsourcing AI computations in the data-center is efficient in practice (\Cref{fig:datacenter-runtime-decode-postccs}). } 

In a secondary use-case illustrated in~\Cref{sec:application-remote}, 
we implement a remote decisional inference
application, where the expensive AI computation is outsourced to 
the cloud over the public internet.
In this application, only a single forward pass may be sufficient to produce a complicated decision,
limiting the number of round-trips. 
We show that such a decision could be produced in under 5 seconds
for large 70B models.


\begin{revisedblock}
\paragraph{Current hardware limitations.} 
Cryptographic protocols such as MOSAIC operate
over integer rings, necessitating 32-bit integer arithmetic,
which is not natively supported in the fastest cores of modern AI accelerators.
We implement an emulation thereof over 8-bit Tensor cores,
which results in a ${\sim}3\times$ (decode) to $5$--$11\times$ (prefill)
slow-down versus a baseline, non-confidential inference on 70B models.
We think much of this gap can be explained by the emulation
overhead, in which each 32-bit integer MatMul 
is realized with 10x underlying 8-bit integer 
MatMul launches. We hope that works such as MOSAIC can
inspire AI hardware manufacturers to consider native 32-bit
integer support in future accelerator designs.
\end{revisedblock}

\section{Related work}
\label{sec:related}
We outline three categories of related work.
(1) \emph{Private inference} refers to the task
of computing inference over client's private inputs and server's private
model. 
(2) \emph{Secure outsourcing of matrix multiplications} assumes a client with both
inputs and model weights that wishes to offload
the computation to an untrusted server. (3) \emph{Model obfuscation} refers to the task
of keeping proprietary model weights private from an untrusted server or accelerator. 
\Cref{tab:intro-comparison} summarizes our discussion and comparison to related work.

\paragraph{Private inference.}
Private inference is typically realized with cryptographic techniques
such as secure multi-party computation (MPC) or 
fully homomorphic encryption (FHE).
MPC techniques focus on secret-sharing (SS)~\cite{chen2020maliciously,mann2023towards,li2024nimbus}
and functional secret-sharing (FSS)~\cite{gupta2023sigma,lu2023bumblebee}.
The latter still requires 40s per autoregressive decoding step for a small 13B transformer model, thus not enabling practical runtimes for most applications. Our proposal
demonstrates autoregressive decoding runtimes
below 1s for a 70B model (\Cref{fig:datacenter-runtime-decode-postccs}).
Approaches with fully homomorphic encryption~\cite{jiang2018secure,gao2024secure,zimerman2024power} employ a powerful, general cryptographic primitive
that incurs a very high computational overhead
that is multiple orders of magnitude slower than 
dedicated matrix multiplication protocols.
Gazelle proposes a hybrid between MPC and FHE~\cite{juvekar2018gazelle}.
We note that FHE schemes like CKKS~\cite{cheon2017homomorphic} exist
which enable approximate arithmetic. 
These are distinctly different than our approximate secure MatMul protocol, 
which does not incur the heavy overhead of FHE.

\paragraph{Secure outsourcing of matrix multiplications.}
Slalom~\cite{tramer2018slalom} 
outsources the computationally expensive matrix multiplication operations layer-by-layer 
to a remote GPU, an outsourcing pattern adopted 
by this work as well.
However, it does not provide
privacy for the model weights and
requires expensive pre-processing material
consumed with each inference run. Our work protects both the model and client inputs.

A recent work~\cite{benhamouda2025encrypted} introduces a matrix encryption technique
from pairs of secret dual codes, whose orthogonality cancels the dominant masking cross-term during recovery. The construction focuses on matrix-vector multiplication
and \revised{operates exclusively over prime fields}. 
This approach can be efficient for models with small dimensions \revised{but (comparable) runtimes are superseded by~\cite{braverman2025practical}.}
%
The authors of~\cite{braverman2025practical} introduce a matrix outsourcing
technique (trapdoor matrices) 
that relies on recursive masking with both 
low-rank, dense and sparse full-rank masking elements - the
security of their scheme reduces to decisional LPN.
The trusted, outsourcing client incurs a $O(mn^{\epsilon}l)$ cost
and scales poorly for concrete $\epsilon$ and 
larger matrix dimensions required in modern LLM models.
Moreover, the paper introduces a conjecture that LPN security holds
for small $\epsilon$ in order for concrete efficiency with large 
matrix dimension $n$, 
and no concrete protocol parameters are provided 
to support this in practice.
\revised{Still, \cite{braverman2025practical} is the closest
to our work in concrete performance and technique};
their recursive masking technique permits
an elegant trade-off between memory and online computation,
which we show can be adapted as an optimization to 
our approximate matrix multiplication scheme (\cref{sec:lwe-lpn-matmul}) whilst retaining optimal client overhead.

Our key cryptographic contribution is to introduce a secure approximate matrix multiplication scheme with optimal computational complexity for the client  ($O((m+n)l)$) and the untrusted accelerator ($O(mnl)$), providing accuracy comparable to full precision in many practical AI computation tasks.
\revised{In~\Cref{fig:matmul-comparison},
we highlight client and server run-times in MOSAIC 
that are 
orders of magnitude faster than those reported by~\cite{braverman2025practical},
which in turn improves on all prior work for matrix dimension
$n \geq 2048$.
Note that~\cite{braverman2025practical} only reports CPU runtimes and provides no public protocol parameters for reproduction. Client/Total walltime ratios are skewed in MOSAIC by slower Nvidia CUDA core implementation of the client, whilst the server is implemented on the faster Nvidia Tensor cores (See client vs. server implementations in~\Cref{sec:conf-ai-casestudies}).
We show the theoretical 
client/total computation complexity ratio 
in~\cref{fig:theoretical-overhead} for a large range
of matrix dimensions.} 


\paragraph{Model obfuscation.} 
A line of work \cite{li2024translinkguard,xu2024tempo,shen2022soter,wang2025game} 
intends to shield private LLM models
from users, by obfuscating weights and hosting
de-obfuscation material inside the user's trusted
execution environment (TEE). These works use 
efficient scaling
and permutation techniques to obfuscate the weights, 
but such approaches lack security guarantees;
\cite{wang2025game} introduces a deobfuscation attack
which exploits similarity between pre-trained open-source
models and fine-tuned, private weights. Further,
the authors introduce an improved model weights masking
scheme called ArrowCloak, 
for which they claim a reduction to LWE hardness.
However, \emph{we show that the security of their reduction 
is broken} in \Cref{sec:obfuscation-security}.
Whilst such techniques could in principle
be related to matrix multiplication outsourcing, 
they rely on heuristic security 
arguments.



\section{Background}
\label{sec:background}

\paragraph{Learning with Errors (LWE) and Learning Parity with Noise (LPN)}

The security of our solution reduces to two standard post-quantum
distinguishing assumptions. \emph{Decisional LWE}~\cite{regev2009lattices}:
distinguish samples $(\vec{a}, \langle\vec{a},\vec{s}\rangle + e)$ from
uniform, where $\vec{a}\in R^n$ is uniform, $\vec{s}\in R^n$ secret, and
$e$ is small Gaussian noise; hardness is conjectured for polynomially
many samples. \emph{Decisional LPN}~\cite{alekhnovich2003more}: the same
shape but $e$ is sparse uniform with noise rate $\mu$ above a secret-rank and security level determined threshold.

In this work, we mask weights and activations with a low-rank dense term plus small Gaussian noise (LWE, \Cref{sec:lwe-matmul-core}), then \emph{nest} the
masking by recursively replacing the dense term with an LPN instance of
strictly smaller rank (\Cref{sec:lwe-lpn-matmul}).

\paragraph{Transformer architecture}

A decoder-only transformer maps tokens to logits via an embedding, $L$
\emph{transformer blocks} (or \emph{model layers}), and a final
language-model head. Each model layer follows the residual + pre-norm pattern
\begin{equation}
\label{eq:transformer-layer}
    \vec{h}_\ell
    = \vec{h}_{\ell-1}
      + \mathrm{MLP}\!\left(
          \vec{h}_{\ell-1}
          + \mathrm{Attn}(\mathrm{RMSNorm}(\vec{h}_{\ell-1}))
        \right),
\end{equation}
combining attention and a position-wise MLP, each preceded by 
$\mathrm{RMSNorm}(\vec{h}) = \vec{h}/(\norm{\vec{h}}_2/\sqrt{d}) \odot
\vec{g}$, which projects $\vec{h}$ onto a sphere of radius
$\norm{\vec{g}}_2$, collapsing any radial disagreement.

Attention projects $X\in\mathbb{R}^{s\times d_\text{model}}$
into $h$ query/key/value heads via $\mat{W}_Q,\mat{W}_K,\mat{W}_V$,
and computes per head $\mathrm{softmax}(Q_i K_i^\top/\sqrt{d_k})V_i$, 
then re-projects through $\mat{W}_O$. 
The position-wise MLP is a SiLU-gated
two-layer feed-forward
$\mat{W}_\text{down}(\mathrm{SiLU}(\mat{W}_\text{gate}X) \odot
\mat{W}_\text{up}X)$ with $d_\text{ff} > d_\text{model}$.

The seven dense projections ($\mat{W}_Q,\mat{W}_K,\mat{W}_V,
\mat{W}_O,\mat{W}_\text{gate},\mat{W}_\text{up},\mat{W}_\text{down}$)
dominate compute and memory; softmax, RMSNorm, and SiLU are cheap
elementwise ops. Our protocol outsources only those seven dense matmuls,
leaving non-linear and per-head work on the trusted client.

\paragraph{GPU compute and quantisation}

Modern NVIDIA GPUs expose two compute paths: \emph{CUDA cores}
(general purpose cores) and \emph{tensor cores}
(matrix-multiply-accumulate units in
FP16/BF16/INT8/FP8 with FP32 or INT32 accumulation). Tensor-core
throughput exceeds CUDA-core matmul by an order of magnitude;
LLM inference is consequently bottlenecked by tensor-core matmul and
memory bandwidth.
\emph{Quantisation} casts weights and activations from BF16 to lower
precision to shrink the model and accelerate compute; we use NF4,
INT8 (and BF16) model accuracy baselines in \Cref{sec:error-propagation}.

\begin{figure*}
\centering
\begin{tikzpicture}[
    op/.style={font=\large, inner sep=2pt},
    dlabel/.style={font=\tiny, gray},
]

\node[draw, minimum width=1.4cm, minimum height=1.4cm, fill=blue!15,
    font=\small, align=center, inner sep=0pt] (W) at (0, 0) {$\mat{W}$};
\node[dlabel, below=1pt] at (W.south) {$m \times n$};

\node[op] at (1.05, 0) {$+$};

\node[draw, minimum width=0.45cm, minimum height=1.4cm, fill=green!20,
    font=\small, align=center, inner sep=0pt] (L) at (1.625, 0) {$\mat{L}$};
\node[dlabel, below=1pt] at (L.south) {$m \times r$};

\node[draw, minimum width=1.4cm, minimum height=0.45cm, fill=green!20,
    font=\small, align=center, inner sep=0pt] (M) at (2.7, 0) {$\mat{M}$};
\node[dlabel, below=1pt] at (M.south) {$r \times n$};

\node[op] at (3.75, 0) {$+$};

\node[draw, densely dotted, minimum width=1.4cm, minimum height=1.4cm, fill=red!10,
    font=\small, align=center, inner sep=0pt] (Ew) at (4.8, 0) {$\mat{E}_{w}$};
\node[dlabel, below=1pt] at (Ew.south) {$m \times n$};

\node[op] at (5.85, 0) {$=$};

\node[draw, minimum width=1.4cm, minimum height=1.4cm, fill=gray!20,
    font=\small, align=center, inner sep=0pt] (Wp) at (6.9, 0) {$\mat{W}'$};
\node[dlabel, below=1pt] at (Wp.south) {$m \times n$};

\end{tikzpicture}
\caption{Structure of the masking scheme, 
        consisting of private weights matrix $W$,
        left low-rank matrix $L$, middle low-rank (r)
        matrix $M$
        and full-rank gaussian $E_{w}$.
         The activation matrix $\mat{X}$ is masked analogously.}
\label{fig:masking-structure}
\end{figure*}
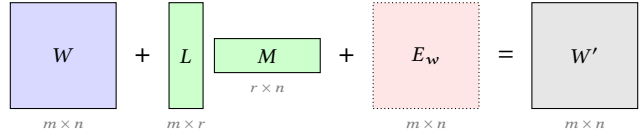

\begin{figure*}
\centering
\begin{tikzpicture}[
    op/.style={font=\large, inner sep=2pt},
    dlabel/.style={font=\tiny, gray},
    mbox/.style={draw, font=\scriptsize, align=center, inner sep=1pt},
    annot/.style={font=\tiny, align=center},
]


\node[mbox, fill=gray!20, minimum width=1.4cm, minimum height=1.4cm] (WpXp) at (0, 0.3) {$\mat{W}'\mat{X}'$};
\node[dlabel, below=0pt] at (WpXp.south) {$m \times l$};

\node[op] at (1.05, 0.3) {$=$};

\node[mbox, fill=blue!15, minimum width=1.4cm, minimum height=1.4cm] (ywx) at (2.10, 0.3) {$\mat{W}\mat{X}$};
\node[dlabel, below=0pt] at (ywx.south) {$m \times l$};

\node[op] at (3.15, 0.3) {$+$};

\node[mbox, fill=green!20, minimum width=0.45cm, minimum height=1.4cm] (L1) at (3.625, 0.3) {$\mat{L}$};
\node[dlabel, below=0pt] at (L1.south) {$m\!\times\!r$};

\node[mbox, fill=green!20, minimum width=1.4cm, minimum height=0.45cm] (MXE) at (4.70, 0.3) {\tiny$\mat{M}(\mat{X}\!+\!\mat{E}_x)$};
\node[dlabel, below=4pt] at (MXE.south) {$r \times l$};

\node[op] at (5.75, 0.3) {$+$};

\node[mbox, fill=green!20, minimum width=0.45cm, minimum height=1.4cm] (WN) at (6.225, 0.3) {\makebox[0pt]{\tiny$\mat{W}'\!\mat{N}$}};
\node[dlabel, below=0pt] at (WN.south) {$m\!\times\!r$};

\node[mbox, fill=green!20, minimum width=1.4cm, minimum height=0.45cm] (R1) at (7.30, 0.3) {$\mat{R}$};
\node[dlabel, below=4pt] at (R1.south) {$r \times l$};

\node[op] at (8.35, 0.3) {$+$};

\node[mbox, fill=red!10, densely dotted, minimum width=1.4cm, minimum height=1.4cm, text width=1.2cm, align=center] (err) at (9.40, 0.3) {\tiny$\mat{E}_{w}\mat{X}$ $+$ $\mat{W}\mat{E}_x$ $+$ $\mat{E}_{w}\mat{E}_{x}$};
\node[dlabel, below=0pt] at (err.south) {$m \times l$};

\end{tikzpicture}
\caption{Unmasking the outsourced matrix product.
         The GPU returns $\mat{W}'\mat{X}'$, which expands into
         the desired product $\mat{W}\mat{X}$ (blue),
         low-rank correction terms (green), and
         small-norm noise (red).
         The client efficiently removes the green terms
         in $O((m+n)l)$ using precomputed material,
         recovering $\mat{W}\mat{X}$ up to a small approximation error.}
\label{fig:unmasking-visual}
\end{figure*}
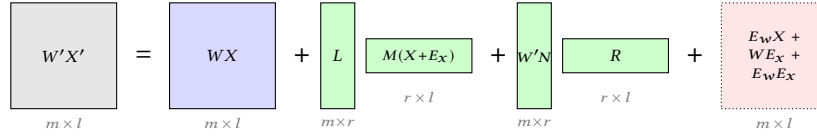

\section{Secure approximate matrix multiplication}
\label{sec:lwe-matmul}

A key task in our confidential AI computation is the outsourcing 
of linear matrix multiplications in the floating-point domain. 
We first focus on the outsourcing of matrix multiplications in the 
the \emph{integer ring} domain of chosen bit-width,
in which our cryptographic protocol operates.
Our solution \emph{approximately}
outsources a matrix multiplication $WX$
over the \emph{integer ring} domain of chosen bit-width - 
where $\mat{W} \in \rngMat{m}{n}$ is a quantized 
model weight matrix and $\mat{X} \in \rngMat{n}{l}$ an activation matrix - to an
untrusted accelerator $\gpu$ while hiding both $\mat{W}$ and $\mat{X}$.
We present an asymptotically optimal, approximate protocol with $O((m+n)l)$ client online cost, and prove post-quantum security under LWE and LPN assumptions.
The optimality claim follows directly from the input $nl$
and output $ml$ dimensions of the outsourcing protocol.

In later ~\Cref{sec:secure-forward-pass}, we bridge the integer ring and 
floating point domains with a \emph{secure linear} protocol
which mitigates the effects of fixed-point quantization and induced 
protocol error in the floating point domain,
in which the interleaved non-linear layers are performed 
on the trusted client.

\subsection{Simple protocol with noisy unmasking}
\label{sec:lwe-matmul-core}

We illustrate the core idea with a simplified construction
to outsource the computation of $WX$, where
$W$ ($m \times n$) and $X$ ($n \times l$).

To mask the quantized weight matrix, 
the trusted client samples
$\mat{L} \leftarrow \uniform{\rngMat{m}{r}}$,
$\mat{M} \leftarrow \uniform{\rngMat{r}{n}}$,
$\mat{E}_{w} \leftarrow \normal{0}{\sigma^{2}}^{m \times n}$
and sets:
\begin{equation}
\label{eq:lwe-weight-mask}
    \mat{W}' = \mat{W} + \mat{L} \mat{M} + \mat{E}_{w}
\end{equation}

Let us denote $L$ the \emph{left} and $M$ the 
\emph{middle} low-rank matrices in the mask, and $E_{w}$ 
the low-norm, full-rank Gaussian \emph{error} term, parameterized by
standard-deviation $\sigma$.
We note that $M$ can be a public matrix whilst retaining
LWE security.
Here, rank $r$ denotes the rank of both uniform,
low-rank matrices 
$L$ and $M$. Critically, $r$ is solely a function of
security parameter $\lambda$ and must be strictly lower than matrix 
dimensions $m,n$ for client efficiency. Given $\lambda \approx 140$,
inner rank $r$ can be set to $1536$ (\Cref{ex:lwe-lpn-parameterisation}). 
For a modern model dimension
of $16\,384 \gg 1536$, the efficiency condition can easily be met (\Cref{fig:client-overhead}).
For smaller model dimensions, we can further reduce the inner
masking rank in our full protocol, detailed in~\Cref{sec:lwe-lpn-matmul}.

Note that the masked matrix $\mat{W}'$ is now indistinguishable from a uniformly
random matrix from the view of the untrusted accelerator $\gpu$.
The security of the encryption reduces to
Learning With Errors (LWE) (\cref{lma:lwe-protocol}). 
Whilst the cost of masking model weights is $O(rmn)$, 
it only has to be performed once for all matrix 
multiplications with fixed $W$.

The input matrix is masked similarly.
For all inference runs, we sample a single
$\mat{N} \leftarrow \uniform{\rngMat{n}{r}}$.
Then, for each individual inference pass,
we sample $\mat{R} \leftarrow \uniform{\rngMat{r}{l}}$
and $\mat{E}_{x} \leftarrow \normal{0}{\sigma^{2}}^{n \times l}$
and set:
\[
    \mat{X}' = \mat{X} + \mat{N}\mat{R} + \mat{E}_{x}
\]
Again, let $N$ be the middle and $R$ denote the
right low-rank masking matrices, where $N$ can be a public matrix.
The client $\cli$ sends $\mat{W}', \mat{X}'$ to the untrusted accelerator $\gpu$, which computes $\mat{Y}' = \mat{W}' \mat{X}'$
and returns the encrypted result $\mat{Y}'$ back to 
the client. In terms of the individual matrix components, this gives
\begin{equation}
    \mat{Y}' = (\mat{W} + \mat{L}\mat{M} + \mat{E}_{w})
               (\mat{X} + \mat{N}\mat{R} + \mat{E}_{x})
\end{equation}
Expanding:
\begin{equation}
    \mat{Y}' 
    = (\mat{W}
      + \mat{L}\mat{M}
      + \mat{E}_{w})\mat{X}
      + (\mat{W} + \mat{L}\mat{M} + \mat{E}_{w})\mat{N}\mat{R}
      + (\mat{W}
      + \mat{L}\mat{M}
      + \mat{E}_{w})\mat{E}_{x}
\end{equation}
\begin{equation}
    = (\mat{W}
      + \mat{L}\mat{M}
      + \mat{E}_{w})\mat{X}
      + W'\mat{N}\mat{R}
      + (\mat{W}
      + \mat{L}\mat{M}
      + \mat{E}_{w})\mat{E}_{x}
\end{equation}
\begin{equation}
    = \mat{W}\mat{X}
      + \mat{L}\mat{M}(\mat{X}+\mat{E}_{x})
      + W'\mat{N}\mat{R}
      + \mat{E}_{w}\mat{X} + \mat{W}\mat{E}_x + \mat{E}_{w}\mat{E}_x
\end{equation}

The client removes two correction terms:
Firstly, $\mat{L}\mat{M}(\mat{X} + \mat{E}_{x})$, which costs $O(r(m+n)l)$;
$\mat{M}(\mat{X} + \mat{E}_{x})$ induces client computational complexity
$O(rnl)$, and subsequent left-multiplication $O(rml)$. 

The second correction term
$(\mat{W}'\mat{N})\mat{R}$ is also a rank $r$ computation 
if $\mat{W}'\mat{N} \in \mathbb{F}_p^{m \times r}$
is precomputed once and reused, costing $O(mrl)$ per subsequent matmul. Since both $W'$ and public matrix 
$N$ can be given to untrusted $\gpu$, 
$W'N$ can also be outsourced.
This process is illustrated
in Figures~\ref{fig:masking-structure} and~\ref{fig:unmasking-visual}, and yields
\begin{equation}\label{eq:approx-error}
    \mat{Y} \approx \mat{Y}' - \mat{L}\mat{M}(\mat{X}+\mat{E}_x) - (\mat{W}'\mat{N})\mat{R}
    = \mat{W}\mat{X} + \mat{E}_{w}\mat{X} + \mat{W}\mat{E}_x + \mat{E}_{w}\mat{E}_{x}
\end{equation}
The three trailing error terms remain as an error in the result.
For fixed security level $\lambda$, the client complexity is $O((m+n)l)$.


\paragraph{Optimal scalability for larger dimensions}
Crucially, LWE security is robust against polynomially many adversarial
samples. This means the LWE parameters (rank $r$ and noise of standard deviation
$\sigma$) are fixed by the security level and do not need to grow with the matrix
dimensions $m, n, l$. Concretely, for 140-bit security, we can fix
$r = 1536$ (a GPU-friendly dimension) and gaussian noise with stddev $\sigma = 0.5$ (\cite{albrecht2015concrete})
regardless of whether the matrix has rank
$4096$ or $29568$.
This constitutes a very small norm
Gaussian relative to the full 32-bit integer ring we instantiate our protocol in.

Whilst the rank $r$ is fixed by the choice of 
security parameter $\lambda$, 
the pure LWE unmasking for client $\cli$ 
still costs $O(r(m+n)l)$ and is linear in
$r$. 
Our full protocol in~\Cref{sec:lwe-lpn-matmul}
improves \emph{concrete efficiency} 
further by replacing 
the uniform, low-rank components with structured, 
sparse masks that are cheaper to undo.



\paragraph{Indistinguishability} 
It is easy to see that the masked weights (and analogously activation) matrices are computationally indistinguishable from uniform. We provide the proof in \Cref{prf:lwe-protocol}.
\begin{restatable}{lemma}{lemLweProtocol}
    \label{lma:lwe-protocol}
    The LWE mask is computationally indistinguishable from uniform:
    \[
        \mat{L}\mat{M} + \mat{E}_{w} \overset{c}{\approx} \mat{U}.
    \]
\end{restatable}
As a consequence, for any $\mat{W}$ the masked weight
$\mat{W}' = \mat{W} + \mat{L}\mat{M} + \mat{E}_{w}$ is computationally
indistinguishable from uniform.
\noindent\textit{Proof sketch in \Cref{prf:lwe-protocol}}.
Analogously for $NR + E_{x} \overset{c}{\approx} U$.

\paragraph{Bounded protocol error}
We first analytically 
bound the protocol error term 
in \Cref{thm:total-error-bound}
 with proof and
supporting lemmas in~\Cref{sec:proofs-bounded-noise}.
In subsequent~\Cref{sec:secure-forward-pass},
we illustrate how to effectively mitigate approximation errors in the context of full AI computations via an efficient random rotation of both $W$ and $X$.

Unmasking of $\mat{W}'\mat{X}'$ leaves the following error terms in the
result: $\mat{W}\mat{E}_{X} + \mat{E}_{W}\mat{X} + \mat{E}_{W}\mat{E}_{X}$,
where $\mat{E}_{W} \in \rngMat{m}{n}$ and $\mat{E}_{X} \in \rngMat{n}{l}$
are discrete Gaussian noise matrices with parameter $\sigma$
(distribution $\normal{\rng;0}{\sigma^{2}}$ over $\rng$, with
$\Pr[X = k] \propto \exp{-k^{2}/(2\sigma^{2})}$).

\begin{restatable}{theorem}{thmTotalErrorBound}
\label{thm:total-error-bound}
    Let $\mat{W} \in \rng^{m \times n}$, $\mat{X} \in \rng^{n \times l}$,
    and let $\mat{E}_{W} \in \rngMat{m}{n}$, $\mat{E}_{X} \in \rngMat{n}{l}$
    have independent entries drawn from $\normal{\rng;0}{\sigma^{2}}$.
    Then entry $(j,k)$ of the total error
    $\mat{E}_{W}\mat{X} + \mat{W}\mat{E}_{X} + \mat{E}_{W}\mat{E}_{X}$
    is approximately
    $\sigma\sqrt{\norm{\vec{x}_{k}}_{2}^{2} + \norm{\vec{w}_{j}}_{2}^{2}}$-sub-Gaussian:
    for all $t > 0$,
    \[
        \Pr\!\left[|(\mat{E}_{W}\mat{X} + \mat{W}\mat{E}_{X} + \mat{E}_{W}\mat{E}_{X})_{jk}| > t\right] 
        \leq 2\exp{-\frac{t^{2}}{2\sigma^{2}(\norm{\vec{x}_{k}}_{2}^{2} + \norm{\vec{w}_{j}}_{2}^{2})}}
    \]
    up to $e^{-\Omega(n)}$,
    where $\vec{w}_{j}$ is the $j$-th row of $\mat{W}$ and
    $\vec{x}_{k}$ is the $k$-th column of $\mat{X}$.
\end{restatable}
\noindent\textit{Proof in \Cref{prf:total-error-bound}.}

\paragraph{Leakage \& integrity} The untrusted accelerator
$\gpu$ clearly learns
the matrix dimensions and invocation of each
secure MatMul outsourcing. Without padding matrices, 
the model architecture can be inferred. With padding,
an upper-bound on the number of model layers and matrix dimensions is leaked,
at the cost of additional computation. 
\revised{A naive padding scheme
could add (1) dummy outsourced linear layers or (2) padded output $m$
or input $n$ dimensions. 
Option (1) incurs additional communication and
client masking per dummy layer, whilst option (2) increases the per-query
client complexity by $O(m'l)$ with output padding $m'$ or $O(n'l)$ with input
padding $n'$, where $l$ is the batch (token) dimension; padding the batch
by $l'$ analogously adds $O((m+n)l')$.}

Integrity can be straightforwardly achieved by applying
well-known Freivalds' algorithm for asserting correctness of
$\mat{W}\mat{X} = \mat{Y}$ computed by the untrusted GPU.
We reproduce this correctness check in~\Cref{sec:freiwald}
for the reader's convenience.

\subsection{Full protocol using noisy recursive masking}
\label{sec:lwe-lpn-matmul}


\begin{figure}
\footnotesize
\begin{tcolorbox}[protocolbox={\textsf{Protocol $\Pi_{\textsf{Sec-Approx-MatMul}}$: Secure Approximate MatMul.}}]
\textbf{Parties:} Trusted client $\cli$
and untrusted accelerator $\mathcal{G}$.

\textbf{Init W:} Compute one-time masking of weights matrix.
\begin{enumerate}[leftmargin=*,labelsep=0.4em,itemsep=2pt]
    \item $\cli$ samples: $\mat{L} \leftarrow \rngMat{m}{r_{d}}$, 
    $\mat{M}_{i \in [d]} \leftarrow \rng^{r_{i} \times r_{i-1}}$, 
    $\mat{N}_{i \in [d]} \leftarrow \rng^{r_{i-1} \times r_{i}}$, \\ 
    $\mat{E}_{w} \leftarrow \mathcal{N}(\rngMat{m}{n},\sigma^{2})$,
    $\mat{S}_{i \in [2:d]} \leftarrow \sparse{\rngMat{r_{i}}{r_{i-1}}; \mu_i}$
    \item $\cli$ computes:
        $\mat{U}_{w} = \mat{L} (\mat{M}_d \mat{M}_{d-1} ... \mat{M}_1) + \sum_{i=1}^{d-1} \mat{S}_{i+1}(\mat{M}_{i} \mat{M}_{i-1} ... \mat{M}_1) + E_{w}$.
    \item $\cli$ stores: 
        $N_1 N_2 ... N_i$ and 
        $\mat{J}_i = (\mat{W}+\mat{E}_{w})\mat{N}_1 \mat{N}_2 ... \mat{N}_i$ for $i \in [d]$.
    \item $\cli$ stores: $\vec{a} \leftarrow \rngVec{n}$, $\vec{Y}^{\textsf{chk}} \leftarrow\vec{a}^{T}\mat{W}'$.
    \item $\cli$ stores $\mat{W}' = \mat{W}+\mat{U}_{w}$
    and $M_{i}M_{i-1}...M_{1}$ for $i \in [d]$ on untrusted $\gpu$.
\end{enumerate} 

\vspace{0.5em}

\textbf{Online phase:} On fresh input $\mat{X}$, securely outsource $\mat{W}\mat{X}$ to $\gpu$ and recover result.
\begin{enumerate}[leftmargin=*,labelsep=0.4em,itemsep=2pt]
    \item \emph{Init X:} $\cli$ computes and stores masking material consumed by X.
    \begin{enumerate}[leftmargin=*,labelsep=0.4em,itemsep=2pt]
        \item $\mat{R} \leftarrow \rngMat{r_{d}}{l}$,
        $\mat{E}_{x} \leftarrow \mathcal{N}(\rngMat{n}{l},\sigma^{2})$
        and $\mat{T}_{i \in [2:d]} \leftarrow \sparse{\rngMat{r_{i-1}}{r_{i}}; \mu_i}$.
        \item $\mat{U}_{x} = \mat{N}_1 \mat{N}_2 ... \mat{N}_d \mat{R} + \sum_{i=1}^{d-1} \mat{N}_1 \mat{N}_2 ... \mat{N}_i \mat{T}_{i+1} + \mat{E}_{x}$.
        \item $\mat{K} = \mat{J}_d\mat{R} + \sum_{i=1}^{d-1} 
        \mat{J}_i T_{i+1}$ (Unmasking pre-computation).
    \end{enumerate} 
    \item $\cli$ forwards to $\gpu$: $\mat{X}' = X +\mat{U}_{x}$.
    \item $\gpu$ returns $\mat{Y}' = \mat{W}'\mat{X}'$ and
        $\mat{Q}_{i} = (\mat{M}_{i} \mat{M}_{i-1} ... \mat{M}_1)\mat{X}'$ for $i \in [d]$.
    \item $\cli$ asserts $\mat{Y}^{\textsf{chk}}\mat{X} = \mat{a}^T \mat{Y}'$ and aborts otherwise.
   \item $\cli$ outputs $\mat{Y} \approx \mat{Y}' - K
    - L \mat{Q}_{d} - \sum_{i=1}^{d-1} S_{i+1} \mat{Q}_i$ (Unmask Y).
\end{enumerate}



\end{tcolorbox}
\caption{Secure approximate matrix multiplication outsourcing protocol.}
\label{proto:approx-matmat-gpu}
\end{figure}

Next, we describe a full protocol,
specified in~\Cref{proto:approx-matmat-gpu},
that takes advantage of 
a recursive masking technique from
~\cite{braverman2025practical} by
recursively replacing dense, uniform matrices
of low-rank with LPN instances of ever decreasing rank.
In contrast to~\cite{braverman2025practical},
the outer LWE instance (as described in~\Cref{sec:lwe-matmul}) has a fixed rank
determined by security parameter $\lambda$,
and is independent of matrix dimensions.
The protocol~\cite{braverman2025practical} in
original form has a mask with rank that
must \emph{scale} with matrix dimensions.
We obtain trusted client efficiency
of $O((m+n)l)$, rather than $O(mn^{\epsilon}l)$
in~\cite{braverman2025practical}.

For intuition, we replace the left private 
$L$ matrix from our naive LWE-only protocol (\Cref{eq:lwe-weight-mask}) with an LPN instance $\mat{L}_{2} \mat{M}_2 + S_2$,
that is computationally indistinguishable from uniform $L$.
\[
\mat{W'} = \mat{W} + \mat{L}_{1}\mat{M}_{1} + \mat{E}_{w} 
\overset{c}{\approx} \mat{W} + (\mat{L}_{2} \mat{M}_2 + S_2)\mat{M}_1 + \mat{E}_{w} 
\]
where $L_2,M_2$ are dense and uniform
and $S_2$ a sparse matrix (sparse LPN error).
The rank of $\mat{L}_{2}\mat{M}_{2}$
is strictly smaller than the original
$L$ ($r_2 < r_1$), allowing us to decrease the internal rank of the
non-error masking term. 

%

We continue to apply the substitution to $\mat{L}_{2}$
\begin{equation}
\mat{W'} 
\overset{c}{\approx} \mat{W} + ((\mat{L}_{3}\mat{M}_3 + \mat{S}_3) \mat{M}_2 + S_2)\mat{M}_1 + \mat{E}_{w} 
\end{equation}
More generally over multiple substitution steps
\begin{align}\label{eq:weight-mask}
    \mat{W'} \overset{c}{\approx} 
    \mat{W} + (((\mat{L}_{d}\mat{M}_{d} + \mat{S}_d)\mat{M}_{d-1} + \mat{S}_{d-1})... \mat{S}_{2})\mat{M}_{1} + \mat{E}_{w}\nonumber \\
    = \mat{W} + \mat{L} (\mat{M}_d \mat{M}_{d-1} ... \mat{M}_1) + \sum_{i=1}^{d-1} \mat{S}_{i+1}(\mat{M}_i \mat{M}_{i-1} ... \mat{M}_{1}) + \mat{E}_{w}
\end{align}
Note the rank of the dense term 
$\mat{L} (\mat{M}_d \mat{M}_{d-1} ... \mat{M}_1)$ has been reduced to $r_{d} < r_{d-1} < ... < r_1$ for $\mat{L}$ as 
a dense $m \times r_{d}$ matrix. 
This ultimately permits a trade-off between storage and
client runtime, by reducing the inner rank of matrix computations.
We illustrate this trade-off concretely in
\Cref{ex:lwe-lpn-parameterisation}
and \Cref{tab:lpn-tradeoff}.

We emphasize that inner 
dimensions $r_i$ are solely a function of security parameter $\lambda$, and independent
of the matrix dimensions.
Further, the terms $\mat{S}_{i+1}(\mat{M}_i \mat{M}_{i-1} ... \mat{M}_{1})$
for $i \in [1:d-1]$
are strictly over sparse matrices, with row sparsity also independent of the matrix sizes.

During the online phase, the input matrix of dimension $n \times l$
is masked analogously, naturally imposing $O(nl)$ online runtime. 
\begin{equation}\label{eq:activation-mask}
\mat{X}' = \mat{X} + \mat{N}_1 \mat{N}_2 ... \mat{N}_d \mat{R} + \sum_{i=1}^{d-1} \mat{N}_1 \mat{N}_2 ... \mat{N}_i \mat{T}_{i+1} + \mat{E}_{x}
\end{equation}
The (public) matrix product chains can be precomputed during initialization and outsourced. 
The masking operation has
$O(nl)\approx O(nl + n r_d l + r_{1} \mu_2 r_{2} l + ... + r_{d-1} \mu_d r_{d} l)$ complexity,
where sparsity $\mu_i$ and $r_i$ are fixed by $\lambda$ and independent of $m,n,l$.
We note that the masks for quantized activation $X$ can be computed independently
from $X$.

\paragraph{Unmasking}
Upon receiving $\mat{W}'\mat{X}'$ from the untrusted accelerator $\gpu$, 
the client removes all cross terms with low-rank dense and sparse matrices in $O(ml)$, whilst retrieving a result
with the same error as the simplified
LWE-only protocol in~\Cref{eq:approx-error}.


Let $\mat{W'}\mat{X'} = (\mat{W} + \mat{U}_{w})(\mat{X} + \mat{U}_{x}) = WX + WU_x + U_w X + U_w U_x$,
where $\mat{U}_{w}$ and $\mat{U}_{x}$ are the matrix masks from 
\Cref{eq:weight-mask} and \Cref{eq:activation-mask} respectively.
To recover $\mat{W}\mat{X}$, 
the trusted client must remove the cross-terms, or
a close approximation thereof. Firstly, consider term $\mat{W}\mat{U}_{x}$, 
\begin{equation}
\mat{W}\mat{U}_{x} = (\mat{W}\mat{N}_1 \mat{N}_2 ... \mat{N}_d) \mat{R} + \sum_{i=1}^{d-1} (\mat{W}\mat{N}_1 \mat{N}_2 ... \mat{N}_i) \mat{T}_{i+1} + \mat{W}\mat{E}_{x} 
\end{equation}
This term can be efficiently computed up to $\mat{W}\mat{E}_{x}$ 
by the client during the online phase,
given pre-computation of the chain products $\mat{W}\mat{N}_{1}\mat{N}_{2}...\mat{N}_{i}$,
amortized over all multiplications over the same weight matrix.
The client complexity to compute this term is $O(ml) = O(mr_dl) + O(m \mu_2 r_2 l) +...+ O(m \mu_d r_d l)$ as $r_2, r_3 ..., r_d$, $\mu_2, \mu_3,...,\mu_d$ 
are fixed by the security parameter and not by matrix dimension $m,n,l$.

The cross-term $\mat{U}_{w}\mat{X}' = U_w (X+ U_x)$ 
is computed up to error
$\mat{E}_{w}\mat{X} + \mat{E}_{w}\mat{E}_{x}$.
To illustrate this, consider the following expansion.
\begin{equation}
\mat{U}_{w}\mat{X}' =
\mat{L} (\mat{M}_d \mat{M}_{d-1} ... \mat{M}_1\mat{X}') + \sum_{i=1}^{d-1} \mat{S}_{i+1}(\mat{M}_i \mat{M}_{i-1} ... \mat{M}_1\mat{X}') + \mat{E}_{w}\mat{X}'
\end{equation}
Here, the multiplicative terms $\mat{M}_i \mat{M}_{i-1} ... \mat{M}_1 \mat{X}'$ can be securely outsourced to the untrusted
accelerator $\gpu$. 
Thus, the first and second terms above can be computed by the client in
$O(ml) = O(mr_dl) + O(m \mu_2 r_2 l) +...+ O(m \mu_d r_d l)$ time.
The error term expands to $\mat{E}_{w}\mat{X}' = \mat{E}_{w}\mat{X} + \mat{E}_{w}\mat{U}_{x}$, of which $\mat{E}_w \mat{X}$ 
is left as protocol error. We expand the trailing term
$\mat{E}_{w}\mat{U}_{x}$ further;
\begin{equation}
     \mat{E}_{w}\mat{U}_{x}  = (\mat{E}_{w} \mat{N}_1 \mat{N}_2 ... \mat{N}_d) \mat{R} + \sum_{i=1}^{d-1} (\mat{E}_{w}\mat{N}_1 \mat{N}_2 ... \mat{N}_i) \mat{T}_{i+1} +  \mat{E}_{w}\mat{E}_{x}
\end{equation}
Here, terms $\mat{E}_{w} \mat{N}_1 ... \mat{N}_i$ can be 
preprocessed by the client and amortized over all matrix multiplications over
the quantized model weights,
whilst $\mat{E}_{w}\mat{E}_{x}$ is left as protocol error.
Thus the trusted client complexity of
computing this is $O(ml) = O(mr_dl) + O(m \mu_2 r_2 l) +...+ O(m \mu_d r_d l)$:
again, rank $r_i$ and sparsity $\mu_{i}$ are fixed and independent of matrix dimensions.

Overall, this represents an optimization to
reduce the rank of dense matrix multiplications 
at runtime compared to the
simplified protocol in~\Cref{sec:lwe-matmul-core}.
The final resulting protocol error is
$\mat{W}\mat{E}_x + \mat{E}_w \mat{X} + \mat{E}_w \mat{E}_x$,
identical to the naive protocol in~\Cref{sec:lwe-matmul}.

\begin{example}
\label{ex:lwe-lpn-parameterisation}
We provide an example parameterization for security level
$\lambda = 140$. Security of the LWE and LPN instances were parameterized
with the LWE~\cite{albrecht2015concrete} 
and LPN~\cite{YYW+25} estimators by Albrecht et al. and 
Yu et al. respectively.
We assume $d = 3$ for the number of nested LWE/LPN
instances. 
\begin{itemize}
    \item LWE: $r_1 = 1536, \sigma = 0.5$
    \item LPN-1: $r_2 = 1024, N_2 = 1536, t_2 = 80$
    \item LPN-2: $r_3 = 652, N_3 = 1024, t_3 = 80$
\end{itemize}
Let $r_i$ be the secret dimension and $N_i$ the number of permitted
adversarial samples (polynomially many for LWE). 
For LWE, $\sigma$ denotes the standard deviation of the dense, Gaussian noise term. For LPN instances, $t$ denotes the hamming weight of the noise vector.
\end{example}

\paragraph{Compute-storage tradeoff of the LPN nesting}
Each additional LPN level shrinks the dense rank that the client
multiplies through ($r_1 \to r_2 \to r_3$) at the cost of one extra
chain product $\mat{N}_1 \cdots \mat{N}_i$ to cache and one extra
sparse mask $\mat{S}_{i+1}$ of Hamming weight $t_{i+1}$.
\Cref{tab:lpn-tradeoff} makes this concrete for the parameters
above, omitting (i) the small sparse-mask client state
$\sum_{i=2}^{d} r_i t_i$ and
(ii) the GPU-side chain products $(r_1{+}\cdots{+}r_d)\,n$.
\begin{table}[H]
\centering
\resizebox{\columnwidth}{!}{%
\begin{tabular}{c c l}
\toprule
$d$ & $\cli$ compute$/(m{+}n)\,l$ &  $\cli$ state \\
\midrule
LWE only
  & $r_1 = 1536$
  & $(m{+}n)\,r_1 + m\,r_1 = 1536\,(2m{+}n)$ \\
+ LPN-1
  & $r_2{+}t_2 = 1104$
  & $(m{+}n)(r_1{+}r_2) + m\,r_2 = 2560\,(m{+}n) + 1024\,m$ \\
+ LPN-2
  & $r_3{+}t_2{+}t_3 = 812$
  & $(m{+}n)(r_1{+}r_2{+}r_3) + m\,r_3 = 3212\,(m{+}n) + 652\,m$ \\
\bottomrule
\end{tabular}%
}
\caption{Client compute/storage tradeoff of mask nesting.}
\label{tab:lpn-tradeoff}
\end{table}


Under decisional LWE and LPN, the nested LWE+LPN mask is computationally
indistinguishable from uniform (\Cref{lma:lwe-lpn-protocol}); hence for
any $\mat{W}$ the masked matrix $\mat{W}'$ is indistinguishable from
uniform.
\begin{restatable}[Security of $\Pi_{\textsf{Sec-Approx-MatMul}}$]{theorem}{MatMulProtocol}
    \label{thm:matmul-protocol}
    Under decisional LWE and LPN:
    \begin{enumerate}
        \item \emph{(Privacy.)} The adversary $\gpu$'s view $(\mat{W}',\mat{X}')$
        is computationally indistinguishable from uniform over
        $\rng^{m \times n} \times \rng^{n \times l}$.
        \item \emph{(Integrity.)} See~\Cref{sec:freiwald}.
    \end{enumerate}
\end{restatable}
\begin{proof}
    \emph{(Privacy.)} \Cref{lma:lwe-lpn-protocol} in~\Cref{prf:lwe-lpn-protocol} gives
    $\mat{W}' \overset{c}{\approx} \mat{U}_{m \times n}$ under decisional
    LWE and LPN. The activation mask in~\Cref{eq:activation-mask}
    follows the same nested LWE+LPN form (with roles of
    $\mat{L},\mat{M}_i,\mat{S}_i,\mat{E}_w$ taken by
    $\mat{R},\mat{N}_i,\mat{T}_i,\mat{E}_x$) and is sampled
    independently, so an analogous hybrid argument yields
    $\mat{X}' \overset{c}{\approx} \mat{U}_{n \times l}$. 
\end{proof}

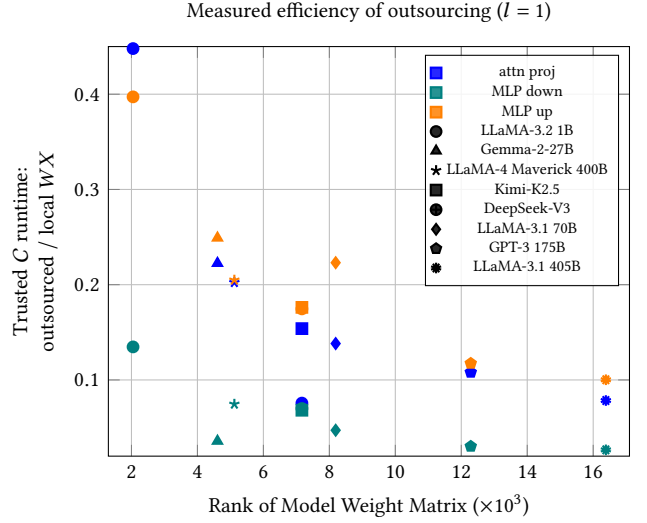
\begin{figure}
  \centering

\pgfplotsset{
    scatter/classes={
        L32_1B={mark=*},
        Gemma2_27B={mark=triangle*},
        L4Mav={mark=star},
        Kimi_K25={mark=square*},
        DSV3={mark=oplus*},
        L31_70B={mark=diamond*},
        GPT3={mark=pentagon*},
        L31_405B={mark=10-pointed star}
    }
}

\begin{tikzpicture}
\begin{axis}[
    title={Measured efficiency of outsourcing ($l=1$)},
    xlabel={Rank of Model Weight Matrix ($\times 10^3$)},
    ylabel={Trusted $\cli$ runtime:\\outsourced $/$ local $\mat{W}\mat{X}$},
    ylabel style={font=\small, align=center},
    xmin=1.3, xmax=17.1,
    ymin=0.020, ymax=0.450,
    grid=major,
    width=\linewidth,
    height=7cm,
    xticklabel style={/pgf/number format/1000 sep={}},
    every axis plot/.append style={thick},
    tick label style={font=\small},
    label style={font=\small},
    title style={font=\small},
    legend style={
        at={(0.98,0.98)},
        anchor=north east,
        font=\scriptsize,
        draw=black,
        fill=white,
        fill opacity=0.85,
        text opacity=1,
        inner sep=2pt,
        row sep=-2pt,
    },
]

\addlegendimage{only marks, mark=square*, color=blue,   fill=blue}
\addlegendentry{attn proj}
\addlegendimage{only marks, mark=square*, color=teal,   fill=teal}
\addlegendentry{MLP down}
\addlegendimage{only marks, mark=square*, color=orange, fill=orange}
\addlegendentry{MLP up}

\addplot[scatter, scatter src=explicit symbolic, only marks, color=blue, fill=blue, forget plot]
    coordinates {
        (2.048,  0.4479) [L32_1B]
        (4.608,  0.2225) [Gemma2_27B]
        (5.120,  0.2023) [L4Mav]
        (7.168,  0.1539) [Kimi_K25]
        (7.168,  0.0757) [DSV3]
        (8.192,  0.1382) [L31_70B]
        (12.288, 0.1080) [GPT3]
        (16.384, 0.0784) [L31_405B]
    };

\addplot[scatter, scatter src=explicit symbolic, only marks, color=teal, fill=teal, forget plot]
    coordinates {
        (2.048,  0.1347) [L32_1B]
        (4.608,  0.0358) [Gemma2_27B]
        (5.120,  0.0746) [L4Mav]
        (7.168,  0.0682) [Kimi_K25]
        (7.168,  0.0700) [DSV3]
        (8.192,  0.0472) [L31_70B]
        (12.288, 0.0303) [GPT3]
        (16.384, 0.0266) [L31_405B]
    };

\addplot[scatter, scatter src=explicit symbolic, only marks, color=orange, fill=orange, forget plot]
    coordinates {
        (2.048,  0.3972) [L32_1B]
        (4.608,  0.2491) [Gemma2_27B]
        (5.120,  0.2049) [L4Mav]
        (7.168,  0.1762) [Kimi_K25]
        (7.168,  0.1746) [DSV3]
        (8.192,  0.2231) [L31_70B]
        (12.288, 0.1174) [GPT3]
        (16.384, 0.1002) [L31_405B]
    };

\addlegendimage{only marks, mark=*, black}
\addlegendentry{LLaMA-3.2 1B}
\addlegendimage{only marks, mark=triangle*, black}
\addlegendentry{Gemma-2-27B}
\addlegendimage{only marks, mark=star, black}
\addlegendentry{LLaMA-4 Maverick 400B}
\addlegendimage{only marks, mark=square*, black}
\addlegendentry{Kimi-K2.5}
\addlegendimage{only marks, mark=oplus*, black}
\addlegendentry{DeepSeek-V3}
\addlegendimage{only marks, mark=diamond*, black}
\addlegendentry{LLaMA-3.1 70B}
\addlegendimage{only marks, mark=pentagon*, black}
\addlegendentry{GPT-3 175B}
\addlegendimage{only marks, mark=10-pointed star, black}
\addlegendentry{LLaMA-3.1 405B}

\end{axis}
\end{tikzpicture}
  \caption{Outsourcing efficiency: 
    Measured client runtime vs.\ local
  $\mat{W}\mat{X}$ runtime ($X$ as vector, $l=1$)
  for various model weight dimensions. The mask initialization is batched for 512 invocations and amortized over each X. 
  One-time initialization of W is excluded.}
  \label{fig:client-overhead}
\end{figure}

\subsection{Outsourcing efficiency in $\Pi_{\textsf{Sec-Approx-MatMul}}$}
Efficiency of a matrix multiplication outsourcing
protocol relies on a smaller \emph{ratio} of (1)
trusted client protocol overhead to (2) the cost of 
the client computing $WX$ locally. 
This depends on both the \emph{concrete} size of the matrix dimensions (larger is better), as well as the \emph{ratio} of matrix dimensions (balanced dims are better, given a fixed element count).
For a theoretical assessment of this client efficiency
ratio, we illustrate this 
in~\Cref{fig:theoretical-overhead} in the appendix
for protocol parameters in~\Cref{ex:lwe-lpn-parameterisation} and raw computational work (not run-time).

To \emph{experimentally} 
evaluate the efficacy of client outsourcing for
different model dimensions in practice,
we show the ratio of (1) trusted client protocol runtime to (2) the runtime of
computing $WX$ locally in~\Cref{fig:client-overhead}
with our protocol implementation and hardware setup described in 
\Cref{sec:conf-ai-casestudies}.
Overall, model weight matrices with higher rank
are clearly more amenable to outsourcing with our protocol
than those with lower rank, validating the theoretical scalability of our protocol.
We note an implementation nuance: for a given rank, 
the client $\cli$ efficiency in the
implemented protocol can be observed to be
positively correlated to the matrix output dimension $m$; 
MLP down projection
matrices offer the best outsourcing efficiency in our 
experiments. We partially attribute this artifact to
our implementation and NVIDIA runtime environment, 
where the client masking/unmasking computation
over low-rank and sparse matrix elements is implemented on the slower, 
yet more general CUDA cores,
and the full-rank $WX$ is optimized to run on the faster Tensor path. 
The trusted client implementation is likely more sensitive
to increased input matrix dimensions, for which it must 
perform additional work on the slower CUDA path.

As future model dimensions increase, we anticipate outsourcing
to become significantly more worthwhile, as shown in \Cref{fig:theoretical-overhead}.

\paragraph{Local attention computation}
Efficiency of our protocol rests on a one-time $O(mn)$
weight-masking step (Init W) that is amortized across every
subsequent forward pass reusing the same $m\times n$ matrix.
During prefill this amortization is unavailable for attention:
$Q_i$, $K_i$, and $V_i$ are all per-query activations,
so neither operand of $Q K_i^\top$ nor of $\text{score}\cdot V_i$
is a frozen matrix that can be masked once and reused.
During autoregressive decoding the situation is better.
The
KV cache is built once per sequence and reused for every
generated token, so one could in principle run Init W on
$K_i$ and $V_i$ at cache-build time and amortize it over all
subsequent decode steps.

Even with that amortization, however, the head dimension 
in models today is
too small for outsourcing to pay off in practice.
The client-side overhead of $\Pi_{\textsf{Sec-Approx-MatMul}}$
scales as $O((m+n)\,l)$ against a GPU cost of $O(mnl)$,
so the useful regime is one where rank $r = \min(m,n)$ is large
(Figure~\ref{fig:theoretical-overhead}).
The outsourced weight projections sit comfortably in this
regime with $\min(m,n) = d_\text{model} \geq 4096$,
whereas per-head attention matmuls only have inner dimension
$d_k = d_\text{model}/h \approx 128$ for the model sizes
we consider (e.g.\ LLaMA-3 70B and Qwen2.5-72B both with
$d_\text{model}=8192$, $h=64$).



\section{Secure forward pass
for LLM models}
\label{sec:secure-forward-pass}

\begin{figure*}
\footnotesize
\begin{tcolorbox}[protocolbox={\textsf{Protocol $\Pi_{\textsf{Sec-Linear}}$: Secure Linear Layer.}}]
\textbf{Parties:} Trusted client $\cli$
and untrusted accelerator $\mathcal{G}$.

\textbf{Inputs:} BF16 weight matrix $\mat{W} \in \mathbb{R}^{m \times n}$,
optional bias $\vec{b} \in \mathbb{R}^{m}$, and on each invocation a
BF16 activation $\mat{X} \in \mathbb{R}^{n \times l}$.

\textbf{Output:} BF16 activation $\mat{Y} \approx \mat{W}\mat{X} + \vec{b}$.

\textbf{Init W:} One-time pre-processing of floating-point weights (reused across inputs).
\begin{enumerate}
    \item $\cli$ samples and stores
        $\mat{R} = \mathsf{WHT}_n \cdot \mathrm{diag}(\vec{d})$
        with $\vec{d}\in\{-1,+1\}^{n}$.
    \item $\cli$ rotates the weights
        $\widetilde{\mat{W}} \leftarrow \mat{W}\mat{R}^{\top}$
        and quantizes each row of $\widetilde{\mat{W}}$ to a signed
        $b$-bit integer with per-row scale
        $s_{w_j} = \norm{\widetilde{\vec{w}}_j}_{\infty} / q_{\max}$,
        where the operand window $q_{\max}$ is chosen so that the
        inner-product accumulator stays inside $\rng$.
    \item $\cli$ embeds $\hat{\mat{W}}$ into $\rng$
        and invokes the \emph{Init W}
        phase of $\Pi_{\textsf{Sec-Approx-MatMul}}$ to
        produce the masked weights $\hat{\mat{W}}'$ stored on
        $\gpu$ and the unmasking material stored on $\cli$.
\end{enumerate}

\textbf{Online phase:} On fresh floating point input $\mat{X}$, outsource
$\mat{W}\mat{X} + \vec{b}$ to $\gpu$ and return a floating point result.
\begin{enumerate}
    \item $\cli$ rotates the activation
        $\widetilde{\mat{X}} \leftarrow \mat{R}\mat{X}$ and quantizes
        each column of $\widetilde{\mat{X}}$ to a signed $b$-bit integer
        with per-column scale
        $s_{x_k} = \norm{\widetilde{\vec{x}}_k}_{\infty} / q_{\max}$,
        where $q_{\max}$ is chosen to ensure inner-product accumulator stays within $\rng$.
    \item $\cli$ and $\gpu$ run the \emph{Online} phase of
        $\Pi_{\textsf{Sec-Approx-MatMul}}$ on $(\hat{\mat{W}}, \hat{\mat{X}})$
        over the integer ring~$\rng$.
    \item $\cli$  recovers
        $\hat{\mat{Y}} \approx \hat{\mat{W}}\hat{\mat{X}} \in \rngMat{m}{l}$ up to the protocol noise of~\Cref{eq:approx-error}, and
        dequantizes entry-wise back to FP from
        $Y_{j,k} \leftarrow s_{w_j} s_{x_k}  \hat{Y}_{j,k}$,
        and adds the bias to output $\mat{Y} \leftarrow \mat{Y} + \vec{b}$.
\end{enumerate}

\end{tcolorbox}
\caption{Secure linear layer with input/outputs in the floating-point domain.}
\label{proto:secure-linear}
\end{figure*}

Our cryptographic protocol $\Pi_{\textsf{Sec-Approx-MatMul}}$ 
works over the integer ring domain. In order to leverage this to 
realize a secure forward pass over an entire model, we 
must instantiate a secure MatMul outsourcing at each linear layer,
and bridge the floating point domain to the inner, integer ring domain.

A common approach in cryptographic protocols (\Cref{sec:related})
is to cast to the fixed point domain for the secure linear computation. 
These works do not consider large modern architectures, 
for which fixed-point quantization results in destructive error propagation 
across the many sequential model layers (for 70B models: $\approx$ 80 blocks, each with 4 linear layers), as illustrated in~\Cref{fig:layer-cosine-sim-qwen72b-nohad}.
We show this arises from a particular sensitivity of fixed-point arithmetic 
to large floating-point ranges.

Thus, we adapt 
an approach from model quantization literature~\cite{ashkboos2024quarot,liu2024spinquant,wu2025polarquant}, 
that applies a random Hadamard rotation to activation and
model weights, which only induces a $O(n \log(n))$ runtime cost. 
We empirically demonstrate that this approach
results in non-destructive error propagation across large (70B)
models, and offer accuracy that is comparable with 
standard 4-bit and 8-bit quantization schemes.
Our resulting secure-linear protocol $\Pi_{\textsf{Sec-Linear}}$ is shown in~\Cref{proto:secure-linear}.

\subsection{Secure forward pass from $\Pi_{\textsf{Sec-Linear}}$}

The secure forward pass replaces each
linear layer with an invocation of $\Pi_{\textsf{Sec-Linear}}$ and evaluates every non-linear layer
locally on the trusted client.
A full secure forward pass on the model can be described as follows.
\begin{enumerate}
    \item $\cli$ holds the BF16 activation $\mat{X}_{1}$.
    \item For $\ell = 1, \ldots, L$: $\cli$ and $\gpu$ jointly invoke 
    \[
    \mat{Y}_{\ell} \leftarrow
    \Pi_{\textsf{Sec-Linear-Layer}}(\mat{W}_{\ell}, \mat{X}_{\ell})
    \]
    returning a BF16 approximation of
    $\mat{W}_{\ell}\mat{X}_{\ell} + \vec{b}_{\ell}$ to $\cli$;
    
    $\cli$ then computes
    non-linear $\mat{X}_{\ell+1} \leftarrow \phi_{\ell}(\mat{Y}_{\ell})$ locally.
    \item $\cli$ outputs activation $\mat{X}_{L+1}$ in BF16.
\end{enumerate}

Security and client efficiency follows by sequential composition of the per-layer protocol - we 
detail asymptotic efficiency in~\Cref{sec:forward-pass-complexity}
for a single secure forward-pass.
The per-invocation error was previously
bounded in~\Cref{thm:total-error-bound},
but is naturally dependent on distributions of
actual model weights and activations.

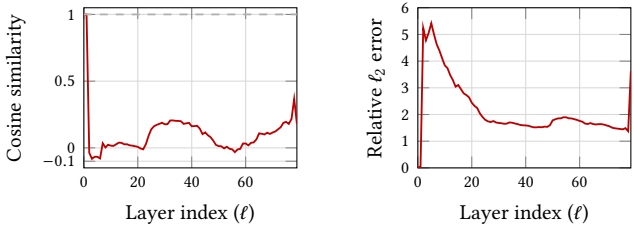
\begin{figure}
  \centering
  \begin{tikzpicture}
\begin{groupplot}[
    group style={group size=2 by 1, horizontal sep=1.6cm},
    xlabel={Layer index ($\ell$)},
    grid=major,
    grid style={gray!30},
    width=4.4cm,
    height=3.7cm,
    every axis plot/.append style={line width=0.7pt},
    xmin=0, xmax=79,
    tick label style={font=\scriptsize},
    label style={font=\small},
]

\nextgroupplot[
    ylabel={Cosine similarity},
    ymin=-0.15, ymax=1.05,
    ytick={-0.1,0.0,0.5,1.0},
]
\addplot[red!75!black, mark=none] coordinates {
    (0,0.9999) (1,0.9990) (2,-0.0364) (3,-0.0819) (4,-0.0669) (5,-0.0666)
    (6,-0.0793) (7,0.0341) (8,0.0031) (9,0.0235) (10,0.0164) (11,0.0148)
    (12,0.0272) (13,0.0393) (14,0.0374) (15,0.0265) (16,0.0271) (17,0.0190)
    (18,0.0165) (19,0.0120) (20,0.0079) (21,-0.0046) (22,-0.0109) (23,0.0277)
    (24,0.0899) (25,0.1392) (26,0.1616) (27,0.1724) (28,0.1797) (29,0.1881)
    (30,0.1770) (31,0.1774) (32,0.2051) (33,0.2057) (34,0.2025) (35,0.2018)
    (36,0.1996) (37,0.1796) (38,0.1840) (39,0.1870) (40,0.1619) (41,0.1633)
    (42,0.1656) (43,0.1442) (44,0.1032) (45,0.1159) (46,0.0945) (47,0.0799)
    (48,0.0537) (49,0.0216) (50,0.0135) (51,0.0117) (52,-0.0039) (53,0.0085)
    (54,-0.0066) (55,-0.0111) (56,-0.0321) (57,-0.0121) (58,-0.0072) (59,0.0326)
    (60,0.0329) (61,0.0188) (62,0.0378) (63,0.0411) (64,0.0641) (65,0.1087)
    (66,0.1055) (67,0.1005) (68,0.1136) (69,0.1043) (70,0.1150) (71,0.1231)
    (72,0.1386) (73,0.1551) (74,0.1875) (75,0.1949) (76,0.1795) (77,0.2159)
    (78,0.3666) (79,0.1813)
};
\addplot[gray!60, dashed, forget plot] coordinates {(0,1) (79,1)};

\nextgroupplot[
    ylabel={Relative $\ell_2$ error},
    ymin=0, ymax=6,
    ytick={0,1,2,3,4,5,6},
]
\addplot[red!75!black, mark=none] coordinates {
    (0,0.0151) (1,0.0458) (2,5.2367) (3,4.7783) (4,5.0560) (5,5.4145)
    (6,4.9999) (7,4.6295) (8,4.3956) (9,4.1113) (10,3.8402) (11,3.7365)
    (12,3.4766) (13,3.2893) (14,3.0429) (15,3.1046) (16,2.9477) (17,2.7915)
    (18,2.7286) (19,2.6419) (20,2.4430) (21,2.3223) (22,2.2443) (23,2.0436)
    (24,1.9068) (25,1.7999) (26,1.7441) (27,1.7174) (28,1.7548) (29,1.7013)
    (30,1.6828) (31,1.6734) (32,1.6515) (33,1.6589) (34,1.7034) (35,1.6909)
    (36,1.6625) (37,1.6408) (38,1.6082) (39,1.5949) (40,1.5909) (41,1.5789)
    (42,1.5450) (43,1.5187) (44,1.5198) (45,1.5288) (46,1.5224) (47,1.5455)
    (48,1.5362) (49,1.5998) (50,1.7811) (51,1.8213) (52,1.8555) (53,1.8572)
    (54,1.8972) (55,1.8947) (56,1.8574) (57,1.8529) (58,1.8247) (59,1.7977)
    (60,1.7595) (61,1.7302) (62,1.6607) (63,1.6438) (64,1.6779) (65,1.6427)
    (66,1.6265) (67,1.6392) (68,1.6392) (69,1.6226) (70,1.5896) (71,1.5619)
    (72,1.5087) (73,1.4847) (74,1.4709) (75,1.4585) (76,1.4445) (77,1.4892)
    (78,1.3752) (79,3.6272)
};

\end{groupplot}
\end{tikzpicture}
  \caption{Per-layer error accumulation for Qwen2.5-72B \emph{without}
  Hadamard rotation, INT16 quantization at $\sigma = 0$ (no protocol
  noise). Cosine similarity/relative $\ell_2$ error vs BF16 reference collapses after 2 layers motivating random Hadamard rotation in our protocol.}
  \label{fig:layer-cosine-sim-qwen72b-nohad}
\end{figure}

\paragraph{Straw-man solution: error explodes with outliers}
We show naive fixed-point scaling over 16-bits with 32-bit accumulations
results in catastrophic error in~\Cref{fig:layer-cosine-sim-qwen72b-nohad}
in the early layers. 
For the chosen Qwen2.5-72B model, both cosine similarity with the full-precision, BF16 activation
is shown, as well as the relative $\ell_2$ error to the BF16 reference, 
for 20 forward-pass runs on randomly 
sampled 512-token windows from Wikitext-2.
Clearly, the noise accumulation from fixed-point quantization alone
causes accuracy to collapse in~\Cref{fig:layer-cosine-sim-qwen72b-nohad}.
Whilst the cosine and $\ell_2$ show a mild tendency for 
recovery in the medium layers (RMSNorm, Non-linear activations can dampen error),
this does not suffice for any useful forward pass computation.

We explain this phenomena analytically. 
In $\Pi_{\textsf{Sec-linear-layer}}$ (\Cref{proto:secure-linear}),
each row of $\mat{W}$ (resp.\ each column of $\mat{X}$) is independently
scaled by its peak magnitude before quantization to a $b$-bit integer range
$[-q_{\max}, q_{\max}]$.
Concretely, for a row $\vec{w}$ of $\mat{W}$, the
quantization scale is
$s_{w} = \norm{\vec{w}}_{\infty} / q_{\max}$,
and the quantized representation is
$\hat{\vec{w}} = \mathrm{round}(\vec{w} / s_{w})$.
Activation columns are scaled analogously with
$s_{x} = \norm{\vec{x}}_{\infty} / q_{\max}$.
The window $q_{\max}$ must also be chosen so that the
inner-product accumulator does not wrap around the $\kappa$-bit ring:
each entry $(\hat{\mat{W}}\hat{\mat{X}})_{j,k}$ has magnitude at most
$n\,q_{\max}^{2}$, so we require $n\,q_{\max}^{2} \le 2^{\kappa-1}$.

After the protocol computes the integer-domain result
$\hat{\mat{Y}} = \hat{\mat{W}} \hat{\mat{X}} + \text{noise}$,
dequantization recovers the approximate floating-point output by
rescaling each entry:
$Y_{j,k} = \hat{Y}_{j,k} \cdot s_{w_j} \cdot s_{x_k}$.
The protocol noise, which has fixed variance in the integer domain,
is therefore \emph{amplified} by the product of scales $s_{w_j} \cdot s_{x_k}$.
When either $\vec{w}$ or $\vec{x}$ contains outlier entries,
$\norm{\cdot}_{\infty}$ is large relative to $\norm{\cdot}_{2}$,
inflating the scales and thus the float-domain noise.

In particular, 
activation vectors in transformer models are inherently concentrated:
a small number of entries carry disproportionately large magnitude
while the majority remain close to zero.
This is not an artefact but a structural necessity - nonlinearities
such as softmax concentrate probability mass on a few channels,
ReLU and its variants zero out negative entries.
Activation outliers of $100\times$ the median magnitude
are common in large language models~\cite{dettmers2022llmint8}.



\paragraph{Random Hadamard rotation}
Model quantization research~\cite{ashkboos2024quarot,liu2024spinquant,wu2025polarquant}
has proposed techniques to \emph{rotate} the activation and spread
the activation energy of individual channels more uniformly
across all internal model dimensions. We adopt this technique
for $\Pi_{\textsf{sec-linear-layer}}$, and show this minimizes
the resulting error in the floating point output domain,
even if intermediary computation is emulated in fixed-point.

Let $\mat{H} \in \reals^{n \times n}$ denote the normalized
Walsh--Hadamard matrix ($\mat{H}\mat{H}^{\top} = \mat{I}$) and let
$\mat{D} = \mathrm{diag}(\vec{d})$ where each $d_{i} \in \{+1,-1\}$
is drawn uniformly at random.
Then, we denote the randomized rotation $\mat{R} = \mat{H}\mat{D}$.
Both factors are orthogonal: $\mat{D}^{\top}\mat{D} = \mat{I}$ because
each diagonal entry squares to $+1$, and
$\mat{H}^{\top}\mat{H} = \mat{I}$ by the normalization above.
Their product is therefore also orthogonal,
\begin{equation*}
    \mat{R}^{\top}\mat{R}
    = \mat{D}^{\top}\mat{H}^{\top}\mat{H}\mat{D}
    = \mat{I},
\end{equation*}
which gives us two properties we rely on below.
First, the matrix product is preserved exactly,
\begin{equation}
\label{eq:had-product-preserved}
    \mat{W}\mat{X}
    = (\mat{W}\mat{R}^{\top})(\mat{R}\mat{X})
    = \mat{W}'\mat{X}'
\end{equation}
where $\hat{\mat{W}} = \mat{W}\mat{R}^{\top}$ and $\hat{\mat{X}} = \mat{R}\mat{X}$,
so applying the secure matrix multiplication protocol to
$\hat{\mat{W}}$ and $\hat{\mat{X}}$ recovers the same output as the original product.
Second, the rotation preserves $\ell_2$ norms
($\norm{\mat{R}\vec{w}}_2 = \norm{\vec{w}}_2$), so only the $\ell_\infty$
norm is changed by the rotation.
The rotation can be applied in $O(n \log n)$ via the fast Walsh--Hadamard
transform~\cite{fino1976unified}.

We illustrate in~\Cref{sec:error-post-rot} 
how this enables the floating-point domain 
error variance induced by our protocol (\Cref{eq:approx-error}) 
to be reduced by factor 
$(\norm{\vec{w}_{j}}_{2}^{2} \log n)/(\norm{\vec{w}_{j}}_{\infty}^{2} n)$.

In~\Cref{sec:error-single-layer}
and~\Cref{sec:error-propagation}
we empirically investigate 
the error contribution by a single securely outsourced linear layer and the
error accumulation over all secure linear layers in selected models respectively.

\subsection{Error from a single secure linear layer}
\label{sec:error-single-layer}

\begin{table}
  \centering
  \footnotesize
  \resizebox{\linewidth}{!}{%
  \begin{tabular}{@{}l ccc ccc ccc@{}}
    \toprule
    & \multicolumn{3}{c}{Qwen2.5-32B}
    & \multicolumn{3}{c}{Qwen2.5-72B}
    & \multicolumn{3}{c}{LLaMA-3-70B} \\
    \cmidrule(lr){2-4} \cmidrule(lr){5-7} \cmidrule(lr){8-10}
    Configuration
        & mean & med. & max & mean & med. & max & mean & med. & max \\
    \midrule
    INT16-\textsf{rot}-$\sigma=0.0$
        & 0.47 & 0.50 & 0.96 & 0.48 & 0.44 & 1.65 & 0.53 & 0.47 & 46.4 \\
    INT16-\textsf{rot}-$\sigma=0.5$
        & 0.66 & 0.66 & 1.48 & 0.77 & 0.65 & 2.42 & 0.85 & 0.71 & 78.7 \\
    INT16-\textsf{rot}-$\sigma=1.0$
        & 0.95 & 0.91 & 2.48 & 1.20 & 0.98 & 3.55 & 1.32 & 1.07 & 122.5 \\
    INT16-\textsf{rot}-$\sigma=1.5$
        & 1.29 & 1.22 & 3.55 & 1.68 & 1.36 & 4.97 & 1.85 & 1.49 & 179.7 \\
    INT16-\textsf{rot}-$\sigma=2.0$
        & 1.65 & 1.56 & 4.66 & 2.18 & 1.76 & 6.39 & 2.40 & 1.93 & 232.4 \\
    \bottomrule
  \end{tabular}%
  }
  \caption{\revised{Error from a \emph{single} outsourced matrix multiplication
    under $\Pi_{\textsf{Sec-Linear}}$, in relative $\ell_2$ error
    $\lVert \widetilde{Y}-Y\rVert_2 / \lVert Y\rVert_2$ (\%) of one
    outsourced instance of $Y=\mat{W}\mat{X} + b$ 
    evaluated on \emph{clean} BF16 inputs $\mat{X}$.     
    We report the mean, median, and maximum over 
    all model linear projections (attention/MLP/LM head).} 
    }
  \label{tab:local-matmul-error}
\end{table}

\begin{revisedblock}
We isolate the error introduced by a \emph{single} outsourced
matrix multiplication in $\Pi_{\textsf{Sec-Linear}}$,
independent of how error propagates through the
network. For each weight matrix $\mat{W}$ we take the reference activation
$\mat{X}$ from a clean BF16 forward pass (no MOSAIC), compute the outsourced product
$\widetilde{Y}$ under $\Pi_{\textsf{Sec-Linear}}$, and compare it to
the exact product $Y=\mat{W}\mat{X}$ via the relative $\ell_2$ error
$\lVert \widetilde{Y}-Y\rVert_2 / \lVert Y\rVert_2$. Because the input is
the clean reference in every case, this measures per-multiplication error
in isolation, with no accumulation across layers.

\Cref{tab:local-matmul-error} reports the mean, median, and maximum of
this error across all weight projections of
Qwen2.5-32B and -70B as well as LLaMA-370B class models at
increasing protocol noise $\sigma$. A single outsourced multiplication and additive bias
is highly accurate: the typical (median) error is $1$--$2\%$ across all
models, and for both Qwen models even the worst-case projection stays
below $6.4\%$ at even $\sigma=2.0$. The large LLaMA-3 maxima are concentrated
in the early-layer value projections (the layer-$0$ \textsf{v\_proj}
alone accounts for the reported maximum); excluding the \textsf{v\_proj}
family, the LLaMA-3 worst case is below $4\%$ for $\sigma\le 1.0$. At this noise level, these
outliers do not propagate and
end-to-end accuracy is
ultimately preserved (\Cref{tab:model-accuracy}); only at $\sigma=2.0$ does LLaMA-3
degrade. Still,
LLaMA-3-70B is known to be very sensitive to
similar noise introduced by standard model weight quantization~\cite{qin2024uniqueness}.

In~\Cref{sec:error-propagation}, we illustrate that model accuracy is well-preserved,
even when per-layer error contributions propagate and accumulate across the entire model.

\end{revisedblock}



\subsection{Error accumulation across layers}
\label{sec:error-propagation}

\begin{figure*}
  \centering
  \input{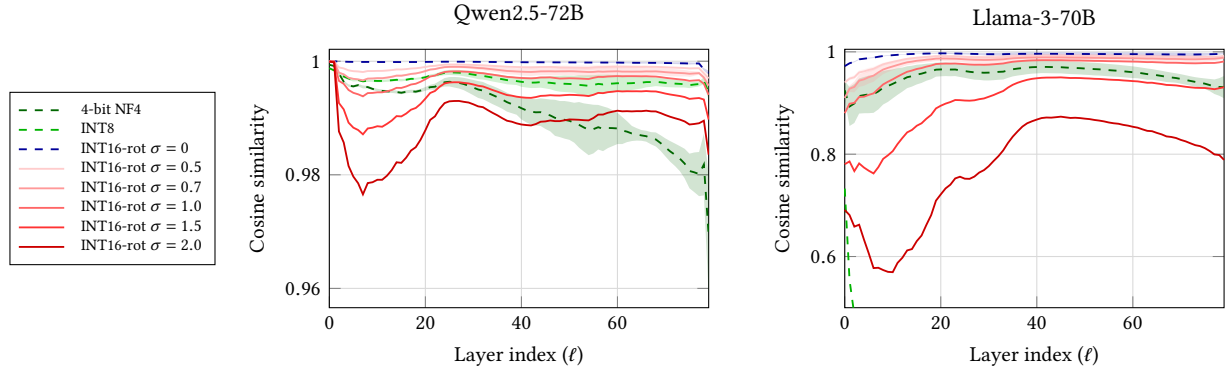}
  \caption{Per-layer error accumulation: 
    Consine similarity between activations of quantized and reference
    forward-pass (BF16) for 70B-class models (left: Qwen2.5-72B;
    right: Llama-3-70B).
    Whilst the two models share similar architectures, Qwen2.5-70B is more resilient against error accumulation. 
    Still, for protocol noise std-deviation range $\sigma \in [0.5,1.0]$, 
    our protocol can still
    be parametized for 140-bit security whilst resulting in well-formed error accumulation in both models.}
  \label{fig:layer-cosine-sim-qwen-llama3}
\end{figure*}

Formal end-to-end error accumulation analysis
requires quantifying 
the Lipschitz constant of each transformer sub-layer,
which determines how much a perturbation at one layer
is amplified before reaching the next.
Scaman and Virmaux~\cite{scaman2018lipschitz}
show that computing the exact Lipschitz constant of even a two-layer ReLU network is NP-hard
and the situation for transformer layers is strictly harder.
Thus, we rely on empirical study of error propagation in models, and compare accumulated error induced by many layers of $\Pi_{\textsf{Sec-Linear-Layer}}$ with a baseline
of commonly used model quantization techniques intended for efficiency rather than security.
We show comparable model accuracy for
protocol noise necessary for 140-bit protocol security,
as parameterized in~\Cref{ex:lwe-lpn-parameterisation}.

\paragraph{Quantization baselines (4-bit NF4, INT8, INT16-\textsf{rot}-$\sigma$).}
We compare our protocol against two standard post-training
quantization schemes as reference points.
\emph{4-bit NF4}~\cite{dettmers2023qlora} uses a non-uniform 4-bit
\emph{NormalFloat} grid whose levels track the quantiles of a
standard normal distribution, exploiting the approximately Gaussian
distribution of transformer weights.
\emph{LLM.int8()}~\cite{dettmers2022llmint8} represents weights and
activations as per-channel-scaled 8-bit integers in $[-127, 127]$,
with matrix multiplications accumulated in higher precision, as well as weight and activation
outliers.

We refer to (1) INT4, (2) INT8 and (3) INT16-$\textsf{rot}$-$\sigma$ 
as running model inference with 4-bit NF4, LLM.int8() 
and our protocol parameterized with gaussian stdev parameter $\sigma$
that induces a random hadamard rotation on activation and model weights.
Note that for $\sigma = 0$, our protocol converges to 
the same precision as 16-bit fixed-point with the additional
hadamard rotation treatment, with the latter strictly improving
the precision of the computation. Thus, we use 
the INT16+$\sigma$=0 baseline to refer to the \emph{optimistic} precision obtainable by prior work~(\Cref{sec:related})
that computes over integer rings or fields in cryptographic settings. 

\paragraph{Empirical error accumulation.}
We measure how per-layer error propagates and accumulates 
through the full forward
pass of 
Qwen2.5-32B with 64 model layers (\Cref{fig:layer-cosine-sim-qwen32b}),
and several 80-layer models, namely
LLaMA-3-70B (\Cref{fig:layer-cosine-sim-llama3-70b}), Qwen2.5-72B (\Cref{fig:layer-cosine-sim-qwen72b}), and DeepSeek-R1-Distill-LLaMA-70B (\Cref{fig:layer-cosine-sim-deepseek-r1-70b}),
with different sensitivity to quantization error.
We report well-formed, 
non-destructive error accumulation across
the forward pass of models for our protocol (INT16+$\sigma$) for $\sigma < 1.5$, 
with final layer error comparable or improving upon INT4 and INT8 quantization (without secure outsourcing).

For each model, we run inference passes on 20 WikiText-2 contexts
(512 tokens each): a \emph{clean} reference at BF16, and a
\emph{perturbed} pass in which weights and activations are
INT16-quantized, rotated via a random Hadamard matrix
(\cref{eq:had-product-preserved}), and sub-Gaussian noise $\mat{E}_{w}, \mat{E}_{x}$ at standard deviation $\sigma$ is
sampled and induced by our protocol as described in~\Cref{eq:approx-error}.
After each transformer model layer $\ell$ we capture the
\emph{residual stream}~$\vec{x}_{\ell}$,
the $d_\textsf{model}$-dimensional
hidden state defined by~\cref{eq:transformer-layer} that is both the
output of block~$\ell$ and the input to block~$\ell+1$, from both \emph{clean} and \emph{perturbed} passes.
We report, 
averaged over all token positions, the per-layer cosine similarity
\begin{equation}
    \mathrm{cos}(\ell)
    \;=\;
    \frac{\langle \vec{x}_{\ell}^{\mathrm{bf16}},\vec{x}_{\ell}^{\mathrm{int16-\textsf{rot}-\sigma}}\rangle}
         {\norm{\vec{x}_{\ell}^{\mathrm{bf16}}}_{2}
          \norm{\vec{x}_{\ell}^{\mathrm{int16-\textsf{rot}-\sigma}}}_{2}}
\end{equation}
and relative $\ell_2$ error
$\norm{\vec{x}_{\ell}^{\mathrm{int16-\textsf{rot}-\sigma}} - \vec{x}_{\ell}^{\mathrm{bf16}}}_2
/ \norm{\vec{x}_{\ell}^{\mathrm{bf16}}}_2$.
Each figure shows error accumulation for protocol noise 
std-deviation
$\sigma \in \{0.5, 0.7, 1.0, 1.5, 2.0\}$ together with three
post-training quantization baselines (NF4, INT8) 
run through
the identical measurement pipeline without rotation or injected noise.
At moderate $\sigma \leq 1.0$ the protocol curves track INT8 and
substantially outperform NF4 throughout the forward pass, confirming
that the per-layer protocol error accumulates non-destructively
across depth. 



\Cref{fig:layer-cosine-sim-qwen-llama3} in the main body
highlights differences in measured cosine similarity
between Qwen2.5-72B and Llama3-70B.
We highlight error accumulation behaviors that
differ between models below. 

\paragraph{Quantization-friendly models} 
The cosine similarity and relative $\ell_2$ error
propagation for chosen models Qwen2.5-72B  
and Qwen2.5-32B
are well-behaved, as illustrated in 
\Cref{fig:layer-cosine-sim-qwen72b}
and
\Cref{fig:layer-cosine-sim-qwen32b}
respectively in the Appendix.
All INT16-$\textsf{rot}$-$\sigma$ runs show error propagation that
match or is less than INT4 for all $\sigma \leq 2.0$,
and improve on INT8 for $\sigma \leq 1.0$.
The earlier layers in Qwen2.5-72B show a higher
sensitivity to our protocol error, but this effect is
non-destructive, as it recovers in both 
cosine-similarity and $\ell_2$ errors in later model 
layers. We attribute this heuristically
to both RMSNorm and channel-wise activation functions, 
that are known to contribute well-formed error accumulation.
For the chosen protocol parameter $\sigma = 0.5$
(\Cref{ex:lwe-lpn-parameterisation}), this error accumulation
is strictly better than INT4 or INT8 quantization schemes
across all layers.

\paragraph{Quantization unfriendly models} 
Llama-3-70B (and derivative DeepSeek-R1-Distill-LLaMA-70B)
are known to be sensitive to 
quantization~\cite{qin2024uniqueness}. 
Their cosine similarity and relative $\ell_2$
error accumulation is shown in \Cref{fig:layer-cosine-sim-llama3-70b} and \Cref{fig:layer-cosine-sim-deepseek-r1-70b}
respectively.
Indeed, in our
experiments, INT8 and INT16-$\textsf{rot}$-$\sigma=2.0$
show destructive error propagation across layers
for both Llama-3-70B and a derivative model DeepSeek-R1-Distill-LLaMA-70B.
The relative $\ell_2$ error is generally large
for both INT4 and INT16-$\textsf{rot}$-$\sigma$ for $\sigma \geq1.0$. 
For Llama-3-70B (\Cref{fig:layer-cosine-sim-llama3-70b}), we propose to parameterize our
protocol with gaussian samples with
parameter $\sigma = 0.5$ for 140-bit security (\Cref{ex:lwe-lpn-parameterisation}). 
For DeepSeek-R1-Distill-Llama-70B (\Cref{fig:layer-cosine-sim-deepseek-r1-70b}), 
the relative $\ell_2$ error
at the final layer remains very high for all quantizations, making the model a border-line or even suboptimal candidate
for quantized computation.

\subsection{Accuracy of outsourced inference}
\label{sec:model-accuracy}

\Cref{tab:model-accuracy} reports model accuracy of an end-to-end
forward pass under our protocol against two reference quantization
baselines (LLM.int8() and 4-bit NF4), on two 70B-class and one 32B model
spanning the quantization spectrum: Qwen2.5-72B and Qwen2.5-32B
(quantization-friendly) and LLaMA-3-70B (quant.-unfriendly).
Each model is evaluated on WikiText-2 perplexity (PPL, lower is
better) and on HumanEval pass@1 (HE, higher is better, $k=1$).

For the quantization-friendly Qwen2.5-72B,
all five INT16-\textsf{rot}-$\sigma$ 
configurations stay within $2\%$ of the
BF16 reference on PPL across the full $\sigma \in [0, 2]$ range, and
HE on Qwen2.5-72B tracks the BF16 baseline to within $\pm 4$ points across
$\sigma \in [0, 2]$. For Qwen2.5-32B, a potential drop in 
HE scores may be observed for  $\sigma =0.5/1.0$,
that recovers at higher noise levels.

For quantization-unfriendly LLaMA-3-70B,
our protocol matches PPL scores of BF16
to within $3\%$ for $\sigma \le 1.0$
and HE drops gracefully from $54.6$ (BF16) to $45.7$ at $\sigma = 1.0$;
beyond $\sigma = 1.0$ both metrics degrade visibly - PPL grows to
$2.055$ at $\sigma = 1.5$ and $3.233$ at $\sigma = 2.0$.
Crucially, the noise parameter required for our $140$-bit-secure
parameterisation in~\Cref{ex:lwe-lpn-parameterisation}
sits at
$\sigma = 0.5$; for every model in the table
there exists a secure parameterization of our protocol 
which performs at or above NF4/INT8 baselines.

The comparison with INT8 on LLaMA-3-70B is informative as a
\emph{challenging} reference model: INT8
fails catastrophically on this model (PPL $= 425.97$, HE $= 10.06$) - a sensitivity our protocol does not exhibit
because the random Hadamard rotation
(\Cref{eq:had-product-preserved}) flattens the same outlier
distribution that defeats LLM.int8() quantization. 
For $\sigma = 0.5$ (140-bit security), our protocol retains the accuracy of BF16
on both PPL and HE benchmarks.
\begin{table}
  \centering
  \footnotesize
  \resizebox{\linewidth}{!}{%
  \begin{tabular}{@{}l cc cc cc@{}}
    \toprule
    & \multicolumn{2}{c}{Qwen2.5-72B}
    & \multicolumn{2}{c}{LLaMA-3-70B}
    & \multicolumn{2}{c}{Qwen2.5-32B} \\
    \cmidrule(lr){2-3} \cmidrule(lr){4-5} \cmidrule(lr){6-7}
    Configuration
        & PPL & HE & PPL & HE & PPL & HE \\
    \midrule
    BF16 (reference)
        & 2.201 & 53.96 &  1.821 & 54.57 & 2.722 & 46.34 \\
    \midrule
    INT16-\textsf{rot}-$\sigma=0.0$
        & 2.201 & 54.57 &  1.827 & 51.83 & 2.723 & 46.04 \\
    INT16-\textsf{rot}-$\sigma=0.5$
        & 2.204 & 53.66 &  1.834 & 53.96 & 2.724 & 42.38 \\
    INT16-\textsf{rot}-$\sigma=1.0$
        & 2.205 & 50.91 &  1.864 & 45.73 & 2.726 & 43.29 \\
    INT16-\textsf{rot}-$\sigma=1.5$
        & 2.221 & 51.83 &  2.055 & 35.98 & 2.734 & 43.90 \\
    INT16-\textsf{rot}-$\sigma=2.0$
        & 2.244 & 53.35 &  3.233 & --    & 2.742 & 47.56 \\
    \midrule
    INT8 (LLM.int8())
        & 2.213 & 49.70 & 425.975 & 10.06 & 2.736 & 48.48 \\
    NF4 (4-bit)
        & 2.303 & 55.79 &  2.247 & 45.43 & 2.894 & 47.56 \\
    \bottomrule
  \end{tabular}%
  }
  \caption{Model accuracy under quantization baselines (INT8, NF4)
    and the protocol's INT16-rot-$\sigma$ noise model at increasing
    noise levels. Two metrics per model: WikiText-2 perplexity
    (PPL, lower is better, sequence length 2048, stride 512) and
    HumanEval pass@1 (HE, higher is better, $k=1$).}
  \label{tab:model-accuracy}
\end{table}

\section{Confidential AI outsourcing}
\label{sec:conf-ai-casestudies}

\begin{revisedblock}
Here, we show how a
small trusted computing base (TCB) tasked with 
AI computations in a modern AI data-center
can support larger AI computations that exceed
its trusted resources by adding powerful, untrusted GPU accelerators. In this setting, 
we demonstrate that MOSAIC is not bottle-necked by communication in the modern AI data center setting where individual GPU's are networked with fast 
interconnect technologies.

In~\Cref{sec:application-remote} we 
additionally illustrate a remote application over 
public networks,
which only requires an LLM classification, 
and does not require communication-intensive
autoregressive decoding.




\end{revisedblock}

\subsection{Secure Outsourcing in an AI Datacenter}

\begin{revisedblock}
The common approach for confidential computation 
in datacenters
is to deploy applications on a trusted computing base (TCB),
realized by confidential computing technologies offered by all major cloud providers today (SGX/TDX/NCC).
Whilst confidential hardware architectures now extend to
GPU~\cite{gpu2025nvidia} hardware,
they are limited in scope, and do 
not offer the same scaling and performance benefits as modern datacenter architectures optimized for AI inference. Confidential computing is also more expensive due to specialized chips, dedicated hardware resources, attestation key management, enhanced physical security, and other factors. 

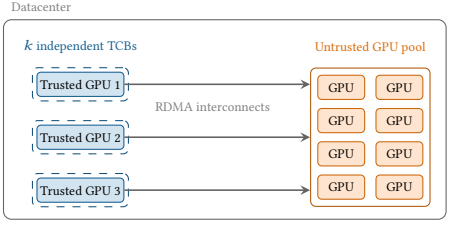
\begin{figure}[h]
\centering
\begin{tikzpicture}[
    >=stealth,
    gpu/.style={draw, rounded corners=1.5pt, minimum width=1.15cm,
        minimum height=0.32cm, font=\tiny, align=center, inner sep=1pt},
    trustedgpu/.style={gpu, fill=cTrusted!20, draw=cTrusted!80!black},
    poolgpu/.style={gpu, minimum width=0.62cm, fill=cUntrusted!20, draw=cUntrusted!80!black},
    tcbbox/.style={draw=cTrusted!70!black, dashed, rounded corners=2pt, inner sep=1.5pt},
    link/.style={->, semithick, draw=cComm!70!black},
]

\node[trustedgpu] (T1) at (0, 0.70) {Trusted GPU 1};
\node[trustedgpu] (T2) at (0, 0.00) {Trusted GPU 2};
\node[trustedgpu] (T3) at (0,-0.70) {Trusted GPU 3};
\begin{scope}[on background layer]
\node[tcbbox, fit=(T1)] (b1) {};
\node[tcbbox, fit=(T2)] (b2) {};
\node[tcbbox, fit=(T3)] (b3) {};
\end{scope}
\node[font=\tiny, cTrusted!80!black, anchor=south, yshift=2pt]
      (tlabel) at (b1.north) {$k$ independent TCBs};

\foreach \r/\y in {1/0.66, 2/0.22, 3/-0.22, 4/-0.66} {
  \node[poolgpu] (L\r) at (3.45,\y) {GPU};
  \node[poolgpu] (R\r) at (4.22,\y) {GPU};
}
\begin{scope}[on background layer]
\node[draw=cUntrusted!70!black, rounded corners=2pt, fit=(L1)(R4), inner sep=2.5pt]
      (pool) {};
\end{scope}
\node[font=\tiny, cUntrusted!85!black, anchor=south, yshift=2pt]
      (plabel) at (pool.north) {Untrusted GPU pool};

\draw[link] (T1.east) -- (T1.east -| pool.west);
\draw[link] (T2.east) -- (T2.east -| pool.west);
\draw[link] (T3.east) -- (T3.east -| pool.west);
\node[font=\tiny, gray, fill=white, inner sep=0.6pt] at (1.75, 0.40) {RDMA interconnects};

\begin{scope}[on background layer]
\node[draw=cComm!70!black, rounded corners=3pt, thin,
      fit=(b1)(b3)(pool)(tlabel)(plabel), inner sep=4.5pt,
      label={[font=\tiny, cComm!90!black, anchor=south west]north west:Datacenter}]
      (dc) {};
\end{scope}

\end{tikzpicture}
\caption{\revised{Datacenter outsourcing topology: $k$ independent trusted clients,
         each a small TCB, elastically share a common pool of untrusted GPUs,
         secure outsourcing AI computations  over fast 
         interconnects.}}
\label{fig:datacenter-topology}
\end{figure}

We emphasize that frontier AI inference needs more than a single GPU cluster. A typical inference task is distributed across phase (prefill/decode), model layer (sharding) and time (interleaving between accelerators) to maximize utilization~\cite{wu2026dualpath,sun2026multi}. Data-center accelerators are connected via remote direct memory access (RDMA)~\cite{spheron2026gpu} technologies, featuring single-digit microsecond latency between GPUs.
Such connectivity is enabled by 
fast intra-rack (e.g. NVLink, PCIe) and intra-data-center interconnects (e.g. Infiniband).

Such distributed, AI-native architectures are not directly compatible with the currently available confidential computing technologies. Offering AI inference at scale would imply extending the TCB to the heterogeneous pool of accelerators and their full networking stack, requiring complex attestation key management, enhanced physical security measures, and so on. Each user of secure AI computation would need their models and requests securely isolated from others, while cloud providers would ideally be free to assign workloads across the distributed data center to efficiently use available resources. Secure outsourcing approaches like MOSAIC offer an alternative approach that combines (a) the benefits of modern (heterogeneous and distributed) AI inference architecture with (b) a small TCB for each user (\Cref{fig:datacenter-topology}).
Confidential (and potentially attested)
compute represents the scarce resource,
and the value of our protocol is that a trusted GPU can
elastically capture additional throughput from untrusted accelerators
\emph{without} expanding the trusted computing base.
\end{revisedblock}


\subsection{Implementation}
\revised{We describe our emulation of 32-bit integer ring arithmetic, the deployment over GPU interconnects, the implemented protocol configuration and techniques for sharding models across GPUs.}

\paragraph{Emulation approach.}
We emulate a real-world implementation with current, limited 
hardware support.
Our $\Pi_{\textsf{Sec-Approx-Matmul}}$ protocol 
is parameterised over a
32-bit integer ring, which modern AI accelerators do not natively
support. Concretely, on Nvidia GPUs the high-throughput Tensor cores
that dominate LLM matmul performance accelerate FP16/BF16, FP8, and
INT8 datatypes, but \emph{not} 32-bit integer arithmetic. Native 32-bit
integer is available only on the general-purpose
CUDA pipeline, which delivers roughly an order of magnitude lower
throughput than the Tensor cores and is therefore impractical for the
matrix sizes involved in 70B-class models.

To remain on the Tensor-core pipeline, we emulate each 32-bit integer
multiply--accumulate as a sum of INT8 multiply--accumulates, splitting
each 32-bit operand into four 8-bit limbs and reassembling the result
in the trusted client. 
This emulation costs a constant overhead per ring
multiplication, a slowdown relative to a
hypothetical Tensor-core kernel with native INT32 support, but still
markedly faster than falling back to the CUDA.

A key detail is that the limbs are unsigned bytes
($[0, 255]$) but the Tensor-core kernel accepts only
signed bytes ($[-128, 127]$). We bridge this by
subtracting $128$ from each limb before the matmul, which slides
every value into signed INT8 range without losing any information.
After the kernel returns, we undo the bias algebraically: expanding
$(a-128)(b-128) = ab - 128(a+b) + 128^{2}$ shows that the unsigned
product $\sum_p a_p b_p$ equals the signed product
$\sum_p (a_p - 128)(b_p - 128)$ plus a correction term that depends
only on the row and column sums 
inducing $O(m+n)$ to apply. The final reduction to
$\mathbb{Z}_{2^{32}}$ requires no extra work either, since native
INT32 addition and bit-shift discard bits above 
position 31 automatically.


The split between the two pipelines is also reflected
in the protocol split between trusted client and
untrusted GPU.
The untrusted GPU performs the dense
$O(mnl)$ masked MatMul and runs strictly on the Tensor-core
pipeline (via the INT8 emulation above). 
The trusted client, by
contrast, performs only low-rank dense and sparse operations of
inner rank $r_d$ or row Hamming weight $t$ (\Cref{sec:lwe-lpn-matmul});
these more general workloads 
cannot benefit
from Tensor-core acceleration, and thus client kernels are implemented for general-purpose CUDA pipeline at INT32 width.

\paragraph{Implementation with fast GPU interconnects.}

\begin{revisedblock}
We run our implementation on 70B models on Nvidia H200 with
141 GB VRAM connected with fast NVLink interconnects,
featuring 900 GB/s GPU-to-GPU transfer
speeds. Whilst similar latency can be achieved with RDMA-style interconnects across the data center in practice~\cite{wu2026dualpath,spheron2026gpu,sun2026multi}, 
on-demand GPUs offered by AI cloud providers 
accessible limit us to such inter-node interconnects.
\james{End of revision.}

For the 70B models we deploy, we use 1 trusted GPU,
that securely outsources computation to 3 untrusted GPUs,
where protocol communication is over the aforementioned
NVLink GPU-to-GPU interconnect. We measure end-to-end
prefill and decode for 70B-class models 
(LLaMA-3, Qwen-2.5) on a deployment with \emph{one} trusted GPU
acting as the trusted client and 3 untrusted GPUs collectively
serving the outsourced MatMul calls
via the secure forward pass
(\Cref{fig:datacenter-runtime-prefill-postccs,fig:datacenter-runtime-decode-postccs}). 

\paragraph{Protocol parameterization.}
We implement protocol parameters from~\Cref{ex:lwe-lpn-parameterisation} for 140-bit security,
but omit the recursive LPN-mask optimization for masking activations. 
This maintains security, but
saves memory for the client, as the layer-specific matrices $J_i$ in step 4 of Init W (\Cref{proto:approx-matmat-gpu})
no longer occupies client memory;
this slightly increases the cost of 
computing the activation mask $U_x = J_d R$ in step 1b of Init X at (model-independent)
rank of 1536 instead of 652.
Exhaustive parameter optimization is left
for future work as it is
highly hardware dependent.

\paragraph{Model sharding.}
We implement two parallelization techniques.
(1) We shard each masked model matrix across the 
3 GPUs, and broadcast each masked activation to all GPUs (model-matrix parallel);
the result is gathered by the trusted GPU thereafter.
(2) We shard the masked model by contiguous 
layer ranges across the untrusted GPUs (model-layer parallel).
\end{revisedblock}

\subsection{Evaluation results}

\begin{figure}[h]
\centering
\begin{tikzpicture}
\begin{axis}[
    width=8.8cm,
    height=5.0cm,
    ybar stacked,
    bar width=5pt,
    clip=false,
    ylabel={Latency (ms / token)},
    xtick={0.3,0.7,1.1,1.7,2.1,2.5,3.1,3.5,3.9,4.9,5.3,5.7,6.3,6.7,7.1,7.7,8.1,8.5},
    xticklabels={B,M,L,B,M,L,B,M,L,B,M,L,B,M,L,B,M,L},
    x tick label style={font=\tiny},
    enlarge x limits=0.03,
    ymin=0, ymax=800,
    ytick={0,200,400,600,800},
    grid=major,
    ymajorgrids=true,
    xmajorgrids=false,
    legend style={
        at={(0.5, 1.02)}, anchor=south,
        legend columns=2,
        font=\scriptsize, draw=none, fill=none,
        column sep=4pt, row sep=-2pt, inner sep=2pt,
    },
    legend cell align=left,
    every axis plot/.append style={thick},
]
\addplot[fill=green!55!black, draw=green!40!black] coordinates {
    (0.3,154.9) (0.7,0.0) (1.1,0.0)
    (1.7,158.1) (2.1,0.0) (2.5,0.0)
    (3.1,150.0) (3.5,0.0) (3.9,0.0)
    (4.9,190.9) (5.3,0.0) (5.7,0.0)
    (6.3,185.5) (6.7,0.0) (7.1,0.0)
    (7.7,197.4) (8.1,0.0) (8.5,0.0)
};
\addlegendentry{Local GPU baseline (no protocol)}
\addplot[fill=blue!70, draw=blue!80!black] coordinates {
    (0.3,0.0) (0.7,36.2) (1.1,37.5)
    (1.7,0.0) (2.1,36.5) (2.5,37.8)
    (3.1,0.0) (3.5,36.5) (3.9,37.7)
    (4.9,0.0) (5.3,40.9) (5.7,42.2)
    (6.3,0.0) (6.7,41.2) (7.1,42.5)
    (7.7,0.0) (8.1,41.2) (8.5,42.4)
};
\addlegendentry{Trusted GPU (mask/unmask + non-linear)}
\addplot[fill=orange!70, draw=orange!80!black] coordinates {
    (0.3,0.0) (0.7,422.0) (1.1,514.2)
    (1.7,0.0) (2.1,422.0) (2.5,514.6)
    (3.1,0.0) (3.5,422.1) (3.9,514.4)
    (4.9,0.0) (5.3,427.7) (5.7,520.5)
    (6.3,0.0) (6.7,428.3) (7.1,520.8)
    (7.7,0.0) (8.1,428.4) (8.5,520.5)
};
\addlegendentry{Untrusted GPU (remote MatMul)}
\addplot[fill=gray!40, draw=gray!70!black, postaction={pattern=north east lines, pattern color=gray!70!black}] coordinates {
    (0.3,0.0) (0.7,32.2) (1.1,21.6)
    (1.7,0.0) (2.1,32.3) (2.5,21.7)
    (3.1,0.0) (3.5,32.2) (3.9,21.7)
    (4.9,0.0) (5.3,32.2) (5.7,21.1)
    (6.3,0.0) (6.7,32.4) (7.1,21.1)
    (7.7,0.0) (8.1,32.7) (8.5,20.8)
};
\addlegendentry{Comm (trusted $\leftrightarrow$ untrusted)}

\node[rotate=90, anchor=west, font=\tiny\bfseries, inner sep=1pt] at (axis cs:0.3,155) {155};
\node[rotate=90, anchor=west, font=\tiny\bfseries, inner sep=1pt] at (axis cs:1.7,158) {158};
\node[rotate=90, anchor=west, font=\tiny\bfseries, inner sep=1pt] at (axis cs:3.1,150) {150};
\node[rotate=90, anchor=west, font=\tiny\bfseries, inner sep=1pt] at (axis cs:4.9,191) {191};
\node[rotate=90, anchor=west, font=\tiny\bfseries, inner sep=1pt] at (axis cs:6.3,186) {186};
\node[rotate=90, anchor=west, font=\tiny\bfseries, inner sep=1pt] at (axis cs:7.7,197) {197};
\node[rotate=90, anchor=west, font=\tiny\bfseries, inner sep=1pt] at (axis cs:0.7,490) {490};
\node[rotate=90, anchor=west, font=\tiny\bfseries, inner sep=1pt] at (axis cs:2.1,491) {491};
\node[rotate=90, anchor=west, font=\tiny\bfseries, inner sep=1pt] at (axis cs:3.5,491) {491};
\node[rotate=90, anchor=west, font=\tiny\bfseries, inner sep=1pt] at (axis cs:5.3,501) {501};
\node[rotate=90, anchor=west, font=\tiny\bfseries, inner sep=1pt] at (axis cs:6.7,502) {502};
\node[rotate=90, anchor=west, font=\tiny\bfseries, inner sep=1pt] at (axis cs:8.1,502) {502};
\node[rotate=90, anchor=west, font=\tiny\bfseries, inner sep=1pt] at (axis cs:1.1,573) {573};
\node[rotate=90, anchor=west, font=\tiny\bfseries, inner sep=1pt] at (axis cs:2.5,574) {574};
\node[rotate=90, anchor=west, font=\tiny\bfseries, inner sep=1pt] at (axis cs:3.9,574) {574};
\node[rotate=90, anchor=west, font=\tiny\bfseries, inner sep=1pt] at (axis cs:5.7,584) {584};
\node[rotate=90, anchor=west, font=\tiny\bfseries, inner sep=1pt] at (axis cs:7.1,584) {584};
\node[rotate=90, anchor=west, font=\tiny\bfseries, inner sep=1pt] at (axis cs:8.5,584) {584};

\node[font=\tiny, anchor=north, yshift=-9pt] at (axis cs:0.7,0) {1k};
\node[font=\tiny, anchor=north, yshift=-9pt] at (axis cs:2.1,0) {4k};
\node[font=\tiny, anchor=north, yshift=-9pt] at (axis cs:3.5,0) {7k};
\node[font=\tiny, anchor=north, yshift=-9pt] at (axis cs:5.3,0) {1k};
\node[font=\tiny, anchor=north, yshift=-9pt] at (axis cs:6.7,0) {4k};
\node[font=\tiny, anchor=north, yshift=-9pt] at (axis cs:8.1,0) {7k};

\node[font=\tiny, anchor=north, yshift=-19pt] at (axis cs:2.1,0) {LLaMA-3-70B};
\node[font=\tiny, anchor=north, yshift=-19pt] at (axis cs:6.7,0) {Qwen2.5-72B};
\draw[dotted, gray!70, thick] (axis cs:4.4,0) -- (axis cs:4.4,800);
\end{axis}
\end{tikzpicture}
\caption{\revised{Per-token decoding latency for
         70B-class models ($G=3$
         untrusted GPUs), across 1k/4k/7k token context lengths.
         Within each context: \textsf{B} = local-GPU baseline (no protocol),
         \textsf{M} = sharding of individual masked weight matrix, and \textsf{L} = sharding of masked model 
         by layer across untrusted GPUs.}}
\label{fig:datacenter-runtime-decode-postccs}
\end{figure}
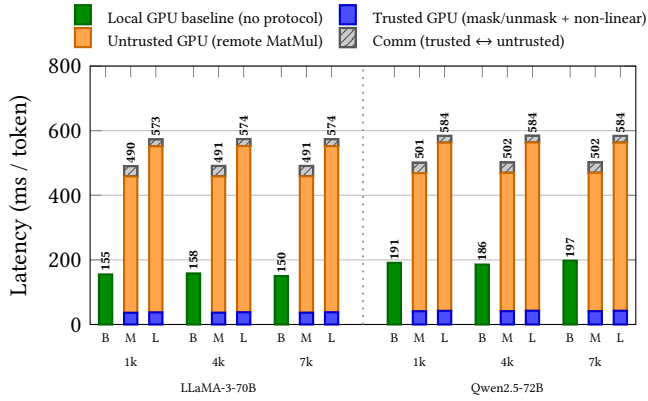
\begin{figure}[h]
\centering
\begin{tikzpicture}
\begin{axis}[
    width=8.8cm,
    height=5.0cm,
    ybar stacked,
    bar width=5pt,
    clip=false,
    ylabel={Latency (seconds)},
    xtick={0.3,0.7,1.1,1.7,2.1,2.5,3.1,3.5,3.9,4.9,5.3,5.7,6.3,6.7,7.1,7.7,8.1,8.5},
    xticklabels={B,M,L,B,M,L,B,M,L,B,M,L,B,M,L,B,M,L},
    x tick label style={font=\tiny},
    enlarge x limits=0.03,
    ymin=0, ymax=45,
    ytick={0,10,20,30,40},
    grid=major,
    ymajorgrids=true,
    xmajorgrids=false,
    legend style={
        at={(0.5, 1.02)}, anchor=south,
        legend columns=2,
        font=\scriptsize, draw=none, fill=none,
        column sep=4pt, row sep=-2pt, inner sep=2pt,
    },
    legend cell align=left,
    every axis plot/.append style={thick},
]
\addplot[fill=green!55!black, draw=green!40!black] coordinates {
    (0.3,1.06) (0.7,0.0) (1.1,0.0)
    (1.7,2.36) (2.1,0.0) (2.5,0.0)
    (3.1,3.46) (3.5,0.0) (3.9,0.0)
    (4.9,1.27) (5.3,0.0) (5.7,0.0)
    (6.3,2.70) (6.7,0.0) (7.1,0.0)
    (7.7,3.92) (8.1,0.0) (8.5,0.0)
};
\addlegendentry{Local GPU baseline (no protocol)}
\addplot[fill=blue!70, draw=blue!80!black] coordinates {
    (0.3,0.0) (0.7,2.28) (1.1,2.11)
    (1.7,0.0) (2.1,8.12) (2.5,7.53)
    (3.1,0.0) (3.5,12.47) (3.9,12.28)
    (4.9,0.0) (5.3,2.90) (5.7,2.42)
    (6.3,0.0) (6.7,8.65) (7.1,8.14)
    (7.7,0.0) (8.1,13.98) (8.5,13.21)
};
\addlegendentry{Trusted GPU (mask/unmask + non-linear)}
\addplot[fill=orange!70, draw=orange!80!black] coordinates {
    (0.3,0.0) (0.7,1.83) (1.1,4.85)
    (1.7,0.0) (2.1,4.80) (2.5,14.86)
    (3.1,0.0) (3.5,7.21) (3.9,22.20)
    (4.9,0.0) (5.3,1.94) (5.7,5.39)
    (6.3,0.0) (6.7,5.31) (7.1,15.75)
    (7.7,0.0) (8.1,7.62) (8.5,22.94)
};
\addlegendentry{Untrusted GPU (remote MatMul)}
\addplot[fill=gray!40, draw=gray!70!black, postaction={pattern=north east lines, pattern color=gray!70!black}] coordinates {
    (0.3,0.0) (0.7,1.03) (1.1,0.96)
    (1.7,0.0) (2.1,1.92) (2.5,1.84)
    (3.1,0.0) (3.5,2.61) (3.9,2.44)
    (4.9,0.0) (5.3,0.90) (5.7,0.82)
    (6.3,0.0) (6.7,2.04) (7.1,1.90)
    (7.7,0.0) (8.1,2.74) (8.5,2.56)
};
\addlegendentry{Comm (trusted $\leftrightarrow$ untrusted)}

\node[rotate=90, anchor=west, font=\tiny\bfseries, inner sep=1pt] at (axis cs:0.3,1.06) {1.1};
\node[rotate=90, anchor=west, font=\tiny\bfseries, inner sep=1pt] at (axis cs:1.7,2.36) {2.4};
\node[rotate=90, anchor=west, font=\tiny\bfseries, inner sep=1pt] at (axis cs:3.1,3.46) {3.5};
\node[rotate=90, anchor=west, font=\tiny\bfseries, inner sep=1pt] at (axis cs:4.9,1.27) {1.3};
\node[rotate=90, anchor=west, font=\tiny\bfseries, inner sep=1pt] at (axis cs:6.3,2.70) {2.7};
\node[rotate=90, anchor=west, font=\tiny\bfseries, inner sep=1pt] at (axis cs:7.7,3.92) {3.9};
\node[rotate=90, anchor=west, font=\tiny\bfseries, inner sep=1pt] at (axis cs:0.7,5.13) {5.1};
\node[rotate=90, anchor=west, font=\tiny\bfseries, inner sep=1pt] at (axis cs:2.1,14.85) {14.8};
\node[rotate=90, anchor=west, font=\tiny\bfseries, inner sep=1pt] at (axis cs:3.5,22.30) {22.3};
\node[rotate=90, anchor=west, font=\tiny\bfseries, inner sep=1pt] at (axis cs:5.3,5.73) {5.7};
\node[rotate=90, anchor=west, font=\tiny\bfseries, inner sep=1pt] at (axis cs:6.7,16.01) {16.0};
\node[rotate=90, anchor=west, font=\tiny\bfseries, inner sep=1pt] at (axis cs:8.1,24.34) {24.3};
\node[rotate=90, anchor=west, font=\tiny\bfseries, inner sep=1pt] at (axis cs:1.1,7.93) {7.9};
\node[rotate=90, anchor=west, font=\tiny\bfseries, inner sep=1pt] at (axis cs:2.5,24.23) {24.2};
\node[rotate=90, anchor=west, font=\tiny\bfseries, inner sep=1pt] at (axis cs:3.9,36.93) {36.9};
\node[rotate=90, anchor=west, font=\tiny\bfseries, inner sep=1pt] at (axis cs:5.7,8.62) {8.6};
\node[rotate=90, anchor=west, font=\tiny\bfseries, inner sep=1pt] at (axis cs:7.1,25.80) {25.8};
\node[rotate=90, anchor=west, font=\tiny\bfseries, inner sep=1pt] at (axis cs:8.5,38.70) {38.7};

\node[font=\tiny, anchor=north, yshift=-9pt] at (axis cs:0.7,0) {1k};
\node[font=\tiny, anchor=north, yshift=-9pt] at (axis cs:2.1,0) {4k};
\node[font=\tiny, anchor=north, yshift=-9pt] at (axis cs:3.5,0) {7k};
\node[font=\tiny, anchor=north, yshift=-9pt] at (axis cs:5.3,0) {1k};
\node[font=\tiny, anchor=north, yshift=-9pt] at (axis cs:6.7,0) {4k};
\node[font=\tiny, anchor=north, yshift=-9pt] at (axis cs:8.1,0) {7k};

\node[font=\tiny, anchor=north, yshift=-19pt] at (axis cs:2.1,0) {LLaMA-3-70B};
\node[font=\tiny, anchor=north, yshift=-19pt] at (axis cs:6.7,0) {Qwen2.5-72B};
\draw[dotted, gray!70, thick] (axis cs:4.4,0) -- (axis cs:4.4,45);
\end{axis}
\end{tikzpicture}
\caption{\revised{System-prompt prefill runtime for 70B-class models
         ($G=3$ untrusted GPUs), across 1k/4k/7k system-prompt lengths,
         split into trusted-GPU walltime, untrusted-GPU remote MatMul, and
         trusted$\leftrightarrow$untrusted communication.
         Within each context: \textsf{B} = local-GPU baseline (no protocol),
         \textsf{M} = weight matrix shard, and \textsf{L} = model sharding by
         layer.}}
\label{fig:datacenter-runtime-prefill-postccs}
\end{figure}
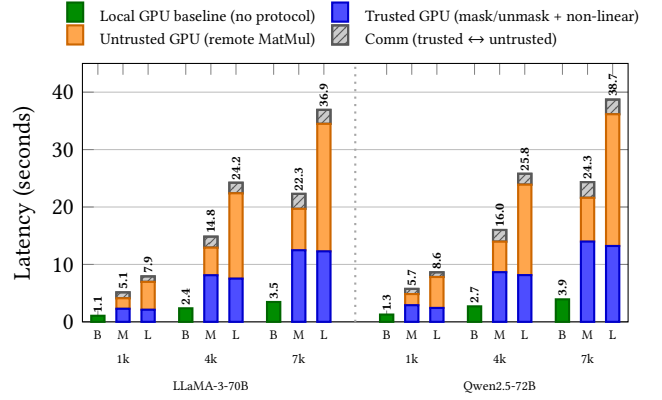
\begin{revisedblock}
MOSAIC is roughly $3\times$ (decode) to $5$--$11\times$ (prefill) slower than
running inference for the same 70B models
on a single local GPU of the same type,
with the higher prefill factor corresponding to layer- rather than
matrix-sharding.
This overhead largely reflects our emulation of 32-bit integer arithmetic
on 8-bit GPU cores (10 INT8 MatMuls per INT32 MatMul; see the emulation
approach above); native 32-bit-integer support in AI accelerator cores
(e.g.\ NVIDIA Tensor) would significantly reduce this gap.
\james{End of added paragraph}

We highlight three observations in prefill 
(\Cref{fig:datacenter-runtime-prefill-postccs})
and decode (\Cref{fig:datacenter-runtime-decode-postccs}) experiments
on 70B class models. 
Firstly, communication between the
trusted and untrusted pools is not the bottleneck under
the available inter-GPU bandwidth: it accounts for only
${\sim}4$--$7\%$ of decode and
${\sim}7$--$16\%$ of
prefill wall-clock, mirroring
practicality of the per-layer interaction pattern as in real-world
settings~\cite{spheron2026gpu,wu2026dualpath}. 

Secondly, the wall-time differences
between model-matrix sharding and
model-layer sharding across untrusted GPU's
is more pronounced in prefill (\Cref{fig:datacenter-runtime-prefill-postccs}) than in autoregressive decoding (\Cref{fig:datacenter-runtime-decode-postccs}).
This is expected as autoregressive decoding is memory-bounded; adding more parallel compute offers
limited benefits.


Third, autoregressive decoding (\Cref{fig:datacenter-runtime-decode-postccs})
is essentially \emph{context-invariant} at ${\sim}0.5$--$0.6$\,s per
generated token across $1$k--$7$k context windows, because per-token
cost is dominated by the dense linear projections (independent of
$s$ for $l = 1$); the KV dimensions only enters through attention, which
are implemented over the native floating-point domain
and negligible at this dimension. Of that per-token latency, the untrusted
GPU accounts for ${\sim}92$--$93\%$ (${\sim}85$--$90\%$ remote MatMul, the
rest communication), leaving under ${\sim}8\%$ on the trusted client. The one-time
init\_W weight-masking cost ($\approx 55$\,s per 70B model) is
amortized across all subsequent queries over the frozen model weights,
quickly becoming a negligible fraction of cumulative wall-clock.

\end{revisedblock}

\section{Conclusion}
\begin{revisedblock}
In this work, we present MOSAIC, a protocol that outsources the linear AI computation to
untrusted accelerators while preserving confidentiality of both model
weights and activations. Its security reduces to standard decisional LWE and LPN
assumptions, and its client cost is asymptotically optimal in the model size. 
We achieve this by permitting the outsourced
computation to be approximate;
this error introduced
by our masking is bounded analytically, and 
after random Hadamard preconditioning it 
leaves end-to-end model accuracy 
at or above standard quantization schemes.

Our end-to-end implementation illustrates 
that the approach can be 
practical in the data-center setting where communication between accelerators 
occurs over fast GPU interconnect technologies.
We think the rapid growth in model sizes will
favor asymptotically optimal protocols of this kind.
The remaining gap is currently dominated by our emulation of 32-bit integer arithmetic on 
natively accelerated 8-bit integer computation;
we hope this work provides motivation for
AI accelerator manufacturers to consider native
support for wider integer arithmetic.
\end{revisedblock}





\bibliographystyle{ACM-Reference-Format}
\bibliography{references}

\appendix 
\crefalias{section}{appendix}
\crefalias{subsection}{appendix}
\section{Generative AI Usage}
Claude Code 4.5-4.8 was used as a coding assistant
for the protocol implementation and for early
drafts of selected paragraphs.
Functional correctness of internal
protocol steps was verified
with standard, manually verified tests.
The full secure forward pass evaluation 
implementing MOSAIC
is manually verified and confirmed 
by asserting convergence
of model evaluation runs against the unaltered Pytorch
library on full floating-pointing reference
precision (BF16) as reported in our benchmarks.


\section{Proofs}
\subsection{Supporting lemmas}
\label{sec:proofs-bounded-noise}

The two supporting lemmas referenced by \Cref{thm:total-error-bound}
are stated and proved below.  \Cref{cor:cross-term-negligible}
collects the comparison between their tail bounds that we invoke in
the proof of \Cref{thm:total-error-bound}.

\begin{lemma}
\label{lem:discrete-noise-inner-product}
    Let $\mat{E} \in \rng^{m \times n}$ have independent entries drawn
    from $\normal{\rng;0}{\sigma^{2}}$, and let $\mat{X} \in \rng^{n \times l}$
    be fixed. Then entry $(j,k)$ of $\mat{E}\mat{X}$ satisfies:
    \begin{enumerate}
        \item $\expect{(\mat{E}\mat{X})_{jk}} = 0$.
        \item $(\mat{E}\mat{X})_{jk}$ is $\sigma\norm{\vec{x}_{k}}_{2}$-sub-Gaussian:
            for all $t > 0$,
            \[
            \Pr\!\left[|(\mat{E}\mat{X})_{jk}| > t\right]
            \leq 2\exp{-t^{2} / (2\sigma^{2}\norm{\vec{x}_{k}}_{2}^{2})}
            \]
    \end{enumerate}
    where $\vec{x}_{k}$ is the $k$-th column of $\mat{X}$.
\end{lemma}
\begin{proof}
    Fix $(j,k)$ and write
    $(\mat{E}\mat{X})_{jk} = \sum_{i=1}^{n} e_i (\vec{x}_{k})_{i}$,
    where $e_i := (\mat{E})_{ji}$ are independent samples from
    $\normal{\rng;0}{\sigma^{2}}$ and $(\vec{x}_{k})_{i}$ are fixed.

    \emph{(1)} By symmetry of $\normal{\rng;0}{\sigma^{2}}$,
    $\expect{e_i} = 0$ for each $i$. By linearity of expectation,
    $\expect{(\mat{E}\mat{X})_{jk}}
     = \sum_{i} (\vec{x}_{k})_{i}\,\expect{e_i} = 0$.

    \emph{(2)}
    A random variable $Z$ is $\alpha$-sub-Gaussian if
    $\expect{\exp{tZ}} \leq \exp{\alpha^{2} t^{2}/2}$ for all
    $t \in \reals$.
    We first show that each discrete gaussian sample $e_i \sim \normal{\rng;0}{\sigma^{2}}$ is
    $\sigma$-sub-Gaussian for any $\sigma > 0$.
    Let $\rho_{\sigma}(S) = \sum_{k \in S} \exp{-k^{2}/(2\sigma^{2})}$
    for any subset $S \subseteq \rng$.
    The MGF of $e_i$ is
    \begin{align}
    \label{eq:mgf-gaussian}
        \expect{\exp{t e_i}}
        &= \sum_{k} \exp{t k} \cdot \Pr[e_i = k] \nonumber \\
        &= \frac{1}{\rho_{\sigma}(\rng)}
          \sum_{k \in \rng} \exp{tk - k^{2}/(2\sigma^{2})}
    \end{align}
    Let us denote $\rho_{\sigma}(\rng - c)
    = \sum_{k \in \rng} \exp{-(k-c)^{2}/(2\sigma^{2})}$.
    Further, manipulating the term in the exponent of~\cref{eq:mgf-gaussian}
    to obtain $tk - k^{2}/(2\sigma^{2})
    = -(k - t\sigma^{2})^{2}/(2\sigma^{2}) + \sigma^{2}t^{2}/2$,
    gives
    \[
        \expect{\exp{t e_i}}
        = \exp{\sigma^{2}t^{2}/2}
          \cdot \frac{\rho_{\sigma}(\rng - t\sigma^{2})}
                     {\rho_{\sigma}(\rng)}
    \]

    By the Poisson summation formula,
    \[
        \rho_{\sigma}(\rng - c)
        = \sigma\sqrt{2\pi}
          \sum_{k \in \rng} \exp{-2\pi^{2}\sigma^{2}k^{2}}
          \cos(2\pi k c)
    \]
    Since $\cos(2\pi k c) \leq 1$ for all $k, c$ and each
    coefficient $\exp{-2\pi^{2}\sigma^{2}k^{2}}$ is non-negative,
    $\rho_{\sigma}(\rng - c) \leq \rho_{\sigma}(\rng)$.
    Thus,
    \[
       1  \leq \frac{\rho_{\sigma}(\rng - t\sigma^{2})}
                 {\rho_{\sigma}(\rng)}
    \]
    and
    $\expect{\exp{t e_i}} \leq \exp{\sigma^{2}t^{2}/2}$,
    confirming that $e_i$ is $\sigma$-sub-Gaussian.

    For independent $\sigma$-sub-Gaussian variables,
    $(\mat{E}\mat{X})_{jk} = \sum_{i} (\vec{x}_{k})_{i}\, e_i$ is
    $(\sum_{i} (\vec{x}_{k})_{i}^{2} \sigma^{2})^{1/2}
    = \sigma\norm{\vec{x}_{k}}_{2}$-sub-Gaussian, since
    \begin{align*}
        \expect{\exp{t \sum_{i} (\vec{x}_{k})_{i} e_i}}
        &= \prod_{i} \expect{\exp{t (\vec{x}_{k})_{i} e_i}} \\
        &\leq \prod_{i} \exp{\sigma^{2} (\vec{x}_{k})_{i}^{2} t^{2}/2} \\
        &= \exp{\sigma^{2} \norm{\vec{x}_{k}}_{2}^{2} t^{2}/2}.
    \end{align*}
    The tail bound follows from the standard Chernoff method:
    for any $t > 0$ and $\lambda > 0$,
    $\Pr[(\mat{E}\mat{X})_{jk} > t]
    \leq \expect{\exp{\lambda (\mat{E}\mat{X})_{jk}}} / \exp{\lambda t}
    \leq \exp{\sigma^{2}\norm{\vec{x}_{k}}_{2}^{2}\lambda^{2}/2 - \lambda t}$.
    Minimizing over $\lambda$ yields
    \[
        \Pr\!\left[|(\mat{E}\mat{X})_{jk}| > t\right]
        \leq 2\exp{-\frac{t^{2}}{2\sigma^{2}\norm{\vec{x}_{k}}_{2}^{2}}}
    \]
\end{proof}

\begin{lemma}
\label{lem:cross-term}
    Let $\mat{E}_{W} \in \rng^{m \times n}$ and
    $\mat{E}_{X} \in \rng^{n \times l}$ have mutually independent
    entries drawn from $\normal{\rng;0}{\sigma^{2}}$.
    Then entry $(j,k)$ of $\mat{E}_{W}\mat{E}_{X}$ satisfies:
    \begin{enumerate}
        \item $\expect{(\mat{E}_{W}\mat{E}_{X})_{jk}} = 0$.
        \item For all $t > 0$,
            \[
            \Pr\!\left[|(\mat{E}_{W}\mat{E}_{X})_{jk}| > t\right]
            \leq 2\exp{-t^{2}/(4n\sigma^{4})}
            + e^{-\Omega(n)}
            \]
    \end{enumerate}
\end{lemma}
\begin{proof}
    Fix indices $j,k$ and write
    $(\mat{E}_{W}\mat{E}_{X})_{jk} = \sum_{i=1}^{n} a_{i} b_{i}$,
    where $a_{i} = (\mat{E}_{W})_{ji}$ and
    $b_{i} = (\mat{E}_{X})_{ik}$ are mutually independent
    samples from $\normal{\rng;0}{\sigma^{2}}$.

    \emph{(1)}
    By independence and $\expect{a_i} = \expect{b_i} = 0$,
    $\expect{\sum_{i} a_i b_i}
    = \sum_{i} \expect{a_i}\expect{b_i} = 0$.

    \emph{(2)}
    Condition on $\vec{b} = (b_1, \ldots, b_n)$.
    Given $\vec{b}$, the sum $\sum_{i} a_i b_i$ is a
    linear combination of independent $\sigma$-sub-Gaussian
    variables with fixed coefficients.
    By~\cref{lem:discrete-noise-inner-product},
    the conditional tail satisfies
    \[
        \Pr\!\left[\left|\sum_{i} a_i b_i\right| > t
        \;\middle|\; \vec{b}\right]
        \leq 2\exp{-\frac{t^{2}}{2\sigma^{2}
        \norm{\vec{b}}_{2}^{2}}}
    \]
    Since each $b_i$ is $\sigma$-sub-Gaussian,
    $b_i^{2}$ is sub-exponential with
    $\expect{b_i^{2}} = \sigma_{d}^{2} \leq \sigma^{2}$.
    Standard sub-exponential concentration
    (Bernstein inequality) yields
    $\Pr[\norm{\vec{b}}_{2}^{2} > 2n\sigma^{2}]
    \leq e^{-\Omega(n)}$.
    On the high-probability event
    $\{\norm{\vec{b}}_{2}^{2} \leq 2n\sigma^{2}\}$,
    the conditional bound becomes
    $2 \exp{-t^{2}/(4n\sigma^{4})}$ 
    after substitution,
    and the total probability law gives the result.
\end{proof}

\begin{corollary}
\label{cor:cross-term-negligible}
    The cross term $\mat{E}_{W}\mat{E}_{X}$ is dominated by the
    linear error terms $\mat{E}_{W}\mat{X}$ and $\mat{W}\mat{E}_{X}$.
    Comparing the tail bounds from
    \Cref{lem:discrete-noise-inner-product} and \Cref{lem:cross-term}:
    the linear terms decay as
    $\exp{-t^{2}/(2\sigma^{2}\norm{\vec{x}}_{2}^{2})}$,
    while the cross term decays as
    $\exp{-t^{2}/(4n\sigma^{4})}$.
    With $\norm{\vec{x}}_{2}^{2} = n\overline{x^{2}}$
    (average squared entry $\overline{x^{2}}$),
    the linear exponent scales as
    $t^{2}/(n\sigma^{2}\overline{x^{2}})$ versus
    $t^{2}/(n\sigma^{4})$ for the cross term.
    The cross term is therefore negligible whenever
    $\sigma^{2} \ll \overline{x^{2}}$, a condition required by the
    protocol for useful signal-to-noise ratio.
\end{corollary}

\subsection{Proof of \Cref{thm:total-error-bound}}
\label{prf:total-error-bound}
\thmTotalErrorBound*
\begin{proof}
    The linear terms $(\mat{E}_{W}\mat{X})_{jk}$ and
    $(\mat{W}\mat{E}_{X})_{jk}$ are independent, since they depend
    on the independent noise matrices $\mat{E}_{W}$ and $\mat{E}_{X}$
    respectively.
    By~\cref{lem:discrete-noise-inner-product},
    $(\mat{E}_{W}\mat{X})_{jk}$ is $\sigma\norm{\vec{x}_{k}}_{2}$-sub-Gaussian
    and $(\mat{W}\mat{E}_{X})_{jk}$ is $\sigma\norm{\vec{w}_{j}}_{2}$-sub-Gaussian.
    For independent sub-Gaussian variables, the sub-Gaussian parameters
    add in quadrature, so their sum is
    $\sigma\sqrt{\norm{\vec{x}_{k}}_{2}^{2} + \norm{\vec{w}_{j}}_{2}^{2}}$-sub-Gaussian.
    By~\cref{cor:cross-term-negligible}, the cross term
    $\mat{E}_{W}\mat{E}_{X}$ is negligible whenever
    $\sigma^{2} \ll \overline{x^{2}}$, contributing only
    an additive $e^{-\Omega(n)}$ to the tail probability.
\end{proof}

\subsection{Proof of \Cref{lma:lwe-protocol}}
\label{prf:lwe-protocol}
\lemLweProtocol*
\begin{proof}[Proof sketch]
    We show $\mat{L}\mat{M} + \mat{E}_{w} \overset{c}{\approx} \mat{U}$
    column by column. Interpret the $j$'th column of
    $\mat{L}\mat{M} + \mat{E}_{w}$ as
    $\mat{L}\vec{m}_{j} + \vec{e}_{w,j} \in \fpVec{m}$ and let
    \[
        H_{0} = \mat{L}\mat{M} + \mat{E}_{w}
              = (\mat{L}\vec{m}_1+\vec{e}_{w,1}, \ldots,
                 \mat{L}\vec{m}_{n}+\vec{e}_{w,n}).
    \]
    For $n \geq i > 0$, define
    $H_{i} = (\vec{u}_1, \ldots, \vec{u}_{i},
              \mat{L}\vec{m}_{i+1}+\vec{e}_{w,i+1}, \ldots,
              \mat{L}\vec{m}_n+\vec{e}_{w,n})$,
    where $\vec{u}_{1}, \ldots, \vec{u}_{i}$ are independent uniform
    samples over $\fpVec{m}$. The two consecutive hybrids
    \begin{itemize}
        \item $H_{i}$, and
        \item $H_{i+1}$, which replaces the $(i+1)$-th entry
              $\mat{L}\vec{m}_{i+1}+\vec{e}_{w,i+1}$ with a fresh uniform
              sample $\vec{u}_{i+1}$,
    \end{itemize}
    differ only in that single coordinate; distinguishing them is
    exactly a decisional-LWE distinguisher for the public matrix
    $\mat{L}$ and noise $\vec{e}_{w,i+1}$.

    By the hybrid argument,
    $H_0 \overset{c}{\approx} H_{n}
          \overset{\text{s}}{\approx} (\vec{u}_1, \ldots, \vec{u}_n)
          \overset{\text{s}}{\approx} \mat{U}$,
    with total advantage bounded by $n$ times the decisional-LWE
    advantage. Hence $\mat{L}\mat{M} + \mat{E}_{w} \overset{c}{\approx} \mat{U}$.
\end{proof}

\subsection{Nested LWE+LPN mask}
\label{prf:lwe-lpn-protocol}
\begin{lemma}
    \label{lma:lwe-lpn-protocol}
    The nested LWE+LPN mask is computationally indistinguishable from uniform:
    \[
        (((\mat{L}_{d}\mat{M}_{d} + \mat{S}_d)\mat{M}_{d-1} + \mat{S}_{d-1})\cdots \mat{S}_{2})\mat{M}_{1} + \mat{E}_{w}
        \overset{c}{\approx} \mat{U}.
    \]
\end{lemma}
\begin{proof}[Proof sketch]
    Let $\mat{F}_{d}$ denote the nested mask of the statement and
    $\mat{F}_{1} = \mat{L}_{1}\mat{M}_{1} + \mat{E}_{w}$ the simplified
    LWE mask. We show
    $\mat{F}_{d} \overset{c}{\approx} \mat{U}$ by a hybrid argument
    that unwinds the nested LPN instances one level at a time until we
    reach $\mat{F}_{1}$, then invokes \Cref{lma:lwe-protocol}.

    Define $d$ hybrids $\mat{F}_{d}, \mat{F}_{d-1}, \ldots, \mat{F}_{1}$,
    where $\mat{F}_{i}$ is obtained from $\mat{F}_{i+1}$ by replacing
    the innermost LPN instance
    $\mat{L}_{i+1}\mat{M}_{i+1} + \mat{S}_{i+1}$ with a freshly sampled
    uniform matrix $\mat{L}_{i}$ of the same dimensions. Here
    $\mat{L}_{i+1} \leftarrow \uniform{\rng^{r_{i} \times r_{i+1}}}$,
    $\mat{M}_{i+1} \leftarrow \uniform{\rng^{r_{i+1} \times r_{i-1}}}$
    and $\mat{S}_{i+1}$ has row Hamming weight $t_{i+1}$ (with
    $r_0 := n$).

    Fix a transition $\mat{F}_{i+1} \to \mat{F}_{i}$. Distinguishing
    these two distributions reduces to distinguishing, at the position
    where $\mat{L}_{i}$ appears,
    \begin{itemize}
        \item $\mat{L}_{i+1}\mat{M}_{i+1} + \mat{S}_{i+1}$, a decisional
              LPN sample with secret dimension $r_{i+1}$, sample count
              $r_{i-1}$, and noise Hamming weight $t_{i+1}$, from
        \item uniform
              $\mat{L}_{i} \leftarrow \uniform{\rng^{r_{i-1} \times r_{i}}}$.
    \end{itemize}
    The reduction receives the LPN challenge $\mat{L}^{*}$, embeds it
    at the $i$th position of the mask, samples all remaining
    $\mat{L}_{j}, \mat{M}_{j}, \mat{S}_{j}$ and $\mat{E}_{w}$ honestly,
    and forwards the constructed mask to the distinguisher. Its
    advantage transfers without loss, so
    $\mat{F}_{i+1} \overset{c}{\approx} \mat{F}_{i}$.

    Chaining the $d-1$ LPN hybrid steps gives
    $\mat{F}_{d} \overset{c}{\approx} \mat{F}_{1}$ with distinguishing
    advantage at most $(d-1)$ times the best decisional-LPN advantage.
    By \Cref{lma:lwe-protocol}, $\mat{F}_{1} \overset{c}{\approx} \mat{U}$.
    Composing yields $\mat{F}_{d} \overset{c}{\approx} \mat{U}$.
\end{proof}

\section{Computational integrity}
\label{sec:freiwald}
We refer to a well-known randomized check (Freivalds' algorithm) to ensure
correctness of $\mat{W}\mat{X} = \mat{Y}$ computed by the GPU at runtime,
at the cost of additional input-independent precomputation.
For verification, the client samples a random vector
$(\vec{a})$ and performs the following assertion;
\[
\vec{a}^{T}\mat{W}\mat{X} =^{?} \vec{a}^{T}\mat{Y}
\]
Observe that $\vec{a}^{T}\mat{W}$ is input independent and can be preprocessed and reused for each inference run;
at runtime, the assertion above can be evaluated in $O((n+m)l)$.

To violate soundness, the cheating prover
must return $\mat{Y}$ such that the following holds
\[
\vec{a}^{T}\mat{W}\mat{X} = \vec{a}^{T}\mat{Y}
        \iff \vec{a}^{T}(\mat{W}\mat{X} - \mat{Y}) = 0
\]
where
$\mat{W}\mat{X} - \mat{Y} \not= 0$.
This occurs with probability $1/|\mathbb{F}|$ (Schwartz-Zippel Lemma).
\section{Theoretical outsourcing efficiency}
\begin{figure}[h]
  \centering
  \begin{tikzpicture}
\begin{axis}[
    width=\linewidth,
    height=7cm,
    title={Theoretical efficiency of outsourcing ($l=1$)},
    xlabel={Rank of Model Weight Matrix $n$ ($\times 10^3$)},
    ylabel={Trusted $\cli$ Work:\\outsourced $/$ local $\mat{W}\mat{X}$},
    ylabel style={font=\small, align=center},
    xmin=1, xmax=210,
    ymode=log,
    ymin=0.005, ymax=3.00,
    grid=both,
    minor grid style={gray!20},
    legend style={
        at={(0.97,0.97)},
        anchor=north east,
        font=\small,
    },
    xticklabel style={/pgf/number format/1000 sep={}},
    every axis plot/.append style={thick},
]


\addplot[no markers, color=purple, solid]
    coordinates {
        (2, 2.1500)
        (5, 0.8600)
        (10, 0.4300)
        (20, 0.2150)
        (30, 0.1433)
        (40, 0.1075)
        (50, 0.0860)
        (60, 0.0717)
        (70, 0.0614)
        (80, 0.0538)
        (90, 0.0478)
        (100, 0.0430)
        (120, 0.0358)
        (140, 0.0307)
        (160, 0.0269)
        (180, 0.0239)
        (200, 0.0215)
    };
\addlegendentry{$m = n/4$}
\addplot[only marks, mark=*, color=purple, fill=purple, forget plot]
    coordinates {
        (2, 2.1500)
        (10, 0.4300)
        (20, 0.2150)
        (50, 0.0860)
        (100, 0.0430)
        (200, 0.0215)
    };

\addplot[no markers, color=blue, solid]
    coordinates {
        (2, 0.8600)
        (5, 0.3440)
        (10, 0.1720)
        (20, 0.0860)
        (30, 0.0573)
        (40, 0.0430)
        (50, 0.0344)
        (60, 0.0287)
        (70, 0.0246)
        (80, 0.0215)
        (90, 0.0191)
        (100, 0.0172)
        (120, 0.0143)
        (140, 0.0123)
        (160, 0.0108)
        (180, 0.0096)
        (200, 0.0086)
    };
\addlegendentry{$m = n$}
\addplot[only marks, mark=square*, color=blue, fill=blue, forget plot]
    coordinates {
        (2, 0.8600)
        (10, 0.1720)
        (20, 0.0860)
        (50, 0.0344)
        (100, 0.0172)
        (200, 0.0086)
    };

\addplot[no markers, color=orange, solid]
    coordinates {
        (2, 0.5375)
        (5, 0.2150)
        (10, 0.1075)
        (20, 0.0538)
        (30, 0.0358)
        (40, 0.0269)
        (50, 0.0215)
        (60, 0.0179)
        (70, 0.0154)
        (80, 0.0134)
        (90, 0.0119)
        (100, 0.0108)
        (120, 0.0090)
        (140, 0.0077)
        (160, 0.0067)
        (180, 0.0060)
        (200, 0.0054)
    };
\addlegendentry{$m = 4n$}
\addplot[only marks, mark=diamond*, color=orange, fill=orange, forget plot]
    coordinates {
        (2, 0.5375)
        (10, 0.1075)
        (20, 0.0538)
        (50, 0.0215)
        (100, 0.0108)
        (200, 0.0054)
    };

\end{axis}
\end{tikzpicture}
  \caption{Theoretical trusted-client work vs.\ local
    $\mat{W}\mat{X}$ work ($l=1$; X as a vector) with
    protocol parameterization in \Cref{ex:lwe-lpn-parameterisation}.}
  \label{fig:theoretical-overhead}
\end{figure}
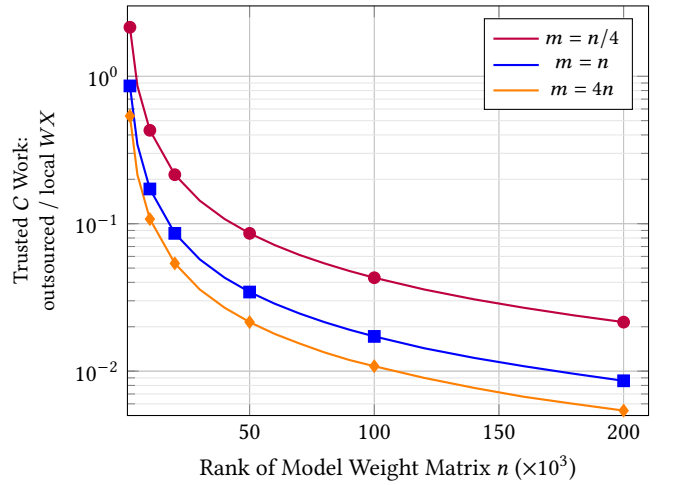

For our proposed protocol, we show the ratio of (1) trusted client protocol overhead to (2) the
cost of the client computing $WX$ locally in \Cref{fig:theoretical-overhead}.
We consider multiplicative operations for (1) and (2).
\section{Protocol error after rotation}
\label{sec:error-post-rot}

\paragraph{Reduction of the infinity norm}
For any fixed vector $\vec{x} \in \reals^{n}$, each entry of
$\mat{R}\vec{x}$ is a normalized sum of $n$ random-sign terms.
By sub-Gaussian concentration~\cite{ailon2006approximate}, 
with high probability,
\begin{equation}
\label{eq:had-inf-norm-bound}
    \norm{\mat{R}\vec{x}}_{\infty}
    \;\lesssim\;
    \norm{\vec{x}}_{2} \cdot \sqrt{\frac{\log n}{n}}
\end{equation}
In contrast, without rotation $\norm{\vec{x}}_{\infty}$ can be as large as
$\norm{\vec{x}}_{2}$ when a single entry dominates.
Therefore the quantization scale after rotation satisfies
\begin{equation}
\label{eq:had-scale-bound}
    s'_{x}
    = \frac{\norm{\mat{R}\vec{x}}_{\infty}}{q_{\max}}
    \;\lesssim\;
    \frac{\norm{\vec{x}}_{2}}{q_{\max}} \cdot \sqrt{\frac{\log n}{n}}
\end{equation}

\paragraph{Float-domain noise variance.}
Consider the dominant error term $\mat{E}_{w}\mat{X}$
from~\cref{eq:approx-error} in the
floating-point domain.
Entry $(j,k)$ of $\mat{E}_{w}\hat{\mat{X}}$ in the integer domain
is the inner product $\sum_{i} (E_{w})_{ji} \hat{x}_{ik}$,
where each $(E_{w})_{ji}$ is an independent Gaussian with variance $\sigma^{2}$.
By~\cref{lem:discrete-noise-inner-product}, the variance of this entry is
$\sigma^{2} \norm{\hat{\vec{x}}_{k}}_{2}^{2}$,
where $\hat{\vec{x}}_{k}$ is the quantized activation column.
Dequantization rescales the integer result by both scales
$s_{w_j} \cdot s_{x_k}$,
giving a full (floating-point) variance of
\begin{equation}
\label{eq:had-var-full}
    \mathrm{Var}(\mat{E}_{w}\mat{X})_{j,k}
    \;=\;
    \sigma^{2} \cdot \norm{\hat{\vec{x}}_{k}}_{2}^{2}
    \cdot s_{w_j}^{2} \cdot s_{x_k}^{2}
\end{equation}
Note that
$\hat{\vec{x}}_{k} = \mathrm{round}(\vec{x}_{k} / s_{x_k})$,
so 
$\norm{\hat{\vec{x}}_{k}}_{2}^{2}
\;\approx\; \norm{\vec{x}_{k}}_{2}^{2}/s_{x_k}^{2}$.
Substituting into~\cref{eq:had-var-full},
\begin{equation*}
    \mathrm{Var}(\mat{E}_{w}\mat{X})_{j,k}
    \;\approx\;
    \sigma^{2} \cdot \norm{\vec{x}_{k}}_{2}^{2} \cdot s_{w_j}^{2}.
\end{equation*}
Substituting
$s_{w_j} = \norm{\vec{w}_j}_\infty / q_{\max}$, the dequantized
variance simplifies to
\begin{equation}
\label{eq:had-var-simplified}
    \mathrm{Var}(\mat{E}_{w}\mat{X})_{j,k}
    \;=\;
    \sigma^{2} \cdot \norm{\vec{x}_{k}}_{2}^{2} \cdot
    \frac{\norm{\vec{w}_{j}}_{\infty}^{2}}{q_{\max}^{2}}
\end{equation}
With Hadamard rotation, substituting~\cref{eq:had-inf-norm-bound}
into~\cref{eq:had-var-simplified}, this becomes
\begin{equation}
\label{eq:had-var-rotated}
    \mathrm{Var}(\mat{E}_{w}\mat{R}\mat{X})_{j,k}
    \;\lesssim\;
    \sigma^{2} \cdot  
    \norm{\vec{x}_{k}}_{2}^{2}
    \cdot
    \frac{\norm{\vec{w}_{j}}_{2}^{2}}{q_{\max}^{2}}
    \cdot
    \frac{\log n}{n}
\end{equation}
By symmetry, the noise from $\mat{W}\mat{E}_{x}$ is reduced analogously
via the activation rotation.
The ratio of~\cref{eq:had-var-rotated} to~\cref{eq:had-var-simplified}
gives the variance reduction factor for the $\mat{E}_{w}\mat{X}$ term:
\begin{equation}
\label{eq:had-var-reduction}
    \frac{\mathrm{Var}(\mat{E}_{M}\mat{R}\mat{X})_{j,k}}
         {\mathrm{Var}(\mat{E}_{M}\mat{X})_{j,k}}
    \;\approx\;
    \frac{\norm{\vec{w}_{j}}_{2}^{2} \cdot \log n}
         {\norm{\vec{w}_{j}}_{\infty}^{2} \cdot n}
         = \rho(w_j)^2 \cdot \frac{\log(n)}{n}
\end{equation}
and analogously for $\mat{W}\mat{E}_{x}$.
This ratio equals $\log(n)/n$ in the worst case of a single dominant outlier
($\norm{\vec{w}}_{\infty} \approx \norm{\vec{w}}_{2}$).
For a transformer dimension $n = 4096$,
this yields an approximate $n / \log n \approx 340\times$
reduction in noise variance.

The cross term $\mat{E}_{w}\mat{E}_{x}$
from~\cref{eq:approx-error} is unaffected by the Hadamard
rotation: $\mat{E}_{M}$ and $\mat{E}_{X}$ are sampled in the
integer domain independently of both $\mat{W}$ and $\mat{X}$,
so their product distribution is invariant to the rotation.
By~\cref{cor:cross-term-negligible}, this term remains
negligible relative to the linear error terms both before and after
rotation.
\section{Secure forward-pass complexity}
\label{sec:forward-pass-complexity}
\Cref{tab:fwdpass-complexity} summarises the per-layer and
full-forward-pass complexity for
an $L$-layer transformer with model dimension $d_\textsf{model}$,
intermediate (MLP) dimension
$d_\textsf{ff}\approx 3.5\,d_\textsf{model}$, total sequence
length $s$, and $l$ new tokens per query.
Per-layer GPU work is dominated by the seven weight projections at
$O(d_\textsf{model}^{2}\,l)$, while per-layer trusted-client work is
dominated by per-head attention at
$O(d_\textsf{model}\cdot s\cdot l)$ for any context
$s \geq \log d_\textsf{model}$.

The $L$-layer forward pass and the
prefill ($l = s$) and autoregressive decode ($l = 1$)
cases follow directly. A one-time per-weight Init~W cost
of $O(d_\textsf{model}^{2}\log d_\textsf{model})$ totals
$O(L\,d_\textsf{model}^{2}\log d_\textsf{model})$ across the model
and is amortised across every subsequent forward pass.

\begin{table}
\centering
\footnotesize
\setlength{\tabcolsep}{3pt}
\resizebox{\linewidth}{!}{%
\begin{tabular}{@{}l c c@{}}
\toprule
 & \multicolumn{2}{c}{Asymptotic cost} \\
\cmidrule(lr){2-3}
Component / Regime
 & Client $\cli$
 & GPU $\gpu$ \\
\midrule
\multicolumn{3}{@{}l}{\textit{Per-layer breakdown}} \\
\midrule
7 weight projections ($Q,K,V,O$, gate, up, down)
 & ---
 & $O(d^{2}l)$ \\
Per-MatMul outsourcing overhead ($\times\,7$)
 & $O(d\,l)$
 & --- \\
Fast Walsh--Hadamard rotation (per call)
 & $O(d\,l\log d)$
 & --- \\
Per-head attention (softmax, $QK^{\top}$, $\cdot V$)
 & $O(d\,s\,l)$
 & --- \\
Element-wise non-linearities (RMSNorm, etc.)
 & $O(d\,l)$
 & --- \\
\midrule
\textbf{Per-layer total} ($s \geq \log d$)
 & $\mathbf{O(d\,s\,l)}$
 & $\mathbf{O(d^{2}l)}$ \\
\midrule
\multicolumn{3}{@{}l}{\textit{Full forward pass over $L$ layers}} \\
\midrule
General ($l$ new tokens, seq.\ length $s$)
 & $O(L\,d\,s\,l)$
 & $O(L\,d^{2}l)$ \\
Prefill ($l = s$)
 & $O(L\,d\,s^{2})$
 & $O(L\,d^{2}s)$ \\
Decode ($l = 1$)
 & $O(L\,d\,s)$
 & $O(L\,d^{2})$ \\
Init~W (one-time, amortised)
 & $O(L\,d^{2}\log d)$
 & --- \\
\bottomrule
\end{tabular}%
}
\caption{Asymptotic per-layer and full-forward-pass complexity of
$\Pi_{\textsf{Sec-FwdPass}}$. Notation: $d = d_\textsf{model}$,
$L$ transformer layers, total sequence length $s$, $l$ new tokens
per query. 
}
\label{tab:fwdpass-complexity}
\end{table}

\section{Broken obfuscation security}
\label{sec:obfuscation-security}
We show a formal break in the security reduction of
ArrowCloak~\cite{wang2025game} (USENIX'25). The scheme protects the
privacy of a fine-tuned model $W_{\textsf{vic}}$ derived from a
pre-trained model $W_{\textsf{pre}}$, where the pre-trained model is
explicitly public and static in the adversary's view. The obfuscated
weight is $W_{\textsf{obf}} = W_{\textsf{vic}} \mat{H}$ for a secret
``encryption'' key matrix $\mat{H}$. The authors claim a reduction to
search-LWE hardness and, as part of that reduction, introduce the
following hybrid step (reproducing eq.~12, section~6.2
of~\cite{wang2025game}):
\[
    \mathcal{H}_0:\; Y = W_{\textsf{obf}}\mat{H}^{-1} + (W_{\textsf{pre}} - W_{\textsf{vic}})
    \qquad
    \mathcal{H}_1:\; Y = W_{\textsf{obf}}\mat{H}^{-1} + \mat{E},
\]
where $\mat{E}$ is a fresh discrete-Gaussian LWE sample.
Computational indistinguishability of $\mathcal{H}_0$ and
$\mathcal{H}_1$ is the essential step in their reduction; the authors
support it with empirical evidence that
$W_{\textsf{pre}} - W_{\textsf{vic}}$ is Gaussian-like.

\paragraph{Trivial distinguisher} However, the two hybrids are
trivially distinguishable because $W_{\textsf{obf}}$ and
$W_{\textsf{pre}}$ are part of the adversary's public view per the
ArrowCloak threat model. Define the deterministic distinguisher
$\mathcal{D}(Y) := [\,Y =^{?} W_{\textsf{pre}}\,]$. Then:
\begin{itemize}
    \item In $\mathcal{H}_0$, $Y = W_{\textsf{obf}}\mat{H}^{-1} +
    (W_{\textsf{pre}} - W_{\textsf{vic}}) = W_{\textsf{vic}} +
    W_{\textsf{pre}} - W_{\textsf{vic}} = W_{\textsf{pre}}$
    deterministically, so
    $\Pr[\mathcal{D}(Y) = 1 \mid \mathcal{H}_0] = 1$.
    \item In $\mathcal{H}_1$, $Y = W_{\textsf{vic}} + \mat{E}$ equals
    $W_{\textsf{pre}}$ only when the fresh Gaussian
    $\mat{E}$ lands on the single fixed point
    $W_{\textsf{pre}} - W_{\textsf{vic}}$, which happens with
    probability negligible in $\lambda$. 
\end{itemize}
Thus $\mathcal{D}$ achieves distinguishing advantage
$1 - \mathrm{negl}(\lambda)$, contradicting the indistinguishability
step $\mathcal{H}_0 \overset{c}{\approx} \mathcal{H}_1$ their
reduction relies on.


\ifpreprint
\section{Supplementary deployment figures}
\label{sec:appendix-deployment-figures}

We complement the main body figures with per-query prefill latencies
across different network environments in the remote decisional setting
(\Cref{fig:network-sweep}) for 1/4/7k context sizes.
We show communication volume for a 1k prefill
required for our proposed protocol (\Cref{tab:comm-volume})
applied to 70B models.

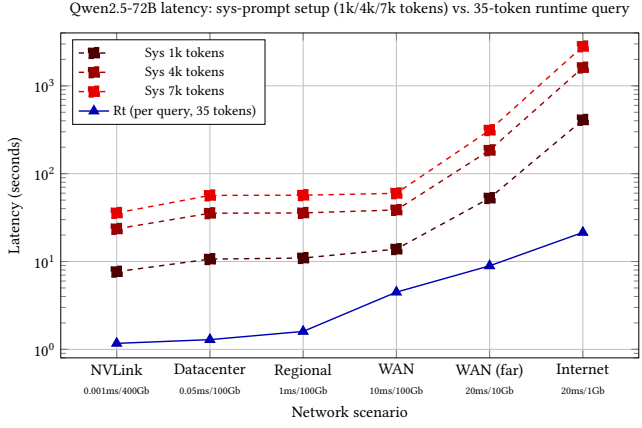
\begin{figure}
\centering
\resizebox{\columnwidth}{!}{\begin{tikzpicture}
\begin{semilogyaxis}[
    width=13cm,
    height=8cm,
    title={Qwen2.5-72B latency: sys-prompt setup (1k/4k/7k tokens) vs.\ 35-token runtime query},
    xlabel={Network scenario},
    ylabel={Latency (seconds)},
    xtick={1, 2, 3, 4, 5, 6},
    xticklabels={
        {NVLink\\{\scriptsize 0.001ms/400Gb}},
        {Datacenter\\{\scriptsize 0.05ms/100Gb}},
        {Regional\\{\scriptsize 1ms/100Gb}},
        {WAN\\{\scriptsize 10ms/100Gb}},
        {WAN (far)\\{\scriptsize 20ms/10Gb}},
        {Internet\\{\scriptsize 20ms/1Gb}}},
    x tick label style={align=center},
    xmin=0.4, xmax=6.6,
    ymin=0.8, ymax=4000,
    grid=major,
    legend style={
        at={(0.02,0.98)},
        anchor=north west,
        font=\small,
    },
    every axis plot/.append style={thick, mark size=3pt},
]


\addplot[color=red!30!black, mark=square*, dashed]
    coordinates {
        (1,   7.66)
        (2,  10.65)
        (3,  10.96)
        (4,  13.84)
        (5,  52.75)
        (6, 409.84)
    };
\addlegendentry{Sys 1k tokens}

\addplot[color=red!60!black, mark=square*, dashed]
    coordinates {
        (1,   23.53)
        (2,   35.45)
        (3,   35.75)
        (4,   38.63)
        (5,  184.67)
        (6, 1613.04)
    };
\addlegendentry{Sys 4k tokens}

\addplot[color=red!90!black, mark=square*, dashed]
    coordinates {
        (1,   35.73)
        (2,   56.58)
        (3,   56.88)
        (4,   59.76)
        (5,  312.93)
        (6, 2812.57)
    };
\addlegendentry{Sys 7k tokens}

\addplot[color=blue!70!black, mark=triangle*]
    coordinates {
        (1,  1.17)
        (2,  1.29)
        (3,  1.60)
        (4,  4.48)
        (5,  8.93)
        (6, 21.42)
    };
\addlegendentry{Rt (per query, 35 tokens)}

\end{semilogyaxis}
\end{tikzpicture}}
\caption{Network sensitivity: system-prompt prefill at
         1k\,/\,4k\,/\,7k tokens (dashed) and per-query 
         runtime prefill (solid) at 35-token 
         across various simulated network scenarios.
         While the (one-time) $7$k-token sys prefill
         approaches $45$ minutes on a $1$\,Gbps Internet link,
         the per-query runtime stays close to $20$\,s.}
\label{fig:network-sweep}
\end{figure}

\begin{table}
\centering
\caption{Communication volume per phase (1k-token system prompt).
         init W is a one-time model setup cost; sys-prefill runs once and is amortized over subsequent runtime prompts.
         rt-prefill is a query of 35 tokens.
         ``Send'' is $\cli$ to $\gpu$; ``Recv'' is $\gpu$ to $\cli$.}
\label{tab:comm-volume}
\resizebox{\columnwidth}{!}{%
\begin{tabular}{l rr rr}
\toprule
 & \multicolumn{2}{c}{\textbf{Llama-3-70B}} & \multicolumn{2}{c}{\textbf{Qwen2.5-72B}} \\
 & Send & Recv & Send & Recv \\
\midrule
init\_W (one-time)         & 278.0\,GB & ---       & 285.9\,GB & ---       \\
sys\_prefill (one-time)    & 17.1\,GB  & 31.6\,GB  & 17.4\,GB  & 32.2\,GB  \\
rt\_prefill (per query)    & 0.60\,GB  & 1.11\,GB  & 0.61\,GB  & 1.13\,GB  \\
\bottomrule
\end{tabular}%
}
\end{table}

\FloatBarrier

\fi
\section{Outsourcing remote classifications}
\label{sec:application-remote}

%
\pgfplotsset{
    remotebarpostccs/.style={
        width=8.0cm, height=4.8cm,
        ybar stacked, bar width=14pt,
        clip=false,
        ymin=0, ymax=110,
        ytick={0,20,40,60,80,100},
        grid=major, ymajorgrids=true, xmajorgrids=false,
        symbolic x coords={initW,{Sys 1k},{Sys 4k},{Sys 7k},Rt},
        xtick=data,
        xticklabels={{init\_W},{Sys 1k},{Sys 4k},{Sys 7k},{Rt (Sys 1/4/7k)}},
        x tick label style={rotate=45, anchor=east, font=\scriptsize},
        enlarge x limits=0.12,
        ylabel={Latency (seconds)},
        every axis plot/.append style={thick},
        legend style={
            at={(0.5, 1.02)}, anchor=south,
            legend columns=3,
            font=\scriptsize, draw=none, fill=none,
            column sep=4pt, row sep=-2pt, inner sep=2pt,
        },
        legend cell align=left,
    }
}

\begin{figure}
\centering
\begin{tikzpicture}
\begin{axis}[remotebarpostccs]
\addplot[fill=blue!70, draw=blue!80!black] coordinates {
    (initW,56.41) ({Sys 1k},2.42) ({Sys 4k},8.14) ({Sys 7k},13.21) (Rt,0.46)
};
\addlegendentry{Trusted GPU}
\addplot[fill=orange!70, draw=orange!80!black] coordinates {
    (initW,0) ({Sys 1k},5.39) ({Sys 4k},15.75) ({Sys 7k},22.94) (Rt,0.63)
};
\addlegendentry{Untrusted GPU}
\addplot[fill=gray!40, draw=gray!70!black, postaction={pattern=north east lines, pattern color=gray!70!black}] coordinates {
    (initW,30.53) ({Sys 1k},8.97) ({Sys 4k},21.11) ({Sys 7k},29.76) (Rt,3.61)
};
\addlegendentry{Comm (10\,ms / 100\,Gbps)}

\node[rotate=90, anchor=west, font=\scriptsize\bfseries, inner sep=1pt] at (axis cs:initW,86.94) {87};
\node[rotate=90, anchor=west, font=\scriptsize\bfseries, inner sep=1pt] at (axis cs:{Sys 1k},16.77) {17};
\node[rotate=90, anchor=west, font=\scriptsize\bfseries, inner sep=1pt] at (axis cs:{Sys 4k},45.00) {45};
\node[rotate=90, anchor=west, font=\scriptsize\bfseries, inner sep=1pt] at (axis cs:{Sys 7k},65.90) {66};
\node[rotate=90, anchor=west, font=\scriptsize\bfseries, inner sep=1pt] at (axis cs:Rt,4.70) {4.7};
\end{axis}
\end{tikzpicture}
\caption{\revised{System prefill and runtime latency under simulated WAN (10\,ms RTT, 100\,Gbps,
         $G=3$ untrusted GPUs) for Qwen2.5-72B.}}
\label{fig:decode-latency-70b-qwen-postccs}
\end{figure}
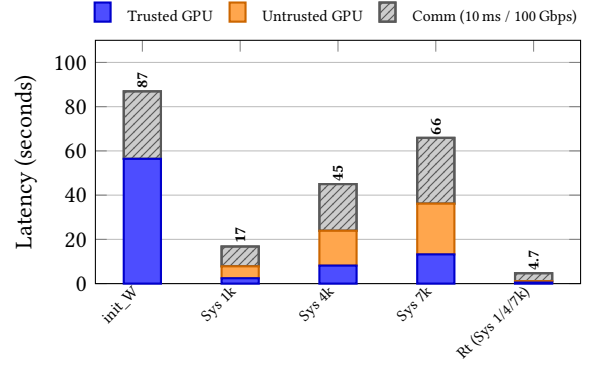

%
\begin{figure}
\centering
\begin{tikzpicture}
\begin{axis}[remotebarpostccs]
\addplot[fill=blue!70, draw=blue!80!black] coordinates {
    (initW,55.02) ({Sys 1k},2.11) ({Sys 4k},7.53) ({Sys 7k},12.28) (Rt,0.44)
};
\addlegendentry{Trusted GPU}
\addplot[fill=orange!70, draw=orange!80!black] coordinates {
    (initW,0) ({Sys 1k},4.85) ({Sys 4k},14.86) ({Sys 7k},22.20) (Rt,0.62)
};
\addlegendentry{Untrusted GPU}
\addplot[fill=gray!40, draw=gray!70!black, postaction={pattern=north east lines, pattern color=gray!70!black}] coordinates {
    (initW,29.83) ({Sys 1k},8.57) ({Sys 4k},20.16) ({Sys 7k},28.57) (Rt,3.59)
};
\addlegendentry{Comm (10\,ms / 100\,Gbps)}

\node[rotate=90, anchor=west, font=\scriptsize\bfseries, inner sep=1pt] at (axis cs:initW,84.84) {85};
\node[rotate=90, anchor=west, font=\scriptsize\bfseries, inner sep=1pt] at (axis cs:{Sys 1k},15.54) {16};
\node[rotate=90, anchor=west, font=\scriptsize\bfseries, inner sep=1pt] at (axis cs:{Sys 4k},42.55) {43};
\node[rotate=90, anchor=west, font=\scriptsize\bfseries, inner sep=1pt] at (axis cs:{Sys 7k},63.06) {63};
\node[rotate=90, anchor=west, font=\scriptsize\bfseries, inner sep=1pt] at (axis cs:Rt,4.65) {4.7};
\end{axis}
\end{tikzpicture}
\caption{\revised{System prefill and runtime latency under simulated WAN (10\,ms RTT, 100\,Gbps, $G=3$ untrusted GPUs) for LLaMA-3-70B.}}
\label{fig:decode-latency-70b-llama3-postccs}
\end{figure}
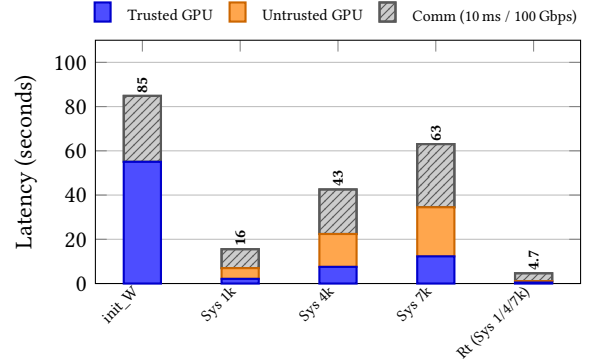

We also consider a case-study of an application that is not latency critical,
and thus permits interactive outsourcing of computation across public networks,
between a local trusted client and a remote cloud accelerator
(\Cref{fig:remote-topology}).
In addition to emulating the ring integer arithmetic
(as mentioned at the beginning of this section)
to measure client and remote GPU computation runtimes,
we simulate network latency induced by our protocol for different
network settings (\Cref{fig:decode-latency-70b-qwen-postccs,fig:decode-latency-70b-llama3-postccs}).

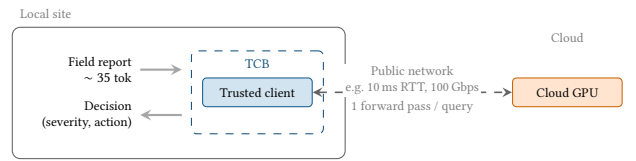
\begin{figure}[h]
\centering
\begin{tikzpicture}[
    >=stealth,
    gpu/.style={draw, rounded corners=1.5pt, minimum width=1.45cm,
        minimum height=0.38cm, font=\tiny, align=center, inner sep=1pt},
    trustedgpu/.style={gpu, fill=cTrusted!20, draw=cTrusted!80!black},
    untrustedgpu/.style={gpu, fill=cUntrusted!20, draw=cUntrusted!80!black},
    netlink/.style={<->, semithick, draw=cComm!70!black, dashed},
]

\node[trustedgpu] (T) at (0, 0) {Trusted client};

\node[untrustedgpu] (U) at (4.10, 0) {Cloud GPU};

\draw[netlink] (T.east) -- (U.west);
\node[font=\tiny, gray, fill=white, inner sep=0.8pt, align=center]
    at ($(T.east)!0.5!(U.west) + (0,0.16)$)
    {Public network\\e.g. 10\,ms RTT, 100\,Gbps};
\node[font=\tiny, cComm!90!black, fill=white, inner sep=0.8pt]
    at ($(T.east)!0.5!(U.west) + (0,-0.20)$)
    {1 forward pass / query};

\node[font=\tiny, anchor=east, align=right] (rep) at (-1.55, 0.30)
    {Field report\\$\sim 35$ tok};
\draw[->, thick, gray!70] (rep.east) -- ++(0.55, 0);

\node[font=\tiny, anchor=east, align=right] (dec) at (-1.55,-0.30)
    {Decision\\(severity, action)};
\draw[<-, thick, gray!70] (dec.east) -- ++(0.55, 0);

\begin{scope}[on background layer]
\node[draw=cTrusted!70!black, dashed, rounded corners=2pt,
      fit=(T), inner xsep=4pt, inner ysep=10pt,
      label={[font=\tiny, cTrusted!70!black, anchor=north]north:TCB}]
      (tcb) {};
\end{scope}

\begin{scope}[on background layer]
\node[draw=cComm!70!black, rounded corners=3pt, thin,
      fit=(tcb)(rep)(dec), inner xsep=8pt, inner ysep=8pt,
      label={[font=\tiny, cComm!90!black, anchor=south west]%
             north west:Local site}]
      (loc) {};
\end{scope}

\node[font=\tiny, cComm!90!black, anchor=south] at (4.10, 0.55) {Cloud};

\end{tikzpicture}
\caption{Remote outsourcing topology: a trusted client at the local site
         outsources a single forward pass per query to an untrusted cloud
         GPU over a public network (e.g.\ 10\,ms RTT, 100\,Gbps).}
\label{fig:remote-topology}
\end{figure}

\paragraph{Industrial diagnostics}
Consider an industrial predictive maintenance scenario in which a company
operates a fleet of equipment---pumps, compressors, motors, gearboxes---and
wishes to outsource diagnostic reasoning to a large language model hosted on
an untrusted cloud accelerator.
The \emph{system prompt} encodes the full diagnostic knowledge base:
equipment-specific failure mode signatures, vibration and temperature
thresholds, oil-analysis limits, multi-symptom escalation rules, and
historical case studies.
This domain context is lengthy (1k--7k tokens in our benchmarks) but
changes infrequently---at most when maintenance procedures are revised.
The \emph{runtime prompt}, by contrast, is a short field report
(${\sim}$35 tokens) written by a technician on the factory floor:
a handful of sensor readings and a one-line observation such as
\emph{``Coupling very hot, faint burnt smell from motor end.''}
The model's task is to return a bounded decision, e.g. a severity
classification and recommended action, not an open-ended generation.

We choose this scenario precisely because it is not latency-critical:
a maintenance classification that completes in seconds is well within
operational requirements, unlike interactive chat or real-time control.
Importantly, the task requires only a single forward pass at runtime over a
bounded number of tokens, avoiding the compounding communication
cost of autoregressive decode rounds.
The expensive yet stable system-prompt setup  
is performed once and amortized over many lightweight
runtime queries.

\paragraph{Evaluation results}
\Cref{fig:decode-latency-70b-qwen-postccs,fig:decode-latency-70b-llama3-postccs} show that latency for decisional
inference can remain practical. 
Despite $\approx$ 80 $\times$ 4 communication rounds,
a single decisional inference query can be completed in under
5 seconds over a fast public network (10ms, 100Gbps).
One can observe greatly reduced transfer volumes for the short 
runtime prompt compared to the system prompt,
shown in~\Cref{tab:comm-volume}.
Overall run-time is clearly affected by
communication latency;
still for the short, real-time diagnostic prompt, 
the system prompt \emph{length} does not have a measurable effect. 
The system prompt prefill of 
7k tokens can be completed in about a minute, implying 
that frequent background, prefill updates are permissible.
Simulated runtimes for a wider range of realistic network settings
are detailed in~\Cref{fig:network-sweep}.

\section{Layer-by-layer error accumulation}

We supplement the main body with 
error accumulation metrics 
on final model layers in~\Cref{tab:final-layer-amplitude} and error accumulation (cosine relative $\ell_2$ error) 
across layers for Qwen2.5-32B (\Cref{fig:layer-cosine-sim-qwen32b}), Qwen2.5-72B (\Cref{fig:layer-cosine-sim-qwen72b}), LLaMA3-70B (\Cref{fig:layer-cosine-sim-llama3-70b}), DeepSeek-R1-Distill-LLaMA-70B (\Cref{fig:layer-cosine-sim-deepseek-r1-70b}).

\begin{table}[t]
\centering
\resizebox{\columnwidth}{!}{%
\begin{tabular}{l c c c c c c}
\toprule
 & \multicolumn{2}{c}{Noisy vs.\ ref} & \multicolumn{4}{c}{BF16 reference} \\
\cmidrule(lr){2-3} \cmidrule(lr){4-7}
Step & cos & rel-$\ell_2$ & $\max|\vec{x}|$ & $\sigma_{\vec{x}}$ & $\overline{|\vec{x}|}$ & $\max|\vec{x}|/\sigma_{\vec{x}}$ \\
\midrule
\multicolumn{7}{l}{\textit{Quantization-friendly models}} \\
\midrule
\multicolumn{7}{l}{\emph{Qwen2.5-32B}, protocol noise $\sigma=1.0$} \\
$\ell = 62$            & 0.9991 & 0.0391 & 5480 & 23.4 & 11.98 & 236 \\
$\ell = 63$            & 0.9984 & 0.0532 & 4340 & 24.1 & 13.95 & 177 \\
post-RMSNorm           & 0.9987 & 0.0448 &  173 &  2.86 & 1.37 &  61 \\
\midrule
\multicolumn{7}{l}{\emph{Qwen2.5-72B}, protocol noise $\sigma=1.0$} \\
$\ell = 78$            & 0.9960 & 0.0724 & 1602 & 14.2 & 7.75  & 116 \\
$\ell = 79$            & 0.9941 & 0.0938 & 2850 & 15.7 & 8.90  & 183 \\
post-RMSNorm           & 0.9949 & 0.0812 &  335 &  3.6 & 1.67  &  93 \\
\midrule
\multicolumn{7}{l}{\textit{Quantization-unfriendly models}} \\
\midrule
\multicolumn{7}{l}{\emph{LLaMA-3-70B}, protocol noise $\sigma=1.0$} \\
$\ell = 78$            & 0.9795 & 0.1774 &  233 & 0.66 & 0.46 & 356 \\
$\ell = 79$            & 0.9813 & 0.1771 &   86 & 0.88 & 0.60 &  95 \\
post-RMSNorm           & 0.9822 & 0.1654 &   94 & 2.03 & 1.42 &  46 \\
\midrule
\multicolumn{7}{l}{\emph{DeepSeek-R1-Distill-LLaMA-70B}, protocol noise $\sigma=1.0$} \\
$\ell = 78$            & 0.8289 & 0.5464 &  400 & 0.70 & 0.45 & 572 \\
$\ell = 79$            & 0.8387 & 0.5399 &   60 & 0.77 & 0.55 &  78 \\
post-RMSNorm           & 0.8466 & 0.5098 &   87 & 2.00 & 1.46 &  44 \\
\bottomrule
\end{tabular}%
}
\caption{Reference (BF16, noise-free) residual-stream amplitude and
matched perturbed-pass metrics at the last two transformer
model layers and after the final RMSNorm.}
\label{tab:final-layer-amplitude}
\end{table}

\clearpage
\begin{figure*}
  \centering
  \input{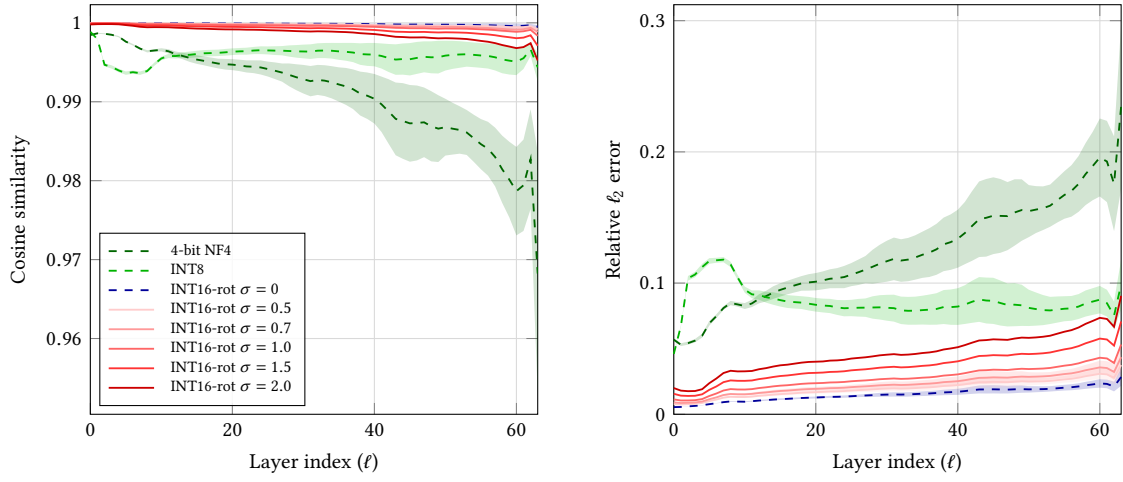}
  \caption{Per-layer error accumulation: Qwen2.5-32B (64 model layers). For 20 recorded runs, the 1-standard deviation band is shown for selected quantization schemes.}
  \label{fig:layer-cosine-sim-qwen32b}
\end{figure*}

\begin{figure*}
  \centering
  \input{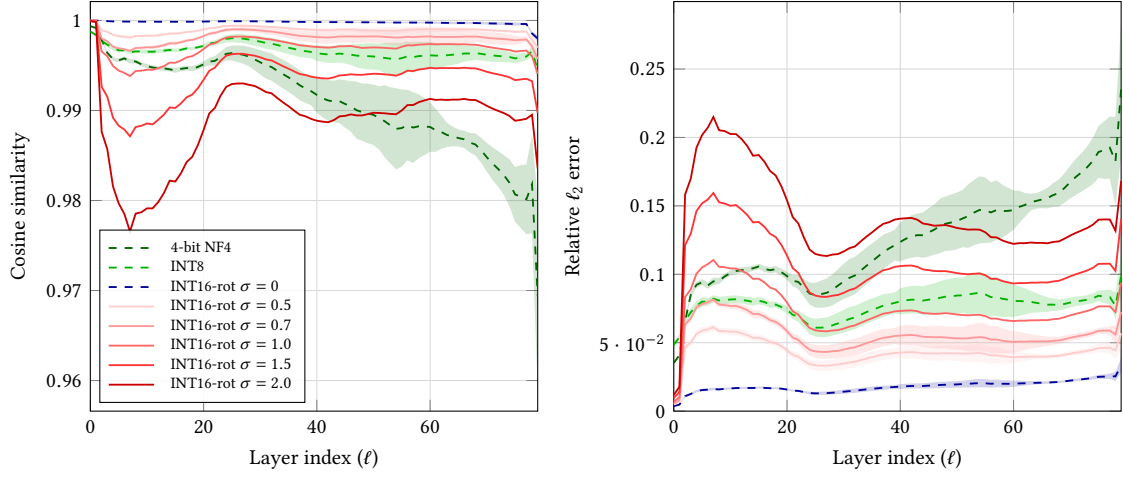}
  \caption{Per-layer error accumulation: Qwen2.5-72B (80 model layers). }
  \label{fig:layer-cosine-sim-qwen72b}
\end{figure*}

\begin{figure*}
  \centering
  \input{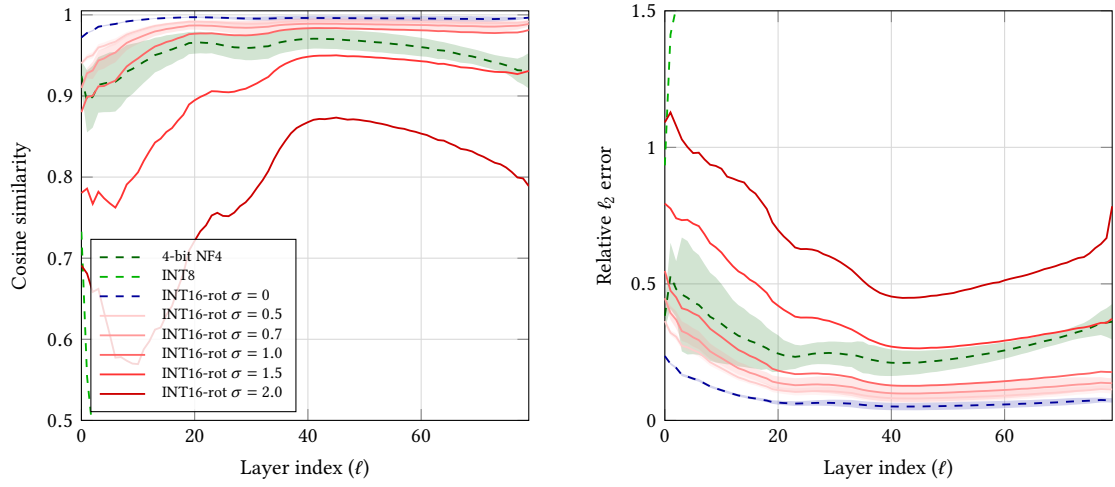}
  \caption{Per-layer error accumulation: LLaMA-3-70B (80 layers).
  INT8 baseline collapses at $\ell=79$ and goes off-scale.}
  \label{fig:layer-cosine-sim-llama3-70b}
\end{figure*}

\begin{figure*}
  \centering
  \input{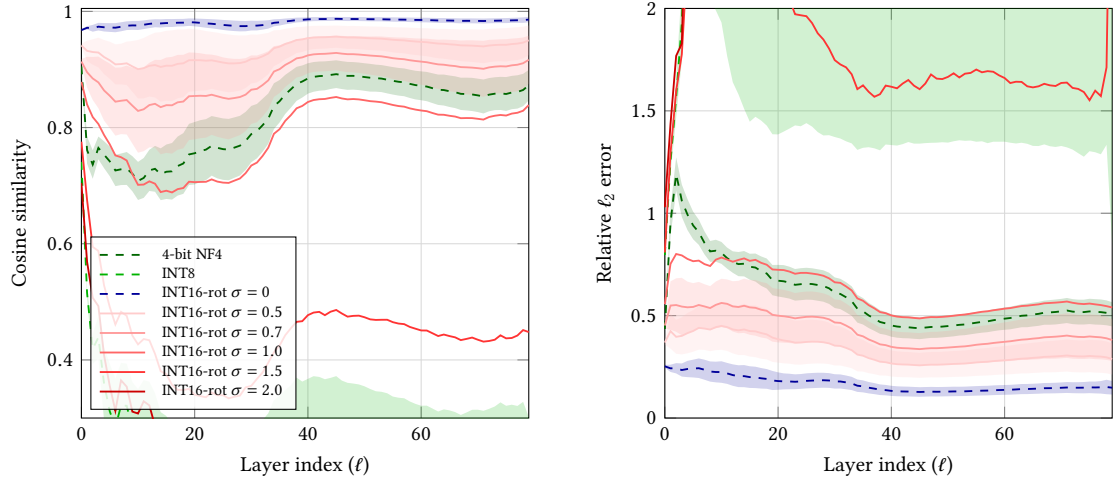}
  \caption{Per-layer error accumulation: DeepSeek-R1-Distill-LLaMA-70B (80 layers).}
  \label{fig:layer-cosine-sim-deepseek-r1-70b}
\end{figure*}

\end{document}